\documentclass{article}
\usepackage{titling}
\usepackage{datetime}
\usepackage{float}
\usepackage{multirow}
\usepackage{makecell}
\usepackage{array}
\usepackage{algorithm}
\usepackage{algpseudocode}
\usepackage{tcolorbox}

\usepackage{xcolor}
\usepackage{hyperref}
\hypersetup{
    colorlinks,
    linkcolor={red!50!black},
    citecolor={blue!50!black},
    urlcolor={blue!80!black}
}

\usepackage{tikz}
\usetikzlibrary{tikzmark}

\newcommand{\highlight}[4]{
    \tikz[overlay, remember picture]{
        \fill[#2!50, opacity=0.25] 
        ([xshift={#3[0]}, yshift={#3[1]}]pic cs:#1-start) rectangle ([xshift={#4[0]}, yshift={#4[1]}]pic cs:#1-end);
    }
}

\usepackage{enumitem}
\makeatletter
\def\namedlabel#1#2{\begingroup
    #2%
    \def\@currentlabel{#2}%
    \phantomsection\label{#1}\endgroup
}
\makeatother

\usepackage{tcolorbox}
\definecolor{lightgray}{gray}{0.9}
\newtcbox{\highlightbox}[1][lightgray]{on line, colback=#1, colframe=#1, boxrule=0pt, arc=0pt, outer arc=0pt, boxsep=0pt, left=0pt, right=0pt, top=0pt, bottom=0pt}

\usepackage[top=2.5cm,bottom=2.5cm,right=2cm,left=2cm]{geometry}

\usepackage{amsmath, amssymb, mathtools, amsthm}
\usepackage{empheq}
\usepackage[thinc]{esdiff}
\usepackage{accents}

\newtheorem{theorem}{Theorem}[section]

\newtheorem{assumption}{Assumption}[section]

\newtheorem{lemma}{Lemma}[section]
\newtheorem{exmp}{Example}[section]

\DeclareMathOperator*{\argmin}{arg\,min}
\DeclareMathOperator*{\R}{\mathbb{R}}
\DeclareMathOperator*{\Rplus}{\mathbb{R}_{+}}
\DeclareMathOperator*{\Rstarplus}{\mathbb{R}_{+}^{*}}
\DeclareMathOperator*{\Natural}{\mathbb{N}}
\DeclareMathOperator*{\expectation}{\mathbb{E}}
\DeclareMathOperator*{\generator}{\mathbb{A}}

\DeclareMathOperator*{\sumoverjumps}{\sum_{l=0 \atop \sigma_{l}<T}^{\infty}}
\DeclareMathOperator*{\enumreac}{[\![1, M]\!]}

\usepackage{tikz}
\usepackage{bbold}
\usepackage[font=small]{caption}
\usepackage[font=small]{subcaption}
\usepackage[section]{placeins}

\usepackage{algpseudocode}
\usepackage{algorithm}

\allowdisplaybreaks

\usepackage{mdframed}
\newmdenv[
  topline=false,
  bottomline=false,
  rightline=false,
  skipabove=\topsep,
  skipbelow=\topsep,
  leftmargin=2pt,
  rightmargin=0pt,
  innertopmargin=0pt,
  innerbottommargin=0pt
]{sideline}

\usepackage[version=4]{mhchem}
\def\arrvline{\hfil\kern\arraycolsep\vline\kern-\arraycolsep\hfilneg}

\usepackage[sorting=none, style=vancouver]{biblatex}
\bibliography{main}

\AtBeginEnvironment{proof}{\footnotesize \linespread{0.7}}

\title{Unbiased estimation of second-order parameter sensitivities for stochastic reaction networks}
\author{Quentin Badolle$^1$, Ankit Gupta$^1$, Mustafa Khammash$^{1*}$}
\date{\monthname{} \the\year}
\date{
$^1$Department of Biosystems Science and Engineering, ETH Zurich, Basel, Switzerland \\
$^*$ mustafa.khammash@bsse.ethz.ch
}

\begin{document}

\maketitle

\begin{abstract}
Stochastic models for chemical reaction networks are increasingly popular in systems and synthetic biology. These models formulate the reaction dynamics as Continuous-Time Markov Chains (CTMCs) whose propensities are parameterized by a vector $\theta$ and parameter sensitivities are introduced as derivatives of their expected outputs with respect to components of the parameter vector. Sensitivities characterise key properties of the output like robustness and are also at the heart of numerically efficient optimisation routines like Newton-type algorithms used in parameter inference and the design of of control mechanisms. Currently the only unbiased estimator for second-order sensitivities is based on the Girsanov transform and it often suffers from high estimator variance. We develop a novel estimator for second-order sensitivities by first rigorously deriving an integral representation of these sensitivities. We call the resulting method the \emph{Double Bernoulli Path Algorithm} and illustrate its efficiency through numerical examples.
\end{abstract}

Chemical reaction networks provide a mechanistic description of the molecular processes hosted in living cells. These models represent the possible reactions between molecular entities within cells as well as the speed at which these interactions take place. Such a representation has become instrumental in systems and synthetic biology~\cite{del2015biomolecular}. In this context, some molecular species are present in low copy numbers. Stochastic reaction networks~(SRNs) account for the discrete nature of the molecular counts and the resulting inherent stochasticity observed experimentally in the dynamics of these systems~\cite{anderson2015stochastic}. SRNs are a class of Continuous-Time Markov Chains (CTMCs) whose generator is expressed in terms of state-dependent propensities. The probability that a certain reaction takes places during a small time interval is proportional to its corresponding propensity.

In applications, these propensities are usually parameterised by a vector $\theta$. The value of these parameters is frequently uncertain and it is therefore crucial to quantify to which degree the predictions made about the system are going to be affected by a change in parameters. Such predictions include expected outputs of the network, the derivatives of which reflect the magnitude of change due to a local perturbation of the parameters. These derivatives are commonly called~\emph{parameter sensitivities}. In the overwhelming majority of SRNs of practical relevance, sensitivities do not have a known analytical expression and computational methods are the only possibility to estimate them for a given nominal value of the parameters. Accordingly, significant effort has been devoted over the last twenty years to develop efficient estimation procedures in the case of first-order sensitivities~\cite{plyasunov2007efficient,rathinam2010efficient,sheppard2012pathwise,anderson2012efficient,gupta2018estimation}. Like for most stochastic models, these methods essentially fall in three broad categories~\cite{asmussen2007stochastic,mohamed2020monte}: Girsanov Transform (GT) methods~\cite{plyasunov2007efficient}, pathwise methods~\cite{sheppard2012pathwise,wolf2015hybrid} and finite-difference methods~\cite{rathinam2010efficient,anderson2012efficient}. In addition to these now classical methods, a fourth class of methods has recently emerged under the name of \emph{gradient estimator} in the jump-diffusion literature~\cite{wang2024efficient} and \emph{exact Integral Path Algorithm} (eIPA) in the SRN literature~\cite{gupta2018estimation}~(see ref.~\cite{badolle2024generator} for the connection between these two approaches). All these approaches can be further distinguished by their range of applicability, variance, bias, and mean-squared error. Formulas for sensitivities have also proven useful in theoretically characterising the robustness of the network outputs to parameter disturbance and uncertainty~\cite{briat2016antithetic,aoki2019universal,gupta2022universal}. When parameters are unknown, sensitivities can be used as part of gradient-based inference procedures~\cite{boyd2004convex}. So far, higher-order sensitivities have been less extensively investigated although they can provide a refined understanding of the dependency of the network response on the parameters~\cite{wolf2012finite}. Second-order sensitivities can also be used in efficient optimisation routines like the Newton-Raphson algorithm~\cite{boyd2004convex}.

In this paper, we first provide what are to our knowledge the first conditions on the existence of second-order parameter sensitivities for SRNs. We also provide a new representation for these second-order sensitivities in the form of an integral over trajectories. From this, we derive an unbiased estimator of these sensitivities based on samples generated by what we will call the \emph{Double Bernoulli Path Algorithm} abbreviated to Double BPA or DBPA. This new estimator falls in the category as the gradient estimator and eIPA. We demonstrate in several examples that it can provide substantial performance improvement over the only unbiased alternative currently available based on the GT method.

\section{Preliminaries}

\subsection{Stochastic reaction networks} SRNs are introduced in detail in~\cite{anderson2015stochastic,erdi1989mathematical,  wilkinson2018stochastic}. Here, let us consider a network with $N\in\Natural^{*}$ molecular species. The state of the system at any time can be described by a vector in $\Natural^N$ whose $i$-th component corresponds  to the number of molecules of the $i$-th species. The chemical species interact through $M\in\Natural^{*}$ chemical reactions and every time the $k$-th reaction fires, the state of the system is displaced by the $N$-dimensional stoichiometric vector $\zeta_{k}\in\mathbb{Z}^N$. The propensity function $\lambda = (\lambda_{k})_{k\in \enumreac}$ depends on the state of the system $x \in \Natural^N$ and a parameter $\theta \in \R^d$, where $d \in \Natural^{*}$. The reaction dynamics are represented by a Continuous-Time Markov Chain~(CTMC) $(X_{\theta}(t))_{t\in\Rplus}$ whose generator is~\cite{anderson2015stochastic}:
\begin{equation}
\label{eq:def_generator}
 \mathbb{A}_{\theta}f(x) \coloneqq \sum_{k=1}^{M} \lambda_k(x,\theta)\Delta_{\zeta_k}f(x),
\end{equation}

for any bounded, real-valued function $f$ on $\mathbb{N}^N$. Given a collection of independent, unit-rate Poisson processes $\{(Y_k(t))_{t \in \mathbb{R}_{+}}\}_{k \in [\![1, M]\!]}$, we associate to each reaction $k$ a counting process $(R_{k}(t))_{t \in \mathbb{R}_{+}}$ defined as:

\begin{equation}
R_{k}(t) \coloneqq Y_{k}\left(\int_0^t \lambda_k(X_\theta(s), \theta) ds\right).
\end{equation}

For an initial state $x \in \Natural^N$, the random time change representation of $(X_{\theta}(t))$ is given by~\cite{anderson2015stochastic, kurtz1980representations}:
\begin{equation}
\label{eq:rtc}    
X_{\theta}(t) = x + \sum_{k=1}^{M} \zeta_{k} R_{k}(t).
\end{equation}

For a function $f : \Natural^{N} \xrightarrow[]{} \R$ that expresses an output of interest, we define:
\begin{equation}
\Psi_{\theta}(x, f, t) \coloneqq \mathbb{E}[f(X_{\theta}(t)) | X_{\theta}(0)=x] = \mathbb{E}_{x}[f(X_{\theta}(t))].    
\end{equation}

For any $(i,j) \in [\![1, d]\!]^2$, we introduce the first- and second-order parameter sensitivities as~\cite{wolf2012finite, gupta2013unbiased}:
\begin{align}
S_{\theta}^{(i)}(x,f,t) &\coloneqq \frac{\partial \Psi_{\theta}}{\partial \theta_{i}}(x,f,t), \label{eq:def_first_order}\\
S_{\theta}^{(i,j)}(x,f,t) &\coloneqq \frac{\partial^2 \Psi_{\theta}}{\partial \theta_{i} \partial \theta_{j}}(x,f,t).\label{eq:def_second_order}     
\end{align}

\subsection{Estimation of second-order parameter sensitivities}

As mentioned in the introduction, most of the research effort on sensitivities for SRNs has been devoted to first-order derivatives so far. These techniques are reviewed extensively in section 2 of~\cite{gupta2018estimation}. We provide a brief overview focused on second-order sensitivities below.\newline

\textbf{Unbiased estimation.} To our knowledge, the Girsanov Transform~(GT) method in currently the only unbiased method to estimate second-order parameter sensitivities of SRNs. It relies on an estimator $s_{\theta}^{(i,j)}(f, T)$ defined as~\cite{glynn1990likelihood, plyasunov2007efficient}:
\begin{align}
\label{eq:girsanov_transform}
\begin{split}
&s_{\theta}^{(i,j)}(f, T) =  f\big(X_{\theta}(T)\big)\Bigg[\sum\limits_{k=1}^{M}\int_0^{T} \left(\diffp{\log \lambda_k}{{\theta_{i}}{\theta_{j}}}(X_{\theta}(s), \theta) dR_{k}(s) - \diffp{\lambda_k}{{\theta_{i}}{\theta_{j}}}(X_{\theta}(s), \theta) ds\right)  \\
&+ \sum\limits_{k=1}^{M}\int_0^{T}\left(\frac{\partial \log \lambda_k}{\partial{\theta_{i}}}(X_{\theta}(s), \theta) dR_{k}(s) - \frac{\partial \lambda_k}{\partial {\theta_{i}}}(X_{\theta}(s), \theta) ds\right) \sum\limits_{k=1}^{M}\int_0^{T}\left(\frac{\partial \log \lambda_k}{\partial{\theta_{j}}}(X_{\theta}(s), \theta) dR_{k}(s) - \frac{\partial \lambda_k}{\partial {\theta_{j}}}(X_{\theta}(s), \theta) ds\right) \Bigg],
\end{split}
\end{align}

where $X_{\theta}(0) = x$. The applicability of the GT method depends crucially on the existence of the stochastic exponential and the validity of an interchange of derivation and integration~\cite{wang2021validity}. These conditions are challenging to verify and have not yet been derived for second-order sensitivities~\cite{wang2021validity}. From eq.~\eqref{eq:girsanov_transform}, it is clear that the GT estimator is not defined whenever some of the propensities are proportional to $\theta_i$ and/or $\theta_j$ with the parameter being evaluated at zero. This means in particular that second-order sensitivities cannot be computed at zero in the important case of mass action kinetics~\cite{anderson2015stochastic}. This precludes the investigation of the sensitivity of an output to the presence/absence of a reaction. When applicable, the GT method has been shown to suffer from large variance not only for SRNs but also for a range of other stochastic models~\cite{gupta2013unbiased, mohamed2020monte}. Depending on the field of application, it is sometimes refereed to as the likelihood ratio method or the REINFORCE algorithm~\cite{mohamed2020monte}.\newline

\textbf{Biased estimation.} Pick an $\epsilon \in \Rstarplus$ and introduce $\hat{\theta} \coloneqq (\theta_{1}^{\epsilon}, \theta_{2}^{\epsilon}, \theta_{3}^{\epsilon},\theta_{4}^{\epsilon})$ specified by:

\begin{equation}
\theta_{1}^{\epsilon} \coloneqq \theta + (e_i + e_j) \epsilon, \quad \theta_{2}^{\epsilon} \coloneqq \theta + e_i \epsilon, \quad \theta_{3}^{\epsilon} \coloneqq \theta + e_j \epsilon, \quad \theta_{4}^{\epsilon} \coloneqq \theta.   
\end{equation}

Finite difference methods provide a biased estimation of the second-order sensitivity by approximating it as~\cite{wolf2012finite}:
\begin{equation}
\label{eq:second_order_finite_difference}
S_{\theta}^{(i,j)}(x,f,T) \approx  \mathbb{E}_{x} \Bigg[\frac{f\big(X_{\theta_{1}^{\epsilon}}^{(1)}(T)\big)-f\big(X_{\theta_{2}^{\epsilon}}^{(2)}(T)\big)-f\big(X_{\theta_{3}^{\epsilon}}^{(3)}(T)\big)+f\big(X_{\theta_{4}^{\epsilon}}^{(4)}(T)\big)}{\epsilon^{2}}\Bigg],
\end{equation}

where each process indexed by $\ell$ has generator $\mathbb{A}_{\theta_{\ell}^{\epsilon}}$. Eq.~\eqref{eq:second_order_finite_difference} is usually replaced by its centred equivalent, leading to a so called centred finite difference estimator with bias $O(\epsilon^{2})$~\cite{wolf2012finite}. In addition, the four processes are often coupled to reduce the variance of the estimator~\cite{wolf2012finite, asmussen2007stochastic}. One of these coupling schemes is the split coupling with variance $O(\epsilon^{-3})$~\cite{wolf2012finite, anderson2012efficient}. With this in mind, it becomes apparent that any effort in reducing the bias will translate in a sharp increase in variance. This effect is more pronounced for second-order than for first-order sensitivity for which the corresponding estimator again has bias $O(\epsilon^{2})$ but a variance growing like $O(\epsilon^{-1})$ only~\cite{anderson2012efficient}.
\section{An integral formula for second-order parameter sensitivities}

\begin{assumption}
\label{assumption:propensity_dynkin}
The propensity function $\lambda$ satisfies this assumption at the parameter value $\theta \in \mathbb{R}^{d}$ if the following are true:
\begin{enumerate}
    \item[(A)] The function $\lambda(\cdot,\theta)$ is stoichiometrically admissible, \emph{i.e.}:
    \begin{equation*}
    \forall k \in [\![1, M]\!], \quad \forall x \in \mathbb{N}^{N}, \quad \lambda_{k}(x,\theta) > 0 \implies (x + \zeta_k) \succcurlyeq 0.
    \end{equation*}
    \item[(B)] The function $\lambda(\cdot,\theta)$ satisfies the following polynomial growth condition:
    \begin{equation*}
    \exists C_{\lambda} \in \mathbb{R}_{+}^{*}, \quad \forall x \in \mathbb{N}^{N}, \quad \sum_{k=1 \atop k \in P}^{M} \lambda_{k}(x, \theta)\leq C_{\lambda} (1+|x|),
    \end{equation*}
    where $P\coloneqq \{k \in [\![1, M]\!]: \langle 1_{M}, \zeta_k \rangle > 0\}$ is the set of indices of those reactions which have a net positive effect on the total population.
    \end{enumerate}
\end{assumption}

Assumption \ref{assumption:propensity_dynkin} guarantees that the martingale problem for $\mathbb{A}_{\theta}$ defined in eq.~\eqref{eq:def_generator} is well posed and thus that the Markov process $(X_{\theta}(t))$ with generator $\mathbb{A}_{\theta}$ and initial state $x$ is well defined (see section 2 in~\cite{gupta2014} and chapter 7 in~\cite{ethier2009markov}). To state the remaining two assumptions, we need to introduce the notion of \emph{function of polynomial growth} which for a function $f : \mathbb{N}^n \xrightarrow[]{} \mathbb{R}$ with $n \in \mathbb{N}^{*}$ means that:
\begin{equation*}
\exists (C, r) \in \mathbb{R}_{+}^{2}, \quad \forall x \in \mathbb{N}^{n}, \quad |f(x)| \leq C (1 + |x|^r).   
\end{equation*}

\begin{assumption}
\label{assumption:propensity_regularity}
The propensity function $\lambda$ satisfies this assumption at the parameter value $\theta \in \mathbb{R}^{d}$ for the index pair $(i,j) \in [\![1, d]\!]^2$ if the following is true:
\begin{enumerate}
    \item[\namedlabel{sub_ass:continuous_diff}{(A)}] $\forall k \in [\![1, M]\!], \quad \forall x \in \mathbb{N}^{N}$, the function $\lambda_k(x, \cdot)$ is in $C_{\text{loc}}^{3}(\mathbb{R}^{d}, \mathbb{R})$,
    \item[\namedlabel{sub_ass:poly_growth}{(B)}] $\forall k \in [\![1, M]\!], \quad \forall x \in \mathbb{N}^{N}$, the function $\partial_{i,j}^{2} \lambda_k(x, \cdot)$ is of polynomial growth,
    \item[\namedlabel{sub_ass:poly_growth_sup}{(C)}] $\forall k \in [\![1, M]\!]$, there exists an $\epsilon \in \Rstarplus$ such that $\sup_{\xi \in (\theta-\epsilon, \theta+\epsilon)}\{|\partial_{i,i,j}^{3} \lambda_k(\cdot, \xi)|, |\partial_{i,j,j}^{3} \lambda_k(\cdot, \xi)|\}$ is of polynomial growth.
\end{enumerate}

\end{assumption}

\begin{assumption}
\label{assumption:polynomial_growth}
The output function $f$ is of polynomial growth.
\end{assumption}

Assumptions~\ref{assumption:propensity_regularity} and~\ref{assumption:polynomial_growth} are used as technical conditions in the proof of theorem~\ref{thm:main_theorem} (see the~\nameref{appendix} as well as section 2 in~\cite{gupta2013unbiased} and section 3 in~\cite{gupta2014} for similar conditions in the context of first-order sensivities for SRNs).

\begin{theorem}
\label{thm:main_theorem}
Let $(X_{\theta}(t))$ be the continuous-time Markov chain with generator $\generator_{\theta}$ defined in eq.~\eqref{eq:def_generator}. Suppose the propensity function $\lambda$ satisfies assumptions \ref{assumption:propensity_dynkin} and \ref{assumption:propensity_regularity} at $\theta$ for $(i,j)$ and the function $f : \Natural \xrightarrow{} \mathbb{R}$ satisfies assumption~\ref{assumption:polynomial_growth}. For any initial state $x_0 \in \Natural^N$ and final time $T \in \Rplus$,
$S_{\theta}^{(i,j)} (x_0, f, T)$ exists and is given by $S_{\theta}^{(i,j)} (x_0, f, T)= \mathbb{E}_{x_0}[s_{\theta}^{(i,j)}
(f, T)]$ where:
\begin{align}
\label{eq:thm}
\begin{split}
s_{\theta}^{(i,j)}
(f, T) &= \sum_{k=1}^{M} \Bigg[\int_{0}^{T} \frac{\partial^2 \lambda_{k}}{\partial \theta_{i} \partial \theta_{j}}(X_{\theta}(s), \theta) \Delta_{\zeta_k} \Psi_{\theta}(X_{\theta}(s), f, T-s)ds \\
&+ \int_{0}^{T} \frac{\partial \lambda_{k}}{\partial \theta_{i}}(X_{\theta}(s), \theta) \Delta_{\zeta_k} S_{\theta}^{(j)}(X_{\theta}(s), f, T-s) ds + \int_{0}^{T} \frac{\partial \lambda_{k}}{\partial \theta_{j}}(X_{\theta}(s), \theta) \Delta_{\zeta_k} S_{\theta}^{(i)}(X_{\theta}(s), f, T-s) ds\Bigg].
\end{split}
\end{align}
\end{theorem}

Theorem~\ref{thm:main_theorem} gives easily verifiable conditions on the network and the output function $f$ under which $\Psi_{\theta}(x,f,t)$ is twice differentiable with respect to $\theta$.
Observe that for a reaction network with only mass action kinetics, the first term in the summation in eq.~\eqref{eq:thm} disappears while the second and third term are only non-zero when $k=i$ or $k=j$ which means that the Hessian in effect combines $M$ elements by summing them two by two.

\section{The Double Bernoulli Path Algorithm}

We now explain how to derive an estimator for the Hessian of $\Psi_{\theta}$ from theorem~\ref{thm:main_theorem} and eq.~\eqref{eq:thm}. In this section, we use the notation $(X_{\theta}(t, x))$ to explicitly indicate that a process with generator $\mathbb{A}_{\theta}$ started from state $x$ at time~$0$, whenever needed.

\subsection{Construction of an unbiased estimator of second-order sensitivities}

\subsubsection{Estimation of the integral formula~\eqref{eq:thm}}

\textbf{Estimation of the first term in eq.~\eqref{eq:thm}.} The first term can be estimated as in the exact Integral Path Algorithm~(eIPA) simply by replacing first-order by second-order derivatives in the approach outlined in~\cite{gupta2018estimation} (see section~3 therein). We henceforth refer to the eIPA as the \emph{Bernoulli Path Algorithm~\emph{(BPA)} for first-order sensitivities}, following~\cite{gupta2021deepcme}. More specifically, introduce what we will call the first-order auxiliary processes:
\begin{align}
\label{eq:first_order_ap}
\begin{cases}
X_{\theta}^{(p, k, 1)}(t) &\coloneqq X_{\theta}(t, X_{\theta}(\sigma_{p})),\\
X_{\theta}^{(p, k, 2)}(t) &\coloneqq X_{\theta}(t, X_{\theta}(\sigma_{p}) + \zeta_{k}),  
\end{cases}
\end{align}
where $\sigma_{p}$ is the $p$-th jump time of $(X_{\theta}(t))$ which we will call the main process, taking $\sigma_{0}=0$ for convenience. Observe that $(X_{\theta}^{(p, k, 2)}(t))$ starts from the state $X_{\theta}(\sigma_{p})$ perturbed by the stoichiometry vector $\zeta_{k}$ of reaction $k$. Using the fact that trajectories of the main process are constant between jump times and the tower property of conditional expectations, observe in eq.~\eqref{eq:thm} that:
\scriptsize
\begin{align}
\sum_{k=1}^{M} \expectation\left[\textcolor[HTML]{D21312}{\int_{0}^{T}}\frac{\partial^2 \lambda_{k}}{\partial \theta_{i} \partial \theta_{j}}(X_{\theta}(s), \theta) \Delta_{\zeta_{k}} \Psi_{\theta}(X_{\theta}(s), f, T-s) \textcolor[HTML]{D21312}{ds}\right] &=\expectation\left[\textcolor[HTML]{D21312}{\sum_{p = 0 \atop \sigma_{p}<T}^{\infty}} \sum_{k=1}^{M}\frac{\partial^2 \lambda_{k}}{\partial \theta_{i} \partial \theta_{j}}(X_{\theta}(\sigma_{p}), \theta)\textcolor[HTML]{D21312}{\int_{\sigma_{p}}^{\sigma_{p+1}\wedge T}}\Delta_{\zeta_{k}} \Psi_{\theta}(X_{\theta}(\sigma_{p}), f, T-s) \textcolor[HTML]{D21312}{ds}\right]\nonumber\\
&=\expectation\left[\textcolor[HTML]{D21312}{\sum_{p = 0 \atop \sigma_{p}<T}^{\infty}} \sum_{k=1}^{M}\frac{\partial^2 \lambda_{k}}{\partial \theta_{i} \partial \theta_{j}}(X_{\theta}(\sigma_{p}), \theta)\textcolor[HTML]{D21312}{\int_{T - (\sigma_{p+1}\wedge T)}^{T - \sigma_{p}}}(f(X_{\theta}^{(p,k,2)}(s)) - f(X_{\theta}^{(p,k,1)}(s)) )\textcolor[HTML]{D21312}{ds}\right].\label{eq:first_term_thm_final}
\end{align}
\normalsize

\textbf{Estimation of the second and third terms in eq.~\eqref{eq:thm}.} Introduce the matrix $\Delta S_{\theta}(x, f, t)$ with:
\begin{equation}
[\Delta S_{\theta}(x, f, t)]_{ik} \coloneqq \Delta_{\zeta_k} S_{\theta}^{(i)}(x, f, t).   
\end{equation}

Observe that the second and third term in eq.~\eqref{eq:thm} are expressed in terms of $\Delta S_{\theta}(x, f, t) \text{J}_{\theta}\lambda(x, \theta)$ and its transpose, where $\text{J}_{\theta}\lambda(x, \theta)$ is the Jacobian of $\lambda$ with respect to $\theta$ at the point $(x, \theta)$. This means that we can simulate only one of the two terms and take its transpose to estimate the second-order sensitivity using eq.~\eqref{eq:thm}. This strategy in particular enforces the symmetry of the Hessian.\newline

Step 0. We know from~\cite{gupta2018estimation} that for any final time $T \in \mathbb{R}_{+}$, $S_{\theta}^{(j)} (x_0, f, T) = \mathbb{E}[s_{\theta}^{(j)}(f, T)]$ where:
\begin{equation}
\label{eq:first_order_bpa}
\begin{split}
s_{\theta}^{(j)}
(f, T) &= \int_{0}^{T}  \sum_{\ell=1}^{M}\frac{\partial \lambda_{\ell}}{\partial \theta_{j}}(X_{\theta}(u, x_{0}), \theta) \Delta_{\zeta_\ell} \Psi_{\theta}(X_{\theta}(u, x_{0}), f, T-u)du.
\end{split}
\end{equation}

Using eq.~\eqref{eq:first_order_bpa} and permuting the resulting two integrals before leveraging again the tower property, we obtain that:
\scriptsize
\begin{equation}
\label{eq:adapted_first_order_bpa}
\textcolor[HTML]{D21312}{\int_{t_{1}}^{t_{2}}} S_{\theta}^{(j)}(x_{0}, f, T-s)\textcolor[HTML]{D21312}{ds} = \expectation\Bigg[\textcolor[HTML]{8728e2}{\int_{0}^{T - t_{1}}} \sum_{\ell=1}^{M}\frac{\partial \lambda_{\ell}}{\partial \theta_{j}}(X_{\theta}(u, x_0), \theta)\textcolor[HTML]{D21312}{\int_{t_{1}}^{\min(t_{2}, T-u)}}
\left(f\left(X_{\theta}(T - s - u, X_{\theta}(u, x_{0})+\zeta_\ell)\right) - f\left(X_{\theta}(T -s-u, X_{\theta}(u, x_{0})\right)\right))\textcolor[HTML]{D21312}{ds}\textcolor[HTML]{8728e2}{du}\Bigg].
\end{equation}
\normalsize

Step 1. Now let us start by recalling that the second term of eq.~\eqref{eq:thm} can be expressed by definition as:
\scriptsize
\begin{align}
\label{eq:second_term_thm}
\begin{split}
\sum_{k=1}^{M}\expectation\Bigg[\textcolor[HTML]{D21312}{\int_{0}^{T}} \frac{\partial \lambda_{k}}{\partial \theta_{i}}(X_{\theta}(s), \theta) \Delta_{\zeta_k} S_{\theta}^{(j)}(X_{\theta}(s), f, T-s) \textcolor[HTML]{D21312}{ds}\Bigg] = \sum_{k=1}^{M}\expectation\Bigg[\textcolor[HTML]{D21312}{\int_{0}^{T}} &\frac{\partial \lambda_{k}}{\partial \theta_{i}}(X_{\theta}(s), \theta) S_{\theta}^{(j)}(X_{\theta}(s)+\zeta_{k}, f, T-s) \textcolor[HTML]{D21312}{ds}\Bigg]\\ 
&- \sum_{k=1}^{M}\expectation\Bigg[\textcolor[HTML]{D21312}{\int_{0}^{T}} \frac{\partial \lambda_{k}}{\partial \theta_{i}}(X_{\theta}(s), \theta) S_{\theta}^{(j)}(X_{\theta}(s)+\zeta_{k}, f, T-s) \textcolor[HTML]{D21312}{ds}\Bigg].
\end{split}
\end{align}
\normalsize

We will focus the derivations on the first term on the right-hand side of eq.~\eqref{eq:second_term_thm}. Using as before that trajectories of the main process are constant between jump times, it can be rewritten as:

\scriptsize
\begin{equation}
\sum_{k=1}^{M}\expectation\left[\textcolor[HTML]{D21312}{\int_{0}^{T}} \frac{\partial \lambda_{k}}{\partial \theta_{i}}(X_{\theta}(s), \theta) S_{\theta}^{(j)}(X_{\theta}(s)+\zeta_{k}, f, T-s) \textcolor[HTML]{D21312}{ds}\right] = \expectation\Bigg[\textcolor[HTML]{D21312}{\sum_{p = 0 \atop \sigma_{p}<T}^{\infty}} \sum_{k=1}^{M}\frac{\partial \lambda_{k}}{\partial \theta_{i}}(X_{\theta}(\sigma_{p}), \theta)\textcolor[HTML]{D21312}{\int_{\sigma_p}^{\sigma_{p+1} \wedge T}}S_{\theta}^{(j)}(X_{\theta}(\sigma_{p})+\zeta_{k}, f, T-s) \textcolor[HTML]{D21312}{ds}\Bigg].
\end{equation}
\normalsize

Using eq.~\eqref{eq:adapted_first_order_bpa} together with the notations introduced in eq.~\eqref{eq:first_order_ap}, it then immediately follows that:
\scriptsize
\begin{align}
\begin{split}
\label{eq:second_term_thm_intermediate}
\sum_{k=1}^{M}\expectation\Bigg[\textcolor[HTML]{D21312}{\int_{0}^{T}} \frac{\partial \lambda_{k}}{\partial \theta_{i}}(X_{\theta}(s), \theta) &S_{\theta}^{(j)}(X_{\theta}(s)+\zeta_{k}, f, T-s) \textcolor[HTML]{D21312}{ds}\Bigg] = \expectation\Bigg[\textcolor[HTML]{D21312}{\sum_{p = 0 \atop \sigma_{p}<T}^{\infty}} \sum_{k=1}^{M}\frac{\partial \lambda_{k}}{\partial \theta_{i}}(X_{\theta}(\sigma_{p}), \theta) \textcolor[HTML]{8728e2}{\int_{0}^{T - \sigma_p}}  \sum_{\ell=1}^{M}\frac{\partial \lambda_{\ell}}{\partial \theta_{j}}(X_{\theta}^{(p, k, 2)}(u), \theta)\\
&\textcolor[HTML]{D21312}{\int_{\sigma_{p}}^{\min(\sigma_{p+1} \wedge T, T-u)}}\left(f\left(X_{\theta}(T - s - u, X_{\theta}^{(p, k, 2)}(u)+\zeta_\ell)\right) - f\left(X_{\theta}(T - s - u, X_{\theta}^{(p, k, 2)}(u))\right) \right)\textcolor[HTML]{D21312}{ds}\textcolor[HTML]{8728e2}{du}\Bigg].
\end{split}
\end{align}
\normalsize

Now introduce what we will call the second-order auxiliary processes:
\begin{equation}
\begin{cases}
X_{\theta}^{(p, k, 1,q, \ell, 1)}(t) &\coloneqq X_{\theta}(t, X_{\theta}^{(p, k, 1)}(\sigma_q^{(p, k, 1)})),\\
X_{\theta}^{(p, k, 1,q, \ell, 2)}(t) &\coloneqq X_{\theta}(t, X_{\theta}^{(p, k, 1)}(\sigma_q^{(p, k, 1)})+\zeta_\ell),
\end{cases} \text{ and } \begin{cases}
X_{\theta}^{(p, k, 2,q, \ell, 1)}(t) &\coloneqq X_{\theta}(t, X_{\theta}^{(p, k, 2)}(\sigma_q^{(p, k, 2)})),\\
X_{\theta}^{(p, k, 2,q, \ell, 2)}(t) &\coloneqq X_{\theta}(t, X_{\theta}^{(p, k, 2)}(\sigma_q^{(p, k, 2)})+\zeta_\ell),
\end{cases}
\end{equation}

where $\sigma_q^{(p, k, \texttt{x})}$ denotes the $q$-th jump of the first-order process $(X_{\theta}^{(p, k, \texttt{x})}(t))$ with $\texttt{x} \in \{1,2\}$.  With the notations, using the fact that trajectories of the first-order processes are themselves constant between jump times, we can further rewrite eq.~\eqref{eq:second_term_thm_intermediate} as:
\scriptsize
\begin{align}
\begin{split}
\label{eq:second_term_thm_intermediate_part}
\sum_{k=1}^{M}\expectation\Bigg[\textcolor[HTML]{D21312}{\int_{0}^{T}} \frac{\partial \lambda_{k}}{\partial \theta_{i}}(X_{\theta}(s), \theta) &S_{\theta}^{(j)}(X_{\theta}(s)+\zeta_{k}, f, T-s) \textcolor[HTML]{D21312}{ds}\Bigg] =
\expectation\Bigg[\textcolor[HTML]{D21312}{\sum_{p = 0 \atop \sigma_{p}<T}^{\infty}} \sum_{k=1}^{M}\frac{\partial \lambda_{k}}{\partial \theta_{i}}(X_{\theta}(\sigma_{p}), \theta)\textcolor[HTML]{8728e2}{\sum_{q = 0 \atop \sigma_{q}^{(p,k,2)}<T - \sigma_p}^{\infty}} \sum_{\ell=1}^{M}\frac{\partial \lambda_{\ell}}{\partial \theta_{j}}(X_{\theta}^{(p,k,2)}(\sigma_q^{(p, k, 2)}), \theta)\\
&\textcolor[HTML]{8728e2}{\int_{\sigma_q^{(p,k,2)}}^{\sigma_{q+1}^{(p,k,2)}\wedge (T - \sigma_p)}} \textcolor[HTML]{D21312}{\int_{\sigma_{p}}^{\min(\sigma_{p+1} \wedge T, T-u)}}\left(f\left(X_{\theta}^{(p, k, 2,q, \ell, 2)}(T-s-u)\right) - f\left(X_{\theta}^{(p, k, 2,q, \ell, 1)}(T-s-u)\right)\right)\textcolor[HTML]{D21312}{ds}\textcolor[HTML]{8728e2}{du}\Bigg].
\end{split}
\end{align}
\normalsize

Step 2. The second term on the right-hand side of eq.~\eqref{eq:second_term_thm} can be treated similarly. Putting everything together, we get that:
\scriptsize
\begin{align}
\begin{split}
\label{eq:second_term_thm_final}
\sum_{k=1}^{M}\expectation\Bigg[\textcolor[HTML]{D21312}{\int_{0}^{T}} \frac{\partial \lambda_{k}}{\partial \theta_{i}}(X_{\theta}(s), \theta) &\Delta_{\zeta_k} S_{\theta}^{(j)}(X_{\theta}(s), f, T-s) \textcolor[HTML]{D21312}{ds}\Bigg]\\
&= \expectation\Bigg[\textcolor[HTML]{D21312}{\sum_{p = 0 \atop \sigma_{p}<T}^{\infty}} \sum_{k=1}^{M}\frac{\partial \lambda_{k}}{\partial \theta_{i}}(X_{\theta}(\sigma_{p}), \theta)\textcolor[HTML]{8728e2}{\sum_{q = 0 \atop \sigma_{q}^{(p,k,2)}<T - \sigma_p}^{\infty}} \sum_{\ell=1}^{M}\frac{\partial \lambda_{\ell}}{\partial \theta_{j}}(X_{\theta}^{(p,k,2)}(\sigma_q^{(p, k, 2)}), \theta)\\
&\textcolor[HTML]{8728e2}{\int_{\sigma_q^{(p,k,2)}}^{\sigma_{q+1}^{(p,k,2)}\wedge (T - \sigma_p)}} \textcolor[HTML]{D21312}{\int_{\sigma_{p}}^{\min(\sigma_{p+1} \wedge T, T-u)}}\left(f\left(X_{\theta}^{(p, k, 2,q, \ell, 2)}(T-s-u)\right) - f\left(X_{\theta}^{(p, k, 2,q, \ell, 1)}(T-s-u)\right)\right)\textcolor[HTML]{D21312}{ds}\textcolor[HTML]{8728e2}{du}\\
-&\textcolor[HTML]{D21312}{\sum_{p = 0 \atop \sigma_{p}<T}^{\infty}} \sum_{k=1}^{M}\frac{\partial \lambda_{k}}{\partial \theta_{i}}(X_{\theta}(\sigma_{p}), \theta) \textcolor[HTML]{8728e2}{\sum_{q = 0 \atop \sigma_{q}^{(p,k,1)}<T - \sigma_p}^{\infty}} \sum_{\ell=1}^{M}\frac{\partial \lambda_{\ell}}{\partial \theta_{j}}(X_{\theta}^{(p,k,1)}(\sigma_q^{(p, k, 1)}), \theta)\\
& \textcolor[HTML]{8728e2}{\int_{\sigma_q^{(p,k,1)}}^{\sigma_{q+1}^{(p,k,1)}\wedge (T - \sigma_p)}} \textcolor[HTML]{D21312}{\int_{\sigma_{p}}^{\min(\sigma_{p+1} \wedge T, T-u)}}\left(f\left(X_{\theta}^{(p, k, 1,q, \ell, 2)}(T-s-u)\right) - f\left(X_{\theta}^{(p, k, 1,q, \ell, 1)}(T-s-u)\right)\right)\textcolor[HTML]{D21312}{ds}\textcolor[HTML]{8728e2}{du}\Bigg].
\end{split}
\end{align}
\normalsize

\textbf{Estimation of eq.~\eqref{eq:thm}.} With eq.~\eqref{eq:first_term_thm_final} and~\eqref{eq:second_term_thm_final} we  have everything to simulate a random variable with the same expectation as $s_{\theta}(f, T)$ in eq.~\eqref{eq:thm} of theorem~\ref{thm:main_theorem}. Observe that we have introduced a hierarchy of auxiliary processes for this, which is constituted of the main, the first-order and the second-order processes. At each jump time $\sigma_p$ of the main process $(X_{\theta}(t))$, $2M$ first-order processes are generated until time $T-\sigma_p$. At each jump time $\sigma_q^{(p,k,\texttt{x})}$ of a first-order process $(X_{\theta}^{(p, k, \texttt{x})}(t))$, $2M$ second-order processes are generated. With this in mind, it is worth highlighting here that this hierarchy does not scale with the number of parameters but only with the number~$M$ of reactions.\newline

Observe that this approach allows us to generate one sample of the whole Hessian for multiple output functions from a single trajectory of the main process, which is also the case for the GT method. For non-centred finite differences, $1+ d + d(d+1)/2$ trajectories are needed to get one sample of the Hessian while $1+2d^2$ trajectories are needed for centred finite differences.
Also mote that estimates for the average output $\Psi_{\theta}(x,f,t)$ and its Jacobian can be obtained from the same trajectory of the main process, like with the GT method. For non-centred finite differences, the trajectory corresponding to the unperturbed parameter can be used to estimate the average output while the $1+d$ trajectories corresponding to the unperturbed parameter and those where a single element is perturbed can be used for the Jacobian. For centred finite differences, the trajectory corresponding to the unperturbed parameter can be used to estimate the average output while the $2d$ trajectories with the same parameter perturbed twice can be exploited to estimate the Jacobian, albeit on a grid with double the step-size.
Finally, it is clear from eq.~\eqref{eq:thm} and the discussion above that the approach developed in~\cite{gupta2018estimation} using tau-leap simulations to obtain a (biased) estimator of first-order sensitivities could be adapted here~(see the presentation of $\tau$IPA in section 3 of~\cite{gupta2018estimation}).

\subsubsection{Variance reduction of the estimator}
In order to reduce the variance of the estimator while keeping it unbiased, we couple together:

\begin{itemize}
    \item the first-order processes $X_{\theta}^{(p, k, 1)}(t)$ and  $X_{\theta}^{(p, k, 2)}(t)$,
    \item the second-order processes $X_{\theta}^{(p, k, \texttt{x}, q, \ell, 1)}(t)$ and  $X_{\theta}^{(p, k, \texttt{x}, q, \ell, 2)}(t)$ where $\texttt{x} \in \{1,2\}$,
\end{itemize}

using the first-order split coupling introduced in~\cite{anderson2012efficient}. This means for example that we use:
\begin{align}
X_{\theta}^{(p, k, 1)}(t) \coloneqq& \text{ }X_{\theta}(t, X_{\theta}(\sigma_{p})) + \sum_{k'=1}^{M}\zeta_{k'} Y_{k'}\left(\int_0^t \left(\lambda_{k'}(X_{\theta}^{(p, k', 1)}(s), \theta) \wedge \lambda_{k'}(X_{\theta}^{(p, k', 2)}(s), \theta) \right)ds\right)\nonumber \\
&+ \sum_{k'=1}^{M}\zeta_{k'} Y_{k'}^{(1)}\left(\int_0^t \left(\lambda_{k'}(X_{\theta}^{(p, k', 1)}(s), \theta) - \lambda_{k'}(X_{\theta}^{(p, k', 1)}(s), \theta) \wedge \lambda_{k'}(X_{\theta}^{(p, k', 2)}(s), \theta) \right)ds\right),\\
X_{\theta}^{(p, k', 2)}(t) \coloneqq& \text{ }X_{\theta}(t, X_{\theta}(\sigma_{p}))+\zeta_k + \sum_{k'=1}^{M}\zeta_{k'} Y_{k'}\left(\int_0^t \left(\lambda_{k'}(X_{\theta}^{(p, k', 1)}(s), \theta) \wedge \lambda_{k'}(X_{\theta}^{(p, k', 2)}(s), \theta) \right)ds\right)\nonumber \\
&+ \sum_{k'=1}^{M}\zeta_{k'} Y_{k'}^{(2)}\left(\int_0^t \left(\lambda_{k'}(X_{\theta}^{(p, k', 2)}(s), \theta) - \lambda_{k'}(X_{\theta}^{(p, k', 1)}(s), \theta) \wedge \lambda_{k'}(X_{\theta}^{(p, k', 2)}(s), \theta) \right)ds\right),
\end{align}

where $\{(Y_{k'}(t)), (Y_{k'}^{(1)}(t)), (Y_{k'}^{(2)}(t))\}$ are independent, unit-rate Poisson processes. We write $\sigma_q^{(p,k)}$ the $q$-th jump of the first-order process $(X_{\theta}^{(p,k,1)}(t), X_{\theta}^{(p,k,2)}(t))$ and generate second-order processes at each of these jump times~(see panel A in figure~\ref{fig:graphical_bpa}). Note that the processes could alternatively be coupled through the Common Reaction Path scheme, leading to a hierarchy with $(M+1)^2$ leaves instead of $4M^2$~\cite{rathinam2010efficient}. On the other hand, this alternative coupling tends to be less tight and might result in a larger estimator variance~\cite{anderson2012efficient}.

\begin{figure}[h!]
\centering
\includegraphics[width=0.7\linewidth]{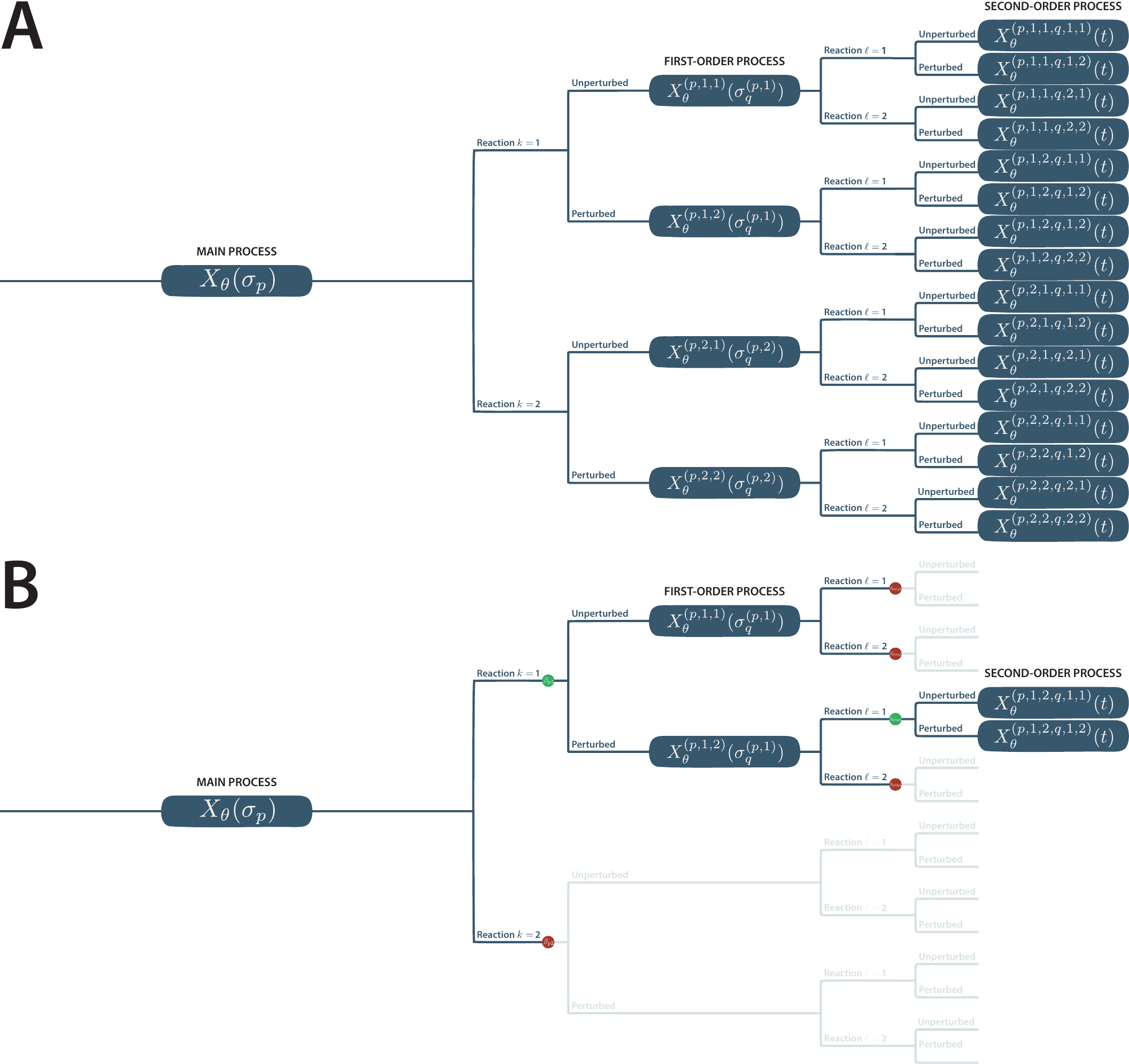}
\caption{\label{fig:graphical_bpa} Representation of the Double BPA as a balanced binary tree where siblings are coupled. (A)~The full tree has $4M^2$ leaves where $M$ is the number of reactions. (B)~Bernoulli random variables are symbolised by dots. They are green when they equal 1 (meaning two coupled auxiliary processes get simulated) and red otherwise. After the introduction of these variables, only parts of the tree get generated while still leaving the estimator unbiased.}
\end{figure}

\subsubsection{Modulation of the computational cost of the estimator}

\textbf{Introduction of Bernoulli random variables.} Introduce the family of independent random variables $\{\beta_{pk}\}$ with Bernoulli distribution of parameter $\gamma_{pk}$ as well as $\{\beta_{pk1q\ell}\}$ and $\{\beta_{pk2q\ell}\}$ with parameter $\gamma_{pk1q\ell}$ and $\gamma_{pk2q\ell}$. These random variables can be inserted in eq.~\eqref{eq:first_term_thm_final} and \eqref{eq:second_term_thm_final} to modulate the computational cost per trajectory of the estimator without introducing any bias. More specifically, we can proceed as follows for eq.~\eqref{eq:first_term_thm_final}:
\scriptsize
\begin{align}
\label{eq:bernoulli_first}
\sum_{k=1}^{M}\expectation\Bigg[\int_{0}^{T}\frac{\partial^2 \lambda_{k}}{\partial \theta_{i} \partial \theta_{j}}(X_{\theta}(s), \theta) &\Delta_{\zeta_{k}} \Psi_{\theta}(X_{\theta}(s), f, T-s) ds\Bigg]\nonumber \\
= &\expectation\Bigg[\sum_{p = 0 \atop \sigma_{p}<T}^{\infty}\sum_{k=1}^{M}\frac{\partial^2 \lambda_{k}}{\partial \theta_{i} \partial \theta_{j}}(X_{\theta}(\sigma_{p}), \theta) \textcolor[HTML]{005b96}{\frac{\beta_{pk}}{\gamma_{pk}}}\int_{T - (\sigma_{p+1}\wedge T)}^{T - \sigma_{p}}(f(X_{\theta}^{(p,k,2)}(s), \theta) - f(X_{\theta}^{(p,k,1)}(s), \theta))ds\Bigg],
\end{align}
\normalsize

and for eq.~\eqref{eq:second_term_thm_intermediate_part}:
\scriptsize
\begin{multline} 
\label{eq:bernoulli_second}
\sum_{k=1}^{M}\expectation\left[\int_{0}^{T} \frac{\partial \lambda_{k}}{\partial \theta_{i}}(X_{\theta}(s), \theta) S_{\theta}^{(j)}(X_{\theta}(s)+\zeta_{k}, f, T-s) ds\right] =
\expectation\Bigg[\sum_{p = 0 \atop \sigma_{p}<T}^{\infty}\sum_{k=1}^{M}\frac{\partial \lambda_{k}}{\partial \theta_{i}}(X_{\theta}(\sigma_{p}), \theta)\textcolor[HTML]{005b96}{\frac{\beta_{pk}}{\gamma_{pk}}}\sum_{q = 0 \atop \sigma_{q}^{(p,k)}<T - \sigma_p}^{\infty}\sum_{\ell=1}^{M}\frac{\partial \lambda_{\ell}}{\partial \theta_{j}}(X_{\theta}^{(p,k,2)}(\sigma_q^{(p, k)}), \theta)\textcolor[HTML]{e55bb7}{\frac{\beta_{pk2q\ell}}{\gamma_{pk2q\ell}}}\\
\int_{\sigma_q^{(p,k)}}^{\sigma_{q+1}^{(p,k)}\wedge (T - \sigma_p)}\int_{\sigma_{p}}^{\min(\sigma_{p+1} \wedge T, T-u)}\left(f\left(X_{\theta}^{(p, k, 2,q, \ell, 2)}(T-s-u)\right) - f\left(X_{\theta}^{(p, k, 2,q, \ell, 1)}(T-s-u)\right)\right)dsdu\Bigg].
\end{multline}
\normalsize

\textbf{Choice of the parameter $\gamma_{pk}$ for $\beta_{pk}$.} Let $\mu_1$ be the upper bound for the expected number of desired pairs of first-order auxiliary processes $(X_{\theta}^{(p, k, 1)}(t))$ and  $(X_{\theta}^{(p, k, 2)}(t))$ per trajectory of the main process $(X_{\theta}(t))$. Choose $\gamma_{pk}$ to be:
\begin{equation}
\label{eq:first_order_bernoulli}
\gamma_{pk} \coloneqq 1 \wedge \frac{\delta_{p}}{C_1} \left(\sum_{i=1}^{d}2\bigg|\frac{\partial \lambda_{k}}{\partial \theta_{i}}(X_\theta(\sigma_{p}), \theta)\bigg|+ \sum_{i,j} \bigg|\frac{\partial^2 \lambda_{k}}{\partial \theta_{i}\partial \theta_{j}}(X_\theta(\sigma_{p}), \theta)\bigg|\right) \text{, with: } \delta_{p} \coloneqq (\sigma_{p+1} \wedge T) - \sigma_{p}.
\end{equation}

Observe that for a given trajectory of the main process and the associated auxiliary processes, the average number of pairs of first-order auxiliary processes is bounded from above as follows:
\begin{align}
\sum_{p = 0 \atop \sigma_{p}<T}^{\infty}\sum_{k=1}^{M} \mathbb{E} \big[\beta_{pk}\big] = \sum_{p = 0 \atop \sigma_{p}<T}^{\infty} \sum_{k=1}^{M}  \gamma_{pk} &\leq  \sum_{p = 0 \atop \sigma_{p}<T}^{\infty} \sum_{k=1}^{M} \frac{\delta_{p}}{C_1} \left(\sum_{i=1}^{d}2\bigg|\frac{\partial \lambda_{k}}{\partial \theta_{i}}(X_\theta(\sigma_{p}), \theta)\bigg|+ \sum_{i,j} \bigg|\frac{\partial^2 \lambda_{k}}{\partial \theta_{i}\partial \theta_{j}}(X_\theta(\sigma_{p}), \theta)\bigg|\right)\nonumber \\
&=  \frac{1}{C_1} \int_{0}^{T} \sum_{k=1}^{M} \left(\sum_{i=1}^{d}2\bigg|\frac{\partial \lambda_{k}}{\partial \theta_{i}}(X_\theta(s), \theta)\bigg|+ \sum_{i,j} \bigg|\frac{\partial^2 \lambda_{k}}{\partial \theta_{i}\partial \theta_{j}}(X_\theta(s), \theta)\bigg|\right)ds.
\end{align}

Choosing the normalisation constant $C_1$ to be: 
\begin{equation}
\label{eq:c_1}
C_1 \coloneqq \frac{1}{\mu_{1}}\mathbb{E}\left[\int_{0}^{T} \sum_{k=1}^{M} \left(\sum_{i=1}^{d}2\bigg|\frac{\partial \lambda_{k}}{\partial \theta_{i}}(X_\theta(s), \theta)\bigg|+ \sum_{i,j} \bigg|\frac{\partial^2 \lambda_{k}}{\partial \theta_{i}\partial \theta_{j}}(X_\theta(s), \theta)\bigg|\right)ds\right],    
\end{equation}

we get as desired that the average number of pairs of first-order processes per trajectory is upper bounded by $\mu_{1}$, i.e.:
\begin{equation}
\label{ineq:upper_bound}
\mathbb{E}\left[\sum_{p = 0 \atop \sigma_{p}<T}^{\infty}\sum_{k=1}^{M}  \beta_{pk}\right] \leq \mu_{1}.
\end{equation}

\textbf{Choice of the parameters $\gamma_{pk1q\ell}$ and $\gamma_{pk2q\ell}$ for $\beta_{pk1q\ell}$ and $\beta_{pk2q\ell}$.} Similarly, write $\mu_2$ the upper bound for the expected number of desired pairs of second-order auxiliary processes $(X_{\theta}^{(p, k, \texttt{x}, q, \ell, 1)}(t))$ and  $(X_{\theta}^{(p, k, \texttt{x}, q, \ell, 2)}(t))$ per trajectory of the main process once the number of first-order auxiliary processes is modulated by $\gamma_{pk}$. Let us define $\gamma_{pk1q\ell}$ and $\gamma_{pk2q\ell}$ as:
\begin{equation}    
\label{eq:second_order_bernoulli}
\begin{cases}
\gamma_{pk1q\ell} &\coloneqq 1 \wedge \frac{\delta_{p}\delta_{q}^{(p,k)}}{C_2} \sum_{j=1}^{d}\bigg|\frac{\partial \lambda_{\ell}}{\partial \theta_{j}}(X_{\theta}^{(p,k,1)}(\sigma_q^{(p, k)}), \theta)\bigg|,\\
\gamma_{pk2q\ell} &\coloneqq 1 \wedge \frac{\delta_{p}\delta_{q}^{(p,k)}}{C_2} \sum_{j=1}^{d}\bigg|\frac{\partial \lambda_{\ell}}{\partial \theta_{j}}(X_{\theta}^{(p,k,2)}(\sigma_q^{(p, k)}), \theta)\bigg|,
\end{cases} \text{ with: } \delta_{q}^{(p,k)} \coloneqq (\sigma_{q+1}^{(p,k)} \wedge (T - \sigma_p)) - \sigma_{q}^{(p,k)}.
\end{equation}

Choosing the second normalisation constant $C_2$ to be:
\begin{equation}
\label{eq:c_2}
C_{2} \coloneqq \frac{1}{\mu_{2}}  \mathbb{E}\left[\sum_{p = 0 \atop \sigma_{p}<T}^{\infty} \sum_{k=1}^{M} 1_{\gamma_{pk}=1}\delta_{p}\int_{0}^{T-\sigma_{p}}\sum_{\ell=1}^{M}\sum_{j=1}^{d} \left(\bigg|\frac{\partial \lambda_{\ell}}{\partial \theta_{j}}(X_{\theta}^{(p,k,1)}(u), \theta)\bigg|+\bigg|\frac{\partial \lambda_{\ell}}{\partial \theta_{j}}(X_{\theta}^{(p,k,2)}(u), \theta)\bigg|\right)du\right],   
\end{equation}

we obtain as desired that the average number of pairs of second-order processes per trajectory is at most $\mu_{2}$, i.e.:
\begin{equation}
\mathbb{E}\left[\sum_{p = 0 \atop \sigma_{p}<T}^{\infty}\sum_{k=1}^{M} 1_{\gamma_{pk}=1} \sum_{q = 0 \atop \sigma_{q}^{(p,k)}<T - \sigma_p}^{\infty}  \sum_{\ell=1}^{M} \left(\beta_{pk1q\ell} + \beta_{pk2q\ell}\right)\right] \leq \mu_{2}.    
\end{equation}

The estimator resulting from this is represented graphically in panel B of figure~\ref{fig:graphical_bpa}.
The introduction of a dependency of $\gamma_{pk}$, $\gamma_{pk1q\ell}$ and $\gamma_{pk2q\ell}$ on the first- and second-order derivative of the propensity $\lambda$ in eq.~\eqref{eq:first_order_bernoulli} and~\eqref{eq:second_order_bernoulli} is motivated by the fact that if the parameters $\theta_i$ and $\theta_j$ have a large influence on $\lambda_k$ at the current state then there is a greater chance that it will affect the final value of the sample as seen in eq.~\eqref{eq:second_term_thm_final}. It is clear how to adapt the definition of the Bernoulli parameters in the case when the derivative with respect to only a specific parameter combination $(\theta_i, \theta_j)$ is of interest.
Before moving to the numerical examples, let us emphasise again that the introduction of $\{\beta_{pk}\}$, $\{\beta_{pk1q\ell}\}$ and $\{\beta_{pk2q\ell}\}$ as well as the specific choice of normalisation constant $C_1$ and $C_2$ does not introduce a bias in the resulting estimator. Also, if modulation of the computational cost per trajectory is not desired, setting $\gamma_{pk} = \gamma_{pk1q\ell}= \gamma_{pk2q\ell} \equiv 1$ results in auxiliary processes being generated at each jump time.

\subsection{Implementation of the Double BPA}

We first introduce the methods which are specific to the Double BPA and then the standard ones for the simulation of SRNs.

\subsubsection{Methods specific to the Double BPA}

To ensure a smooth transition from the estimator introduced above and the pseudo-code below, we provide in figure~\ref{fig:graphical_pseudocode} the correspondence between the hierarchy of processes defined by the Double BPA and the quantities being updated in the pseudo-code. In this code, we highlight in salmon the parts corresponding to the simulation of SRNs using the Stochastic Simulation Algorithm~(SSA)~\cite{gillespie1976general,gillespie1977exact} or the coupled modified Next Reaction Method~(mNRM)~\cite{anderson2007modified} to clearly dissociate them from the Double BPA-specific routines.\newline

\begin{figure*}[h!]
\centering
\includegraphics[width=0.6\linewidth]{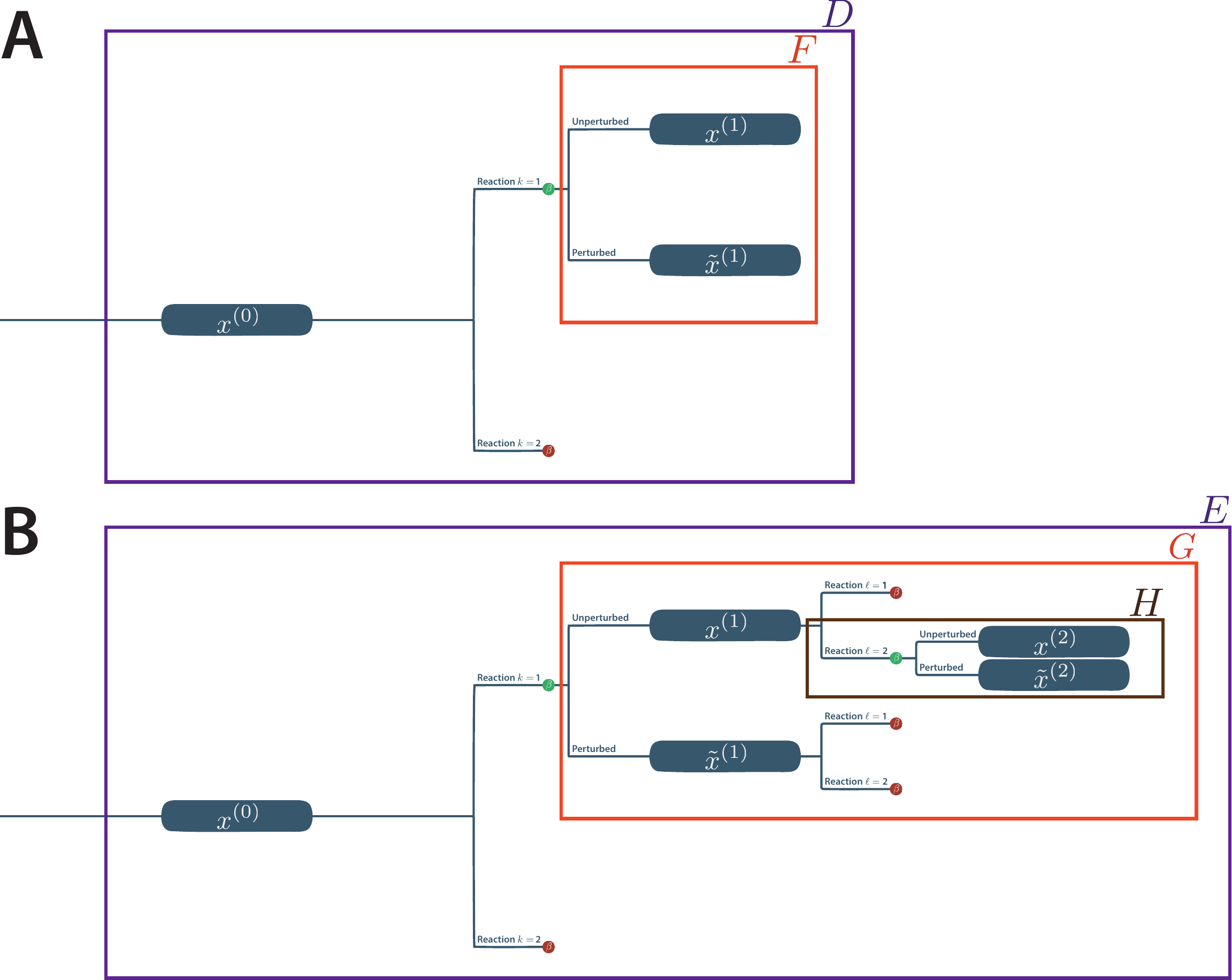}
\caption{\label{fig:graphical_pseudocode} Correspondence between the hierarchy of processes defined by the Double BPA and the quantities being updated in the pseudo-code.}
\end{figure*}

Associated to the variables in the pseudo-code, superscript $(0)$ refers to the main process, $(1)$ to first-order processes and $(2)$ to second-order processes. The constants $C_1$ and $C_2$ defined in eq.~\eqref{eq:c_1} and \eqref{eq:c_2} are estimated by running the BPA a limited number of times. Estimation of these constants need not be accurate as their specific value will only modulate the computational cost per trajectory without compromising the estimator's unbiasedness. $\mathcal{B}(\gamma)$ with $\gamma \in [0,1]$ is a Bernoulli variable with parameter $p$.

\begin{algorithm}[H]
\caption{Generate one sample with the Double BPA by computing the integral over the trajectory of the main process from $0$ to $T$ in eq.~\eqref{eq:first_term_thm_final} and~\eqref{eq:second_term_thm_final}}
\begin{algorithmic}[1]
\Function{GenerateBPASecondOrder}{$x^{(0)}$, $T^{(0)}$, $C$}
\State $D_{ijk} \leftarrow 0$, $E_{ijk} \leftarrow 0$ for $i,j \in [\![1, d]\!]$ and $k \in [\![1, M]\!]$ 
\State \tikzmark{highlight1-start}$t^{(0)} \leftarrow 0$
\While{$t^{(0)} < T^{(0)}$}
\State $(\Delta t^{(0)}, k^{*}) \leftarrow$ \Call{NextReactionSSA}{$x$}
\State $\Delta t^{(0)} \leftarrow \min(\Delta t^{(0)}, T^{(0)} - t^{(0)})$\tikzmark{highlight1-end}
\For{$k \in [\![1, M]\!]$}
\State $\gamma \leftarrow  1 \wedge \frac{\Delta t^{(0)}}{C_1} \left(\sum_{i=1}^{d}2\left|\frac{\partial \lambda_{k}}{\partial \theta_{i}}(x^{(0)}, \theta)\right|+ \sum_{i,j} \left|\frac{\partial^2 \lambda_{k}}{\partial \theta_{i}\partial \theta_{j}}(x^{(0)}, \theta)\right|\right)$
\State $\beta \sim \mathcal{B}(\gamma)$
\If{$\beta = 1$} \Comment{start (coupled) first-order auxiliary paths for reaction $k$}
\State $x^{(1)} \leftarrow x^{(0)}$, $\Tilde{x}^{(1)} \leftarrow x^{(0)} + \zeta_{k}$
\State \parbox[t]{420pt}{($\Delta D$, $\Delta E$) $\leftarrow$  \Call{FirstOrderDifference}{$x^{(1)}$, $\Tilde{x}^{(1)}$, $t^{(0)}$, $\Delta t^{(0)}$, $T^{(0)}$, $C_2$}\strut}
\State $D_{ijk} \leftarrow D_{ijk} + \frac{\partial^2 \lambda_{k}}{\partial \theta_{i}\partial \theta_{j}}(x^{(0)}, \theta)  \times  \frac{\Delta D}{\gamma}$ for $i,j \in [\![1, d]\!]$
\State $E_{ijk} \leftarrow E_{ijk} + \frac{\partial \lambda_{k}}{\partial \theta_{i}}(x^{(0)}, \theta) \times \frac{\Delta E_j}{\gamma}$ for $i,j \in [\![1, d]\!]$
\EndIf
\EndFor 
\State \tikzmark{highlight2-start}$t^{(0)} \leftarrow t^{(0)} + \Delta t^{(0)}$
\State $x^{(0)} \leftarrow$ \Call{UpdateState}{$x^{(0)}$, $ k^{*}$}
\EndWhile\tikzmark{highlight2-end}
\State $[\mathcal{H}_1]_{ij} \leftarrow \sum_{k=1}^{M} D_{ijk} $
\State $[\mathcal{H}_2]_{ij} \leftarrow \sum_{k=1}^{M} E_{ijk} $
\State $\mathcal{H} \leftarrow \mathcal{H}_1 + \mathcal{H}_2 + \mathcal{H}_2^{T}$
\State \Return $\mathcal{H}$
\EndFunction
\end{algorithmic}
\end{algorithm}

\highlight{highlight1}{red}{{-1cm,1em}}{{11.3cm,-0.3em}}
\highlight{highlight2}{red}{{-1.65cm,1em}}{{14.8cm,-0.3em}}

Eq.~\eqref{eq:first_term_thm_final} requires the computation of an integral over a first-order trajectory from $T - (\sigma_{p+1}\wedge T)$ to $T - \sigma_{p}$ while eq.~\eqref{eq:second_term_thm_final} requires one from $0$ to $T - \sigma_{p}$. For convenience, we therefore use the Markov property of first-order trajectories to compute the required integrals first from $0$ to $T - (\sigma_{p+1}\wedge T)$ and then from $T - (\sigma_{p+1}\wedge T)$ to $T - \sigma_{p}$ in algorithm~\ref{alg:algo_2}.

\begin{algorithm}[H]
\caption{Compute the integral over a first-order trajectory from $T - (\sigma_{p+1}\wedge T)$ to $T - \sigma_{p}$ in eq.~\eqref{eq:first_term_thm_final} and from $0$ to $T - \sigma_{p}$ in eq.~\eqref{eq:second_term_thm_final}}
\label{alg:algo_2}
\begin{algorithmic}[1]
\Function{FirstOrderDifference}{$x^{(1)}$, $\Tilde{x}^{(1)}$, $t^{(0)}$, $\Delta t^{(0)}$, $T^{(0)}$, $C_2$}
\State $\Delta D \leftarrow 0$, $\Delta E \leftarrow 0_d$
\State $\varphi^{(1)} \leftarrow 0$, $T^{(1)} \leftarrow T^{(0)} - (t^{(0)} + \Delta t^{(0)})$ \Comment{integrate from $0$ to $T - (\sigma_{p+1} \wedge T)$}
\State \parbox[t]{460pt}{($x^{(1)}$, $\Tilde{x}^{(1)}$, $\Delta D$, $\Delta E$) $\leftarrow$  \Call{FirstOrderDifferenceSegment}{$x^{(1)}$, $\Tilde{x}^{(1)}$, $\varphi^{(1)}$, $T^{(1)}$, $\Delta D$, $\Delta E$, $t^{(0)}$, $\Delta t^{(0)}$, $T^{(0)}$, $C_2$}\strut}
\State $\varphi^{(1)} \leftarrow T^{(0)} - (t^{(0)} + \Delta t^{(0)}$), $T^{(1)} \leftarrow T^{(0)} - t^{(0)}$ \Comment{integrate from $T - (\sigma_{p+1} \wedge T)$ to $T - \sigma_{p}$}
\State  \parbox[t]{460pt}{($x^{(1)}$, $\Tilde{x}^{(1)}$, $\Delta D$, $\Delta E$) $\leftarrow$  \Call{FirstOrderDifferenceSegment}{$x^{(1)}$, $\Tilde{x}^{(1)}$, $\varphi^{(1)}$, $T^{(1)}$, $\Delta D$, $\Delta E$, $t^{(0)}$, $\Delta t^{(0)}$, $T^{(0)}$, $C_2$}\strut}
\State \Return $\Delta D$, $\Delta E$
\EndFunction
\end{algorithmic}
\end{algorithm}

\newpage
\begin{algorithm}[H]
\caption{Compute the integral over of a first-order trajectory either from $0$ to $T - (\sigma_{p+1}\wedge T)$ in eq.~\eqref{eq:second_term_thm_final} or from $T - (\sigma_{p+1}\wedge T)$ to $T - \sigma_{p}$ in eq.~\eqref{eq:first_term_thm_final} and~\eqref{eq:second_term_thm_final}}
\begin{algorithmic}[1]
\Function{FirstOrderDifferenceSegment}{$x^{(1)}$, $\Tilde{x}^{(1)}$, $\varphi^{(1)}$, $T^{(1)}$, $\Delta D$, $\Delta E$, $t^{(0)}$, $\Delta t^{(0)}$, $T^{(0)}$, $C_2$}
\State $F \leftarrow 0$, $G _{j\ell} \leftarrow 0$ for $j \in [\![1, d]\!]$ and $\ell \in [\![1, M]\!]$ 
\State  \tikzmark{highlight3-start}$t^{(1)} \leftarrow \varphi^{(1)}$
\State $I_{\ell m} \leftarrow 0$ and $\tau_{\ell m} \sim \mathcal{E}(1)$ for $\ell \in [\![1, M]\!]$ and $m \in [\![1, 3]\!]$
\While{$t^{(1)} < T^{(1)}$}
\State ($\Delta t^{(1)}$, $k^{*}$, $m^{*}$, $I$, $\tau$) $\leftarrow$ \Call{NextReactionCoupledNRM}{$x^{(1)}$, $\Tilde{x}^{(1)}$, $I$, $\tau$}
\State $\Delta t^{(1)} \leftarrow \min(\Delta t^{(1)}, T^{(1)}-t^{(1)})$\tikzmark{highlight3-end}
\If{$T^{(1)} = T^{(0)} - t^{(0)}$} \Comment{integral over first-order trajectory in eq.~\eqref{eq:first_term_thm_final}}
\State $\Delta F \leftarrow$ \Call{FirstTermDifference}{$x^{(1)}$, $\Tilde{x}^{(1)}$, $\Delta t^{(1)}$}
\State $F \leftarrow F + \Delta F$
\EndIf
\For{\texttt{x} $\in \{x^{(1)},\Tilde{x}^{(1)}\}$} \Comment{integral over first-order trajectory in eq.~\eqref{eq:second_term_thm_final}}
\State $u_{\text{min}} \leftarrow t^{(1)}$, $u_{\text{max}} \leftarrow t^{(1)} + \Delta t^{(1)}$
\State $r_{\text{bott., left}} \leftarrow T^{(0)} - (t^{(0)}+\Delta t^{(0)}) - u_{\min}$, $r_{\text{top, left}} \leftarrow T^{(0)} - t^{(0)} - u_{\min}$
\State $r_{\text{bott., right}} \leftarrow T^{(0)} - (t^{(0)}+\Delta t^{(0)}) - u_{\max}$, $r_{\text{top, right}} \leftarrow T^{(0)} - t^{(0)} - u_{\max}$
\For{$\ell \in [\![1, M]\!]$}
\State $\gamma \leftarrow 1 \wedge \frac{\Delta t^{(0)}\Delta t^{(1)}}{C_2} \sum_{j=1}^{d}\bigg|\frac{\partial \lambda_{\ell}}{\partial \theta_{j}}(\texttt{x}, \theta)\bigg|$
\State $\beta \sim \mathcal{B}(\gamma)$
\If{$\beta = 1$} \Comment{start (coupled) second-order auxiliary paths for reaction $\ell$}
\State $x^{(2)} \leftarrow \texttt{x}$, $\Tilde{x}^{(2)} \leftarrow \texttt{x} + \zeta_{\ell}$
\State \parbox[t]{400pt}{$\Delta G \leftarrow$ \Call{SecondOrderDifference}{$x^{(2)}$, $\Tilde{x}^{(2)}$, $\varphi^{(1)}$, $t^{(0)}$, $\Delta t^{(0)}$, $T^{(0)}$, $r_{\text{bott., left}}$, $r_{\text{top, left}}$, $r_{\text{bott., right}}$, $r_{\text{top, right}}$, $u_{\text{min}}$, $u_{\text{max}}$}}
\If{$\texttt{x} = x^{(1)}$} 
\State $G_{j\ell} \leftarrow G_{j\ell} - \frac{\partial \lambda_{\ell}}{\partial \theta_{j}}(\texttt{x}, \theta) \frac{\Delta G}{\gamma}$ for $j \in [\![1, d]\!]$
\ElsIf{$\texttt{x} = \Tilde{x}^{(1)}$}
\State $G_{j\ell} \leftarrow G_{j\ell} + \frac{\partial \lambda_{\ell}}{\partial \theta_{j}}(\texttt{x}, \theta) \frac{\Delta G}{\gamma}$ for $j \in [\![1, d]\!]$
\EndIf
\EndIf 
\EndFor
\EndFor
\State \tikzmark{highlight4-start}$t^{(1)} \leftarrow$ $t^{(1)} + \Delta t^{(1)}$
\If{$t^{(1)} \leq T^{(1)}$}
\State ($x^{(1)}$, $\Tilde{x}^{(1)}$) $\leftarrow$ \Call{UpdateCoupledState}{$x^{(1)}$, $\Tilde{x}^{(1)}$, $k^{*}$, $m^{*}$}
\EndIf
\EndWhile\tikzmark{highlight4-end}
\State $G_{j} \leftarrow \sum_{\ell=1}^{M} G_{j\ell}$
\State $\Delta D \leftarrow \Delta D + F$, $\Delta E \leftarrow \Delta E + G$
\State \Return $x^{(1)}$, $\Tilde{x}^{(1)}$, $\Delta D$, $\Delta E$
\EndFunction
\end{algorithmic}
\end{algorithm}

\highlight{highlight3}{red}{{-1cm,1em}}{{11.3cm,-0.3em}}
\highlight{highlight4}{red}{{-1.65cm,1em}}{{14.8cm,-0.3em}}

\begin{algorithm}[H]
\caption{Update the integral over a first-order trajectory in eq.~\eqref{eq:first_term_thm_final}}
\begin{algorithmic}[1]
\Function{FirstTermDifference}{$x^{(1)}$, $\Tilde{x}^{(1)}$, $\Delta t^{(1)}$}
\State $f_d \leftarrow f(\Tilde{x}^{(1)}) - f(x^{(1)})$
\State $\Delta F \leftarrow f_d \times \Delta t^{(1)}$
\State \Return $\Delta F$
\EndFunction
\end{algorithmic}
\end{algorithm}

Eq.~\eqref{eq:second_term_thm_final} requires the computation of an integral over of a second-order trajectory over the domain $[\sigma_q^{(p,k)},\sigma_{q+1}^{(p,k)}\wedge (T - \sigma_p)] \times [\sigma_{p},\min(\sigma_{p+1} \wedge T, T-u)]$. 
To explain how to perform this integration, let us change the integration variable in the integral against $s$ in eq.~\eqref{eq:second_term_thm_intermediate_part} by setting $r = T- s -u$ and consider $\texttt{x} \in \{1,2\}$. When $\min(\sigma_{p+1} \wedge T, T-u) = \sigma_{p+1} \wedge T$, this translates into:

\scriptsize
\begin{align}
\label{eq:before_switching}
\begin{split}
\textcolor[HTML]{8728e2}{\int_{\sigma_q^{(p,k)}}^{\sigma_{q+1}^{(p,k)}\wedge (T - \sigma_p)}}\textcolor[HTML]{D21312}{\int_{\sigma_{p}}^{\min(\sigma_{p+1} \wedge T, T-u)}}&\left(f\left(X_{\theta}^{(p, k, \texttt{x}, q, \ell, 2)}(T-s-u)\right) - f\left(X_{\theta}^{(p, k, \texttt{x}, q, \ell, 1)}(T-s-u)\right)\right)\textcolor[HTML]{D21312}{ds}\textcolor[HTML]{8728e2}{du}\\ &= \textcolor[HTML]{8728e2}{\int_{\sigma_q^{(p,k)}}^{\sigma_{q+1}^{(p,k)}\wedge (T - \sigma_p)}}\textcolor[HTML]{D21312}{\int_{T - (\sigma_{p+1} \wedge T) - u}^{T - \sigma_{p} - u}} \left(f\left(X_{\theta}^{(p, k, \texttt{x}, q, \ell, 2)}(r)\right) - f\left(X_{\theta}^{(p, k, \texttt{x}, q, \ell, 1)}(r)\right)\right)\textcolor[HTML]{D21312}{dr}\textcolor[HTML]{8728e2}{du},
\end{split}
\end{align}
\normalsize

while when $\min(\sigma_{p+1} \wedge T, T-u) = T-u$, we have:

\scriptsize
\begin{align}
\label{eq:after_switching}
\begin{split}
\textcolor[HTML]{8728e2}{\int_{\sigma_q^{(p,k)}}^{\sigma_{q+1}^{(p,k)}\wedge (T - \sigma_p)}}\textcolor[HTML]{D21312}{\int_{\sigma_{p}}^{\min(\sigma_{p+1} \wedge T, T-u)}}&\left(f\left(X_{\theta}^{(p, k, \texttt{x}, q, \ell, 2)}(T-s-u)\right) - f\left(X_{\theta}^{(p, k, \texttt{x}, q, \ell, 1)}(T-s-u)\right)\right)\textcolor[HTML]{D21312}{ds}\textcolor[HTML]{8728e2}{du}\\ &= \textcolor[HTML]{8728e2}{\int_{\sigma_q^{(p,k)}}^{\sigma_{q+1}^{(p,k)}\wedge (T - \sigma_p)}}\textcolor[HTML]{D21312}{\int_{0}^{T - \sigma_{p} - u}} \left(f\left(X_{\theta}^{(p, k, \texttt{x}, q, \ell, 2)}(r)\right) - f\left(X_{\theta}^{(p, k, \texttt{x}, q, \ell, 1)}(r)\right)\right)\textcolor[HTML]{D21312}{dr}\textcolor[HTML]{8728e2}{du}.
\end{split}
\end{align}
\normalsize

The double integral define various shapes which are represented graphically in figure~\ref{fig:graphical_second_order_difference}.
For convenience, we repeteadly use the Markov property of second-order trajectories to compute them.

\begin{figure*}[h!]
\centering
\includegraphics[width=\linewidth]{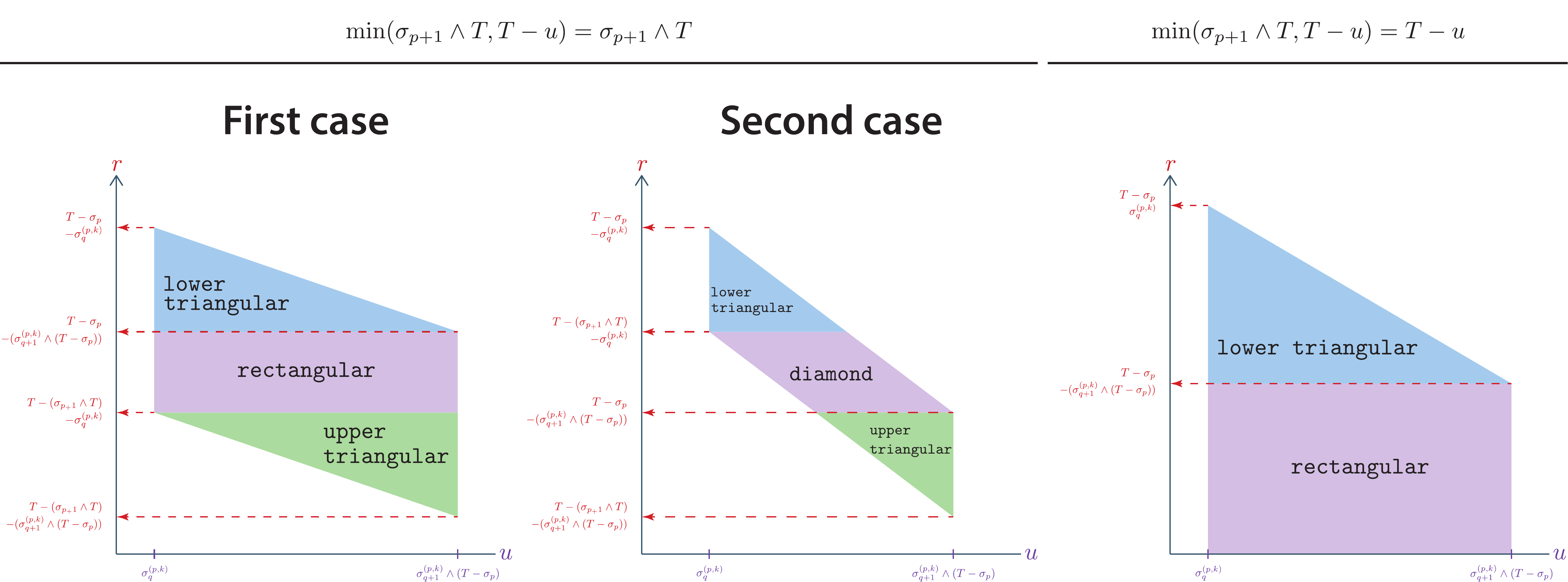}
\caption{\label{fig:graphical_second_order_difference} Integration domains over the second-order trajectories. The shape of the domain over which to integrate is dictated by (i)~the value of $\min(\sigma_{p+1} \wedge T, T-u)$ (see eq.~\eqref{eq:before_switching} and~\eqref{eq:after_switching}) and (ii)~the relative magnitude of $T-\sigma_p - (\sigma_{q+1}^{(p,k)} \wedge (T - \sigma_{p}))$ and $T - (\sigma_{p+1}\wedge T) -\sigma_q^{(p,k)}$. } 
\end{figure*}

\begin{algorithm}[H]
\caption{Compute the integral over of a second-order trajectory over the domain $\textcolor[HTML]{8728e2}{[\sigma_q^{(p,k)},\sigma_{q+1}^{(p,k)}\wedge (T - \sigma_p)]}$ $\times$ $\textcolor[HTML]{D21312}{[\sigma_{p},\min(\sigma_{p+1} \wedge T, T-u)]}$ in eq.~\eqref{eq:second_term_thm_final}}
\begin{algorithmic}[1]
\Function{SecondOrderDifference}{$x^{(2)}$, $\Tilde{x}^{(2)}$, $T^{(1)}$, $t^{(0)}$, $\Delta t^{(0)}$, $T^{(0)}$, $r_{\text{bott., left}}$, $r_{\text{top, left}}$, $r_{\text{bott., right}}$, $r_{\text{top, right}}$, $u_{\text{min}}$, $u_{\text{max}}$}
\State $\Delta G \leftarrow 0$
\If{$T^{(1)} = T^{(0)} - (t^{(0)} + \Delta t^{(0)})$} \Comment{$\min(\sigma_{p+1} \wedge T, T-u) = \sigma_{p+1} \wedge T$}
\If{$r_{\text{bott., left}} \leq r_{\text{top, right}}$} \Comment{first case in figure \ref{fig:graphical_second_order_difference}}
\State $r_{\min} \leftarrow 0$ , $\varphi^{(2)} \leftarrow r_{\text{bott., right}}$, $T^{(2)} \leftarrow r_{\text{bott., left}}$
\State \parbox[t]{440pt}{$(x^{(2)}$, $\Tilde{x}^{(2)},  \Delta G)$ $\leftarrow$ \Call{SecondOrderDifferenceSegment}{$x^{(2)}$, $\Tilde{x}^{(2)}$, $r_{\min}$, $\varphi^{(2)}$, $T^{(2)}$, $\Delta G$, $t^{(0)}$, $\Delta t^{(0)}$, $T^{(0)}$, $u_{\min}$, $u_{\max}$, \texttt{upper triangular}}\strut}
\State $r_{\min} \leftarrow r_{\text{bott., left}}$, $\varphi^{(2)} \leftarrow 0$, $T^{(2)} \leftarrow r_{\text{top, right}}$
\State \parbox[t]{440pt}{$(x^{(2)}$, $\Tilde{x}^{(2)},  \Delta G)$ $\leftarrow$ \Call{SecondOrderDifferenceSegment}{$x^{(2)}$, $\Tilde{x}^{(2)}$, $r_{\min}$, $\varphi^{(2)}$, $T^{(2)}$, $\Delta G$, $t^{(0)}$, $\Delta t^{(0)}$, $T^{(0)}$, $u_{\min}$, $u_{\max}$, \texttt{rectangular}}\strut}
\State $r_{\min} \leftarrow r_{\text{top, right}}$, $\varphi^{(2)} \leftarrow 0$, $T^{(2)} \leftarrow r_{\text{top, left}}$
\State \parbox[t]{440pt}{$(x^{(2)}$, $\Tilde{x}^{(2)},  \Delta G)$ $\leftarrow$ \Call{SecondOrderDifferenceSegment}{$x^{(2)}$, $\Tilde{x}^{(2)}$, $r_{\min}$, $\varphi^{(2)}$, $T^{(2)}$, $\Delta G$, $t^{(0)}$, $\Delta t^{(0)}$, $T^{(0)}$, $u_{\min}$, $u_{\max}$, \texttt{lower triangular}}\strut}
\ElsIf{$r_{\text{bott., left}} > r_{\text{top, right}}$} \Comment{second case in figure \ref{fig:graphical_second_order_difference}}
\State $r_{\min} \leftarrow 0$, $\varphi^{(2)} \leftarrow r_{\text{bott., right}}$, $T^{(2)} \leftarrow r_{\text{top, right}}$
\State \parbox[t]{440pt}{$(x^{(2)}$, $\Tilde{x}^{(2)},  \Delta G)$ $\leftarrow$ \Call{SecondOrderDifferenceSegment}{$x^{(2)}$, $\Tilde{x}^{(2)}$, $r_{\min}$, $\varphi^{(2)}$, $T^{(2)}$, $\Delta G$, $t^{(0)}$, $\Delta t^{(0)}$, $T^{(0)}$, $u_{\min}$, $u_{\max}$,  \texttt{upper triangular}}\strut}
\State $r_{\min} \leftarrow r_{\text{top, right}}$, $\varphi^{(2)} \leftarrow 0$, $T^{(2)} \leftarrow r_{\text{bott., left}}$
\State \parbox[t]{440pt}{$(x^{(2)}$, $\Tilde{x}^{(2)}, \Delta G)$ $\leftarrow$ \Call{SecondOrderDifferenceSegment}{$x^{(2)}$, $\Tilde{x}^{(2)}$, $r_{\min}$, $\varphi^{(2)}$, $T^{(2)}$, $\Delta G$, $t^{(0)}$, $\Delta t^{(0)}$, $T^{(0)}$, $u_{\min}$, $u_{\max}$, \texttt{diamond}}\strut}
\State $r_{\min} \leftarrow r_{\text{bott., left}}$, $\varphi^{(2)} \leftarrow 0$, $T^{(2)} \leftarrow r_{\text{top., left}}$
\State \parbox[t]{440pt}{$(x^{(2)}$, $\Tilde{x}^{(2)},  \Delta G)$ $\leftarrow$ \Call{SecondOrderDifferenceSegment}{$x^{(2)}$, $\Tilde{x}^{(2)}$, $r_{\min}$, $\varphi^{(2)}$, $T^{(2)}$, $\Delta G$, $t^{(0)}$, $\Delta t^{(0)}$, $T^{(0)}$, $u_{\min}$, $u_{\max}$, \texttt{lower triangular}}\strut}
\EndIf
\ElsIf{$T^{(1)} = T^{(0)} - t^{(0)}$} \Comment{$\min(\sigma_{p+1} \wedge T, T-u) = T-u$}
\State $r_{\min} \leftarrow 0$, $\varphi^{(2)} \leftarrow 0$, $T^{(2)} \leftarrow r_{\text{top, right}}$
\State \parbox[t]{460pt}{$(x^{(2)}$, $\Tilde{x}^{(2)},  \Delta G)$ $\leftarrow$ \Call{SecondOrderDifferenceSegment}{$x^{(2)}$, $\Tilde{x}^{(2)}$, $r_{\min}$, $\varphi^{(2)}$, $T^{(2)}$, $\Delta G$, $t^{(0)}$, $\Delta t^{(0)}$, $T^{(0)}$, $u_{\min}$, $u_{\max}$,  \texttt{rectangular}}\strut}
\State  $r_{\min} \leftarrow r_{\text{top, right}}$, $\varphi^{(2)} \leftarrow 0$, $T^{(2)} \leftarrow r_{\text{top, left}}$
\State \parbox[t]{460pt}{$(x^{(2)},\Tilde{x}^{(2)}, \Delta G) \leftarrow$ \Call{SecondOrderDifferenceSegment}{$x^{(2)}$, $\Tilde{x}^{(2)}$, $r_{\min}$, $\varphi^{(2)}$, $T^{(2)}$, $\Delta G$, $t^{(0)}$, $\Delta t^{(0)}$, $T^{(0)}$, $u_{\min}$, $u_{\max}$, \texttt{lower triangular}}\strut}
\EndIf
\State \Return $\Delta G$
\EndFunction
\end{algorithmic}
\end{algorithm}

\begin{algorithm}[H]
\caption{Compute the integral over part of a second-order trajectory in eq.~\eqref{eq:second_term_thm_final}}
\label{alg:second_order_integral}
\begin{algorithmic}[1]
\Function{SecondOrderDifferenceSegment}{$x^{(2)}$, $\Tilde{x}^{(2)}$, $r_{\min}$, $\varphi^{(2)}$, $T^{(2)}$, $\Delta G$, $t^{(0)}$, $\Delta t^{(0)}$, $T^{(0)}$, $u_{\min}$, $u_{\max}$, \texttt{domain}}
\State $H \leftarrow 0$
\State \tikzmark{highlight5-start}($x^{(2)}$, $\Tilde{x}^{(2)}$) $\leftarrow$ \Call{GenerateCoupledNRM}{$x^{(2)}$, $\Tilde{x}^{(2)}$, $\varphi^{(2)}$}
\State $t^{(2)} \leftarrow r_{\min} + \varphi^{(2)}$
\State $I_{o m} \leftarrow 0$ and $\tau_{om} \sim \mathcal{E}(1)$ for $o \in [\![1, M]\!]$ and $m \in [\![1, 3]\!]$
\While{$t^{(2)} < T^{(2)}$}
\State ($\Delta t^{(2)}$, $o^{*}$, $m^{*}$, $I$, $\tau$) $\leftarrow$ \Call{NextReactionCoupledNRM}{$x^{(2)}$, $\Tilde{x}^{(2)}$, $I$, $\tau$}
\State $\Delta t^{(2)} \leftarrow \min(\Delta t^{(2)}, T^{(2)}-t^{(2)})$\tikzmark{highlight5-end}
\State \parbox[t]{450pt}{$\Delta H \leftarrow$ \Call{DoubleIntegralDifference}{$x^{(2)}$, $\Tilde{x}^{(2)}$, $t^{(2)}$, $\Delta t^{(2)}$, $t^{(0)}$, $\Delta t^{(0)}$, $T^{(0)}$, $u_{\min}$, $u_{\max}$, \texttt{domain}}}
\State $H \leftarrow H + \Delta H$
\State \tikzmark{highlight6-start} $t^{(2)}$ $\leftarrow$ $t^{(2)} + \Delta t^{(2)}$
\If{$t^{(2)} \leq T^{(2)}$}
\State ($x^{(2)}$, $\Tilde{x}^{(2)}$) $\leftarrow$ \Call{UpdateCoupledState}{$x^{(2)}$, $\Tilde{x}^{(2)}$, $o^{*}$, $m^{*}$}
\EndIf
\EndWhile \tikzmark{highlight6-end}
\State $\Delta G \leftarrow \Delta G + H$
\State \Return $x^{(2)}$, $\Tilde{x}^{(2)}$, $\Delta G$
\EndFunction
\end{algorithmic}
\end{algorithm}

\highlight{highlight5}{red}{{-1cm,1em}}{{11.3cm,-0.3em}}
\highlight{highlight6}{red}{{-1.65cm,1em}}{{14.8cm,-0.3em}}

\begin{algorithm}[H]
\caption{Update the integral over part of a second-order trajectory in eq.~\eqref{eq:second_term_thm_final}}
\begin{algorithmic}[1]
\Function{DoubleIntegralDifference}{$x^{(2)}$, $\Tilde{x}^{(2)}$, $t^{(2)}$, $\Delta t^{(2)}$, $t^{(0)}$, $\Delta t^{(0)}$, $T^{(0)}$, $u_{\min}$, $u_{\max}$, \texttt{domain}}
\State $f_d \leftarrow f(\Tilde{x}^{(2)}) - f(x^{(2)})$
\If{\texttt{domain} \textbf{is} \texttt{upper triangular}}
\State $u^{*}_r \leftarrow T^{(0)} - (t^{(0)} + \Delta t^{(0)}) - t^{(2)}$
\State $u^{*}_l \leftarrow T^{(0)} - (t^{(0)} + \Delta t^{(0)}) - (t^{(2)} + \Delta t^{(2)})$
\State $U \leftarrow u_{\max}-u^{*}_r + (u^{*}_r - u^{*}_l)/2$
\ElsIf{\texttt{domain} \textbf{is} \texttt{rectangular}}
\State $U \leftarrow u_{\max}-u_{\min}$
\ElsIf{\texttt{domain} \textbf{is} \texttt{lower triangular}}
\State $u^{*}_r \leftarrow T^{(0)} - t^{(0)} - t^{(2)}$
\State $u^{*}_l \leftarrow T^{(0)} - t^{(0)} - (t^{(2)} + \Delta t^{(2)})$
\State $U \leftarrow u^{*}_l-u_{\min} + (u^{*}_r - u^{*}_l)/2$
\ElsIf{\texttt{domain} \textbf{is} \texttt{diamond}}
\State $u^{*}_{rb} \leftarrow T^{(0)} - (t^{(0)} + \Delta t^{(0)}) - t^{(2)}$
\State $u^{*}_{lb} \leftarrow T^{(0)} - (t^{(0)} + \Delta t^{(0)}) - (t^{(2)} + \Delta t^{(2)})$
\State $u^{*}_{rt} \leftarrow T^{(0)} - t^{(0)} - t^{(2)}$
\State $u^{*}_{lt} \leftarrow T^{(0)} - t^{(0)} - (t^{(2)} + \Delta t^{(2)})$
\State $U \leftarrow (u^{*}_{rb} - u^{*}_{lb})/ 2 + (u^{*}_{rt} - u^{*}_{lt}) / 2 + u^{*}_{lt} - u^{*}_{rb}$
\EndIf
\State $U \leftarrow f_d \times U \times \Delta t^{(2)}$
\State \Return $U$
\EndFunction
\end{algorithmic}
\end{algorithm}

\subsection{Generic methods to simulate (coupled) stochastic reaction networks}

For ease of exposition, we suppress the dependency of the propensity $\lambda$ on the system parameter $\theta$. $\mathcal{E}(\mu)$ is an exponential variable with parameter $\mu > 0$ and $\mathcal{B}(p_1,\dots, p_{M})$ a random variable which takes each value $k$ with probability $p_k$. 

\begin{algorithm}[H]
\caption{Select next reaction for SSA}
\begin{algorithmic}[1]
\Function{NextReactionSSA}{$x$}
\label{alg:next_reaction_ssa}
\State $\lambda_0 \leftarrow \sum_{k=1}^{M} \lambda_{k}(x)$
\If{$\lambda_0 = 0$}
\State $\Delta t \leftarrow \infty$ 
\ElsIf{$\lambda_0 > 0$}
\State $\Delta t \sim \mathcal{E}(\lambda_0)$
\State $p_k \leftarrow \lambda_k(x) / \lambda_0(x)$ for $k \in [\![1, M]\!]$
\State $k^{*} \sim \mathcal{B}(p_1,\dots, p_{M})$
\EndIf
\State \Return $\Delta t$, $k^{*}$
\EndFunction
\end{algorithmic}
\end{algorithm}

\begin{algorithm}[H]
\caption{Update state for SSA}
\begin{algorithmic}[1]
\Function{UpdateState}{$x$, $k^{*}$}
\State $x \leftarrow x + \zeta_{k^{*}}$
\State \Return $x$
\EndFunction
\end{algorithmic}
\end{algorithm}

\begin{algorithm}[H]
\caption{Select next reaction for coupled mNRM}
\begin{algorithmic}[1]
\Function{NextReactionCoupledNRM}{$x$, $\Tilde{x}$, $I$, $\tau$}
\State $\lambda_{k1} \leftarrow \lambda_k(x)$, $\lambda_{k2} \leftarrow \lambda_k(\Tilde{x})$, $\lambda_{k3} \leftarrow \lambda_k(x) \wedge \lambda_k(\Tilde{x})$ for $k \in [\![1, M]\!]$
\For{$k \in [\![1, M]\!]$ and $\ell \in [\![1, 3]\!]$}
\If{$\lambda_{k\ell} = 0$}
\State $\Delta t_{k\ell} \leftarrow \infty$
\Else
\State $\Delta t_{k\ell} \leftarrow (\tau_{k\ell} - I_{k\ell})/\lambda_{k\ell}$
\EndIf
\EndFor
\State $\Delta t \leftarrow \min_{k\ell} \Delta t_{k\ell}$
\State $(k^{*},\ell^{*}) \leftarrow \argmin_{k\ell} \Delta t_{k\ell}$
\State $I_{k\ell} \leftarrow \lambda_{k\ell} \times \Delta t$ for $k \in [\![1, M]\!]$ and $\ell \in [\![1, 3]\!]$
\State $\tau_{k^{*}\ell^{*}} \sim \mathcal{E}(1)$
\State \Return $\Delta t$, $k^{*}$, $\ell^{*}$, $I$, $\tau$
\EndFunction
\end{algorithmic}
\end{algorithm}

\begin{algorithm}[H]
\caption{Update state for coupled mNRM}
\begin{algorithmic}[1]
\Function{UpdateCoupledState}{$x$, $\Tilde{x}$, $k^{*}$, $\ell^{*}$}
\If{$\ell^{*} \in \{0,1\}$}
\State  $x\leftarrow x + \zeta_{k^{*}}$
\EndIf
\If{$\ell^{*} \in \{0,2\}$}
\State $\Tilde{x} \leftarrow \Tilde{x} + \zeta_{k^{*}}$
\EndIf
\State \Return $x$, $\Tilde{x}$
\EndFunction
\end{algorithmic}
\end{algorithm}

\begin{algorithm}[H]
\caption{Simulate two trajectories with the coupled mNRM}
\begin{algorithmic}[1]
\Function{GenerateCoupledNRM}{$x$, $\Tilde{x}$, $T$}
\State $t \leftarrow 0$
\State $I_{\ell m} \leftarrow 0$ and $\tau_{\ell m} \sim \mathcal{E}(1)$ for $\ell \in [\![1, M]\!]$ and $m \in [\![1, 3]\!]$
\While{$t \leq T$}
\State ($\Delta t$, $k^{*}$, $\ell^{*}$, $I$, $\tau$) $\leftarrow$ \Call{NextReactionCoupledNRM}{$x$, $\Tilde{x}$, $I$, $\tau$}
\State $t$ $\leftarrow$ $t + \Delta t$
\If{$t \leq T$}
\State ($x$, $\Tilde{x}$) $\leftarrow$ \Call{UpdateCoupledState}{$x$, $\Tilde{x}$, $k^{*}$, $\ell^{*}$}
\EndIf
\EndWhile
\State \Return $x$, $\Tilde{x}$
\EndFunction
\end{algorithmic}
\end{algorithm}

\section{Numerical examples}

In this section, we compare the performance of the estimator associated to the DBPA with the only unbiased alternative, which uses the GT method. While both are guaranteed to converge to the exact value as the number of samples increases, we nonetheless expect them to lead to estimates with different variances or mean-squared errors for a finite-number of samples. Given that the cost per sample, as measured by the average simulation time per sample, \emph{a priori} differs between both methods, we assess their performance in terms of the mean-squared error $\text{MSE}(t_{\text{comp}})$ defined as:
\begin{equation}
\label{eq:mc_mse}
\text{MSE}(t_{\text{comp}}) \coloneqq \sigma_{Y}^{2} / \bar n(t_{\text{comp}}),
\end{equation}

where $\sigma_{Y}^{2}$ is the variance of one sample from the DBPA or the GT method and  $\bar n(t_{\text{comp}})$ is the average number of samples generated within a given computational time budget $t_{\text{comp}}$. Introducing $\mathcal{M}_{0}:= t_{\text{comp}}/\bar{n}(t_{\text{comp}})$ the average simulation time per sample and $\mathcal{M} \coloneqq \mathcal{M}_{0}\sigma_{Y}^{2}$ the variance-adjusted cost per sample, we have:
\begin{equation}
\text{MSE}(t_{\text{comp}}) = \mathcal{M}/t_{\text{comp}}.
\end{equation}

This means that it is sufficient to report $\mathcal{M}$ to fully characterise the relative performance of the two unbiased methods for any computational time $t_{\text{comp}}$. In~\cite{gupta2018estimation}, it was observed that the chosen value of the upper bound on the number of auxiliary paths left the performance of the BPA for first-order sensitivities largely unaffected and was arbitrarily set to $10$ for all numerical experiments. Here, we use throughout the numerical examples $\mu_1 = 10$ as the upper bound for the expected number of desired pairs of first-order auxiliary processes and $\mu_2 = 100$ as the upper bound for the expected number of desired pairs of second-order auxiliary processes per trajectory of the main process once the number of first-order auxiliary processes is modulated. All simulations were run on a MacBook Pro with a 2 GHz Quad-Core Intel Core i5 processor. The corresponding scripts have been made available in a GitHub repository: \href{https://github.com/quentin-badolle/DoubleBPA}{{\tt github.com/quentin-badolle/DoubleBPA}}.

\begin{exmp}[Gene expression network]\label{ex:cge}
\end{exmp}

Let us first consider the constitutive gene expression network from~\cite{thattai2001intrinsic}:
\begin{equation}
\emptyset \xrightarrow[]{\theta_1} \mathbf{S_1},\quad \mathbf{S_1} \xrightarrow[]{\theta_2} \mathbf{S_1} + \mathbf{S_2},\quad \mathbf{S_1} \xrightarrow[]{\theta_3} \emptyset,\quad \mathbf{S_2} \xrightarrow[]{\theta_4} \emptyset,
\end{equation}

with mass-action kinetics:
\begin{equation}
\label{eq:prop_cge}
\lambda_1(x, \theta) = \theta_1,\quad \lambda_2(x, \theta) = \theta_2 x_1,\quad \lambda_3(x, \theta) = \theta_3 x_1,\quad \lambda_4(x, \theta) = \theta_4 x_2.
\end{equation}
The network consists of two species, where $\mathbf{S_1}$ corresponds to mRNA and $\mathbf{S_2}$ to protein. The first reaction represents transcription and the second translation. The last two reactions describe the degradation of mRNA and protein. As in~\cite{gupta2013unbiased}, we set $\theta_1 = 0.6$, $\theta_2 = 1.7329$, $\theta_3 = 0.3466$, $\theta_4 = 0.0023$. We choose the initial state to be $x=(0,0)$ and define the output function $f$ as $f(x) \coloneqq x_2$ which means that we will look at the second-order sensitivity of the mean protein abundance. On the left-hand side of Figure~\ref{fig:cge}, we provide estimates of $S_{\theta}^{(i,j)}(x,f,t)$ given by the DBPA and GT method for multiple pairs $(i,j)$ and a range of times $t$. Given that the estimator associated to the two approaches is unbiased, it is unsurprising that the estimates provided are close even for a finite number of samples, and that the error bars overlap. Each propensity in eq.~\eqref{eq:prop_cge} is linear in the abundance of at most one species and moments can therefore be computed in closed form, from which expressions for the sensitivities can be obtained without resorting to a computational method. The performance of the DBPA and GT method, as assessed by the value of their variance-adjusted cost per sample $\mathcal{M}$, is given on the right-hand side of Figure~\ref{fig:cge}. The results indicate that the DBPA can offer a substantial performance improvement over the GT method (up to 319-fold here).

\begin{exmp}[Genetic toggle switch]
\end{exmp}
We now move to the genetic toggle switch given in~\cite{gardner2000construction}:

\begin{equation}
\label{eq:ts}
\emptyset \xrightarrow[]{\lambda_1} \mathbf{S_1},\quad \mathbf{S_1} \xrightarrow[]{\lambda_2} \emptyset,\quad \emptyset \xrightarrow[]{\lambda_3} \mathbf{S_2},\quad \mathbf{S_2} \xrightarrow[]{\lambda_4} \emptyset,
\end{equation}

where the propensities are given by:
\begin{equation}
\lambda_1(x, \theta) = \frac{\theta_1}{\theta_2 + x_2^{\theta_3}} + \theta_4,\quad \lambda_2(x,\theta) = \theta_5 x_1,\quad \lambda_3(x, \theta) = \frac{\theta_6}{\theta_7 + x_1^{\theta_8}} + \theta_9,\quad \lambda_4(x,\theta) = \theta_{10} x_2.
\end{equation}

In this network, two proteins $\mathbf{S_1}$ and $\mathbf{S_2}$ inhibit each other's production non-linearly through the propensities $\lambda_1$ and $\lambda_2$. The rates $\theta_4$ and $\theta_8$ correspond to constitutive expression of the proteins. In contrast to the previous model, transcription and translation are lumped in a single step here. The propensities $\lambda_2$ and $\lambda_4$ account for the degradation of $\mathbf{S_1}$ and $\mathbf{S_2}$. We set $\theta_1 = 1.0$, $\theta_2 = 1.0$, $\theta_3 = 2.5$, $\theta_4 = 0.5$, $\theta_5 = 0.0023$, $\theta_6 = 1.0$, $\theta_7 = 1.0$, $\theta_8 = 1.0$, $\theta_9 = 0.5$, $\theta_{10} = 0.0023$. In eq.~\eqref{eq:ts}, $\mathbf{S_1}$ and $\mathbf{S_2}$ play symmetric roles and we choose $f$ as $f(x) \coloneqq x_1$.  Again, we use $x = (0,0)$ as the initial state and consider the estimation of $S_{\theta}^{(i,j)}(x,f,t)$ by the DBPA and GT method for multiple pairs $(i,j)$ and a range of times $t$ in Figure~\ref{fig:ts}. In contrast to the previous example, numerical estimation is currently the only possible approach to compute the sensitivities now. Like in the previous example, while the estimates provided by both methods are in close agreement, the DBPA comes with a markedly lower variance-adjusted cost per sample $\mathcal{M}$ (up to 722-fold here).

\begin{exmp}[Antithetic integral controller]
\end{exmp}
We finally consider a gene expression network under the control of the antithetic integral controller from~\cite{briat2016antithetic, aoki2019universal}:
\begin{equation}
\mathbf{S_3} \xrightarrow[]{\lambda_1} \mathbf{S_1} + \mathbf{S_3},\text{ }
\mathbf{S_1} \xrightarrow[]{\lambda_2} \mathbf{S_1} + \mathbf{S_2},\text{ }
\mathbf{S_1} \xrightarrow[]{\lambda_3} \emptyset,\text{ }
\mathbf{S_2} \xrightarrow[]{\lambda_4} \emptyset,\text{ }
\mathbf{S_2} \xrightarrow[]{\lambda_5} \mathbf{S_2} + \mathbf{S_4},\text{ }
\mathbf{S_3} + \mathbf{S_4} \xrightarrow[]{\lambda_6} \emptyset,\text{ }
\emptyset \xrightarrow[]{\lambda_7} \mathbf{S_3},
\end{equation}

where:
\begin{align}
\lambda_1(x,\theta) &= \theta_1 x_3, & \lambda_2(x,\theta) &= \theta_2 x_1,\\
\lambda_3(x,\theta) &= \theta_3 x_1, & \lambda_4(x,\theta) &= \theta_4 x_2,\\
\lambda_5(x,\theta) &= \theta_5 x_2, & \lambda_6(x,\theta) &= \theta_6 x_3 x_4,\\
\lambda_7(x,\theta) &= \theta_7.
\end{align}

Here, the abundance of the protein $\mathbf{S_2}$ is controlled by the antithetic integral controller involving the regulatory species $\mathbf{S_3}$ and $\mathbf{S_4}$. The transcription rate of the mRNA $\mathbf{S_1}$ is proportional to the abundance of $\mathbf{S_3}$ (actuation reaction) while the production rate of $\mathbf{S_4}$ is proportional to the abundance of the protein $\mathbf{S_2}$ (sensing reaction). Like in example \ref{ex:cge}, the propensity $\lambda_2$ corresponds to the translation of the mRNA $\mathbf{S_1}$ into the protein $\mathbf{S_2}$ and the propensities $\lambda_3$ and $\lambda_4$ to their degradation. $\mathbf{S_3}$ and $\mathbf{S_4}$ are jointly degraded in an annihilation reaction with propensity $\lambda_6$. We set $\theta_1 = 0.6$, $\theta_2 = 1.7329$, $\theta_3 = 0.3466$, $\theta_4 = 0.0023$, $\theta_5 = 0.1$, $\theta_6 = 10.0$, $\theta_7 = 0.5$ and choose the initial state to be $x=(0,0,0,0)$. Regulated variables like glucose levels are omnipresent in biological processes and the antithetic integral controller has been suggested as a synthetic way to specify some characteristics of living cells including in therapeutic applications~\cite{filo2023biomolecular}. Both the abundance of the protein $\mathbf{S_2}$ and its variability are therefore of practical relevance and we report on the performance of the DBPA and GT method for $f(x)=x_2$ and $f(x)=(x_2)^2$. When $f(x)=x_2$, it is known that the stationary mean equals $\theta_5 / \theta_7$~\cite{briat2016antithetic}. Approximations for the stationary variance have also been obtained~\cite{briat2018antithetic}. In the meantime, no expression is known for the finite-time mean and variance. Computational methods to estimate second-order sensitivities  therefore emerge as a valuable tool to characterise the transient response of the network. Results are provided in Figure~\ref{fig:aic} and the conclusions follow those from the previous two examples. In particular, the performance displayed in panel D of Figure~\ref{fig:aic} corresponds to a $364$-fold improvement of the DBPA over the GT method.\newline

\begin{figure}[htbp]
    \centering
    \begin{subfigure}[b]{\textwidth}
    \centering
    \subfloat[\centering$S_{\theta}^{(1,4)}(x,f,t)$]{
        \includegraphics[width=0.46\textwidth]{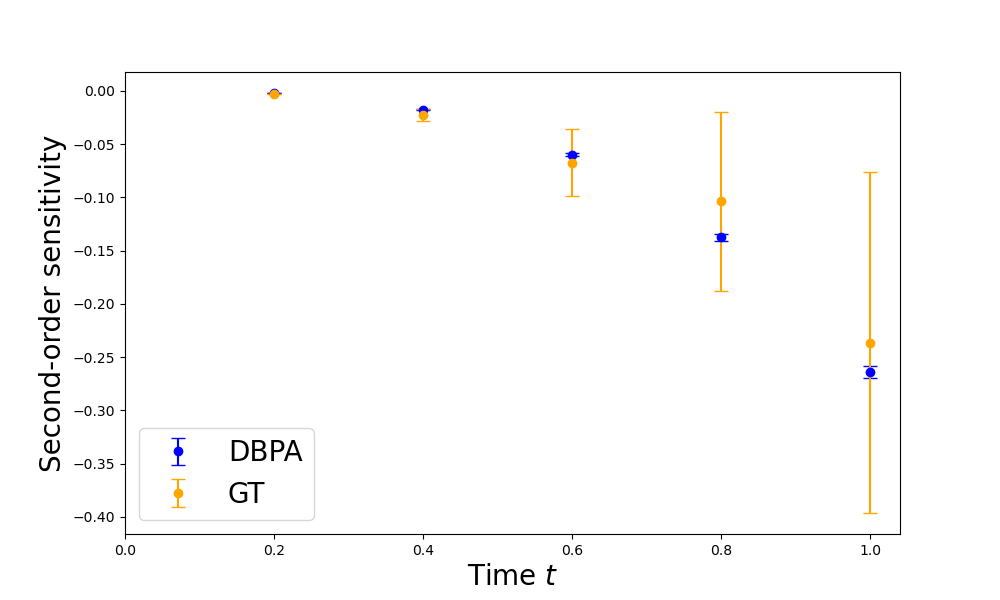}
        \label{fig:sub1}
    }
    \subfloat[\centering$\mathcal{M}$ for $S_{\theta}^{(1,4)}(x,f,t)$]{
        \includegraphics[width=0.46\textwidth]{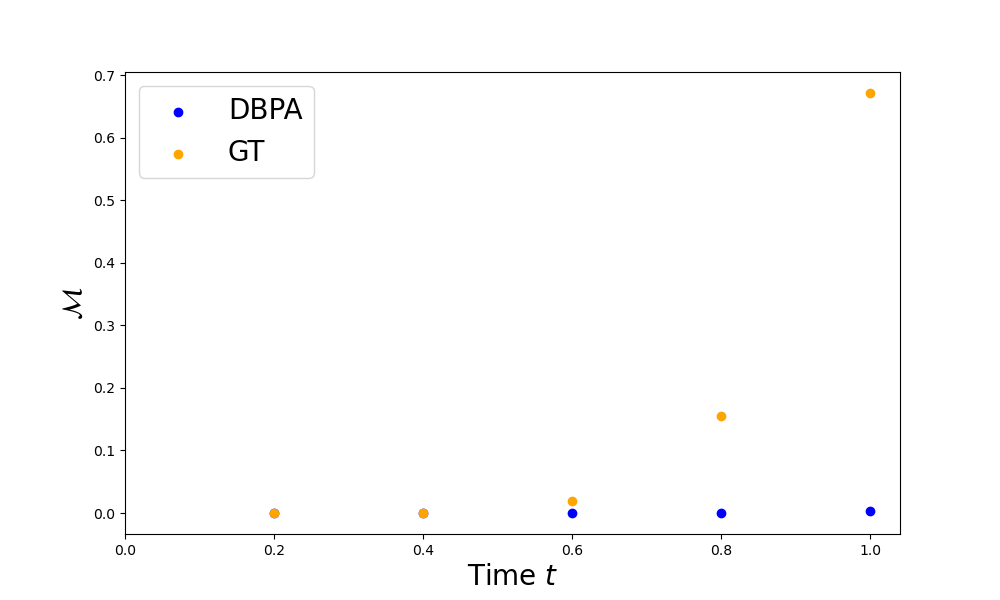}
        \label{fig:sub2}
    }
    \end{subfigure}
    \begin{subfigure}[b]{\textwidth}
    \centering
    \subfloat[\centering$S_{\theta}^{(2,4)}(x,f,t)$]{
        \includegraphics[width=0.46\textwidth]{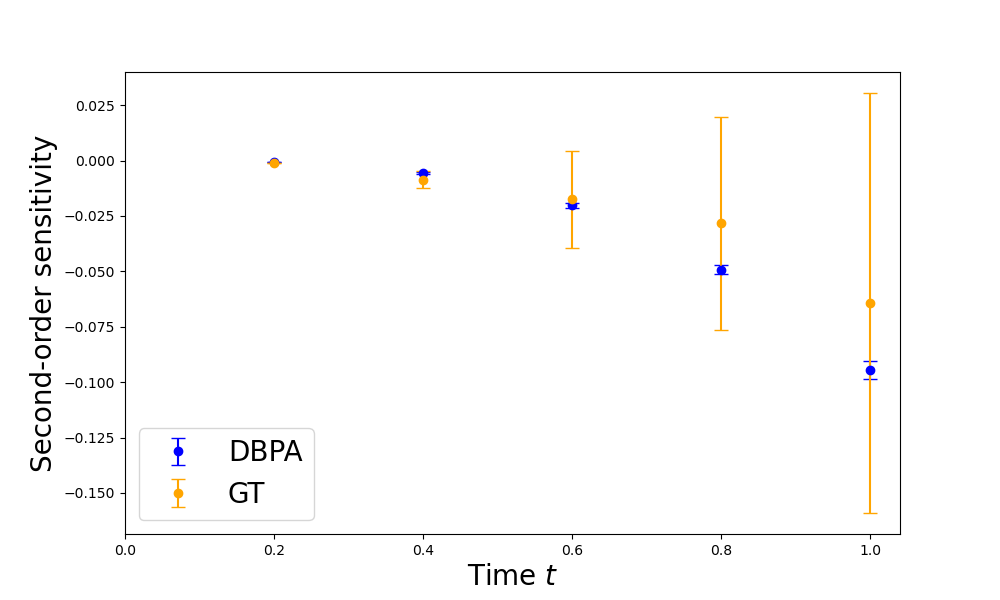}
        \label{fig:sub3}
    }
    \subfloat[\centering$\mathcal{M}$ for $S_{\theta}^{(2,4)}(x,f,t)$]{
        \includegraphics[width=0.46\textwidth]{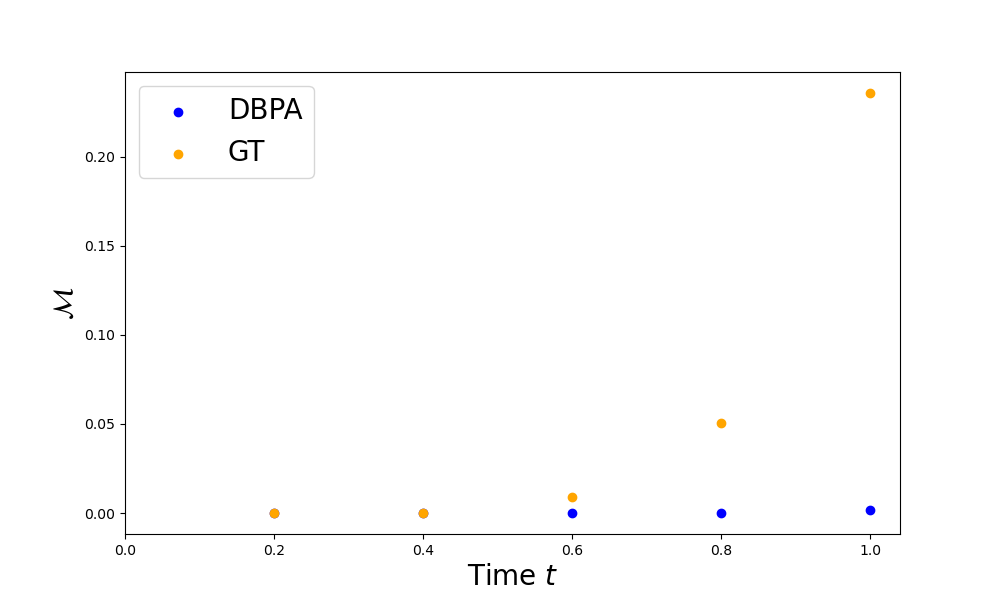}
        \label{fig:sub4}
    }
    \end{subfigure}
    \begin{subfigure}[b]{\textwidth}
    \centering
    \subfloat[\centering$S_{\theta}^{(3,4)}(x,f,t)$]{
        \includegraphics[width=0.46\textwidth]{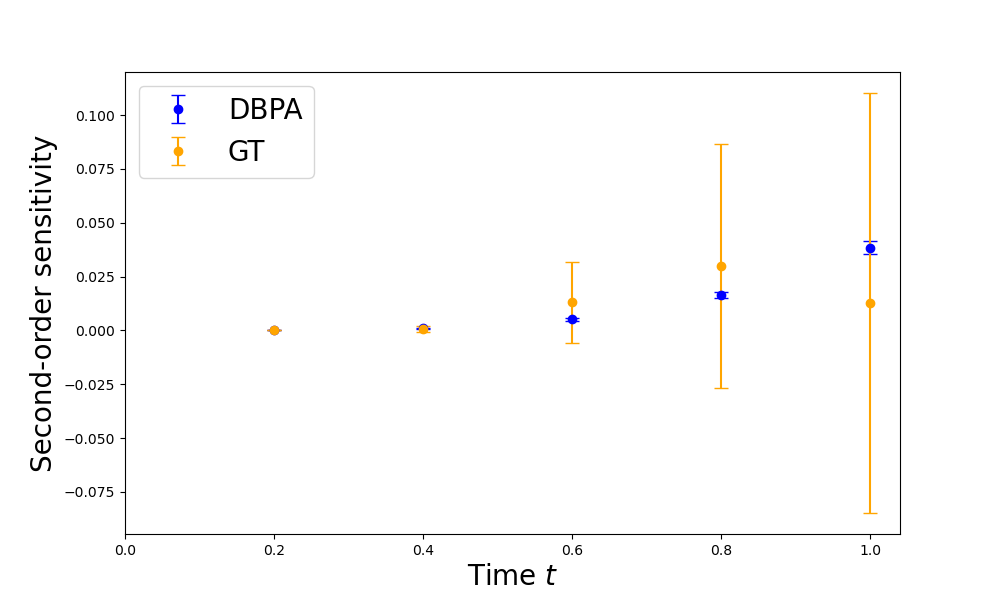}
        \label{fig:sub5}
    }
    \subfloat[\centering$\mathcal{M}$ for $S_{\theta}^{(3,4)}(x,f,t)$]{
        \includegraphics[width=0.46\textwidth]{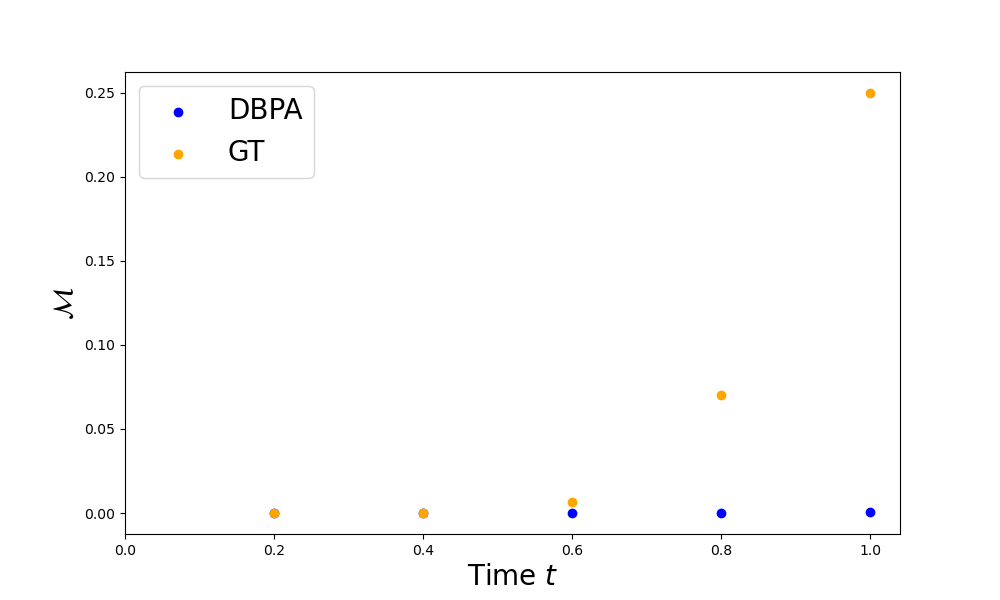}
        \label{fig:sub6}
    }
    \end{subfigure}
    \caption{Gene expression network. The sensitivity $S_{\theta}^{(i,j)}(x,f,t)$ is computed using $10^4$ DBPA simulations and $5\times10^5$ GT simulations. The parameters of the network are set to $\theta_1 = 0.6$, $\theta_2 = 1.7329$, $\theta_3 = 0.3466$, $\theta_4 = 0.0023$. The initial state is chosen to be $x=(0,0)$ and the output function $f$ to be $f(x) \coloneqq x_2$. In the panels on left-hand side, error bars correspond to two standard deviations. While both methods are unbiased and will lead to the same sensitivity estimate asymptotically, the DBPA can offer large performance improvements over the GT method by having a lower variance-adjusted cost per sample $\mathcal{M}$, as showcased in the panels on the right-hand side.}
    \label{fig:cge}
\end{figure}


\begin{figure}[htbp]
    \centering
    \begin{subfigure}[b]{\textwidth}
    \centering
    \subfloat[\centering$S_{\theta}^{(1,5)}(x,f,t)$]{
        \includegraphics[width=0.46\textwidth]{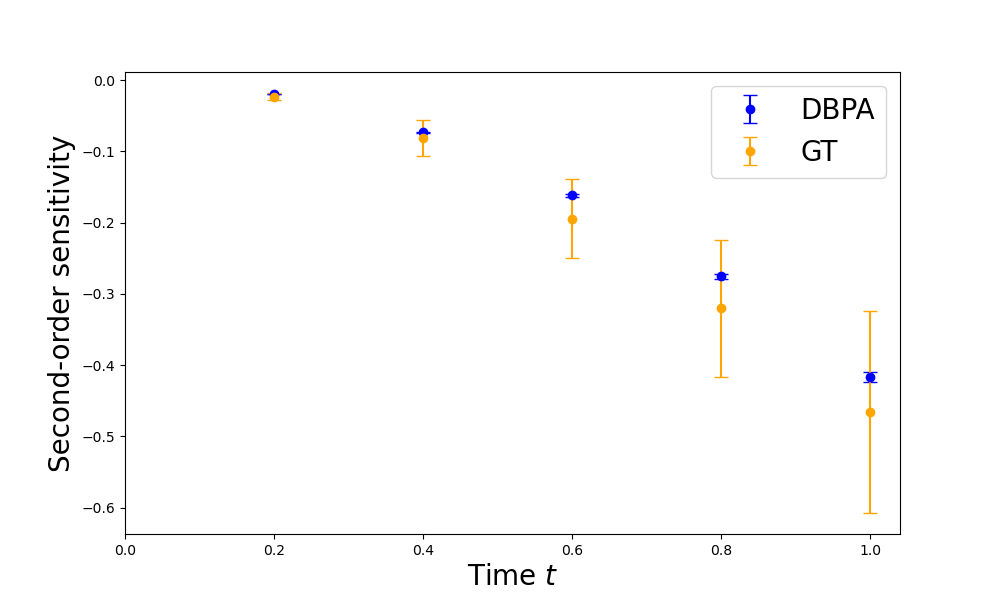}
        \label{fig:sub1}
    }
    \subfloat[\centering$\mathcal{M}$ for $S_{\theta}^{(1,5)}(x,f,t)$]{
        \includegraphics[width=0.46\textwidth]{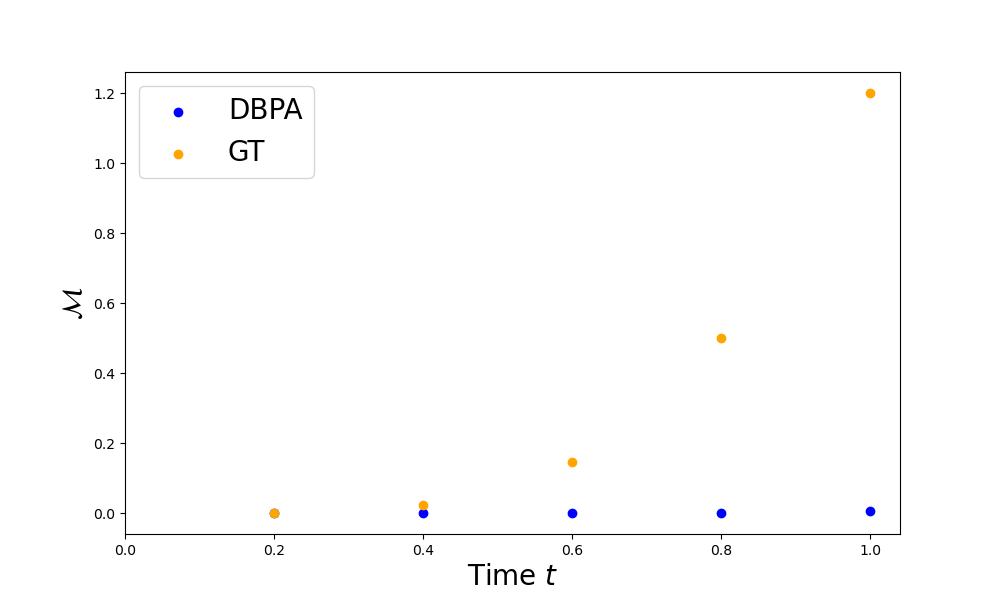}
        \label{fig:sub2}
    }
    \end{subfigure}
    \begin{subfigure}[b]{\textwidth}
    \centering
    \subfloat[\centering$S_{\theta}^{(1,10)}(x,f,t)$]{
        \includegraphics[width=0.46\textwidth]{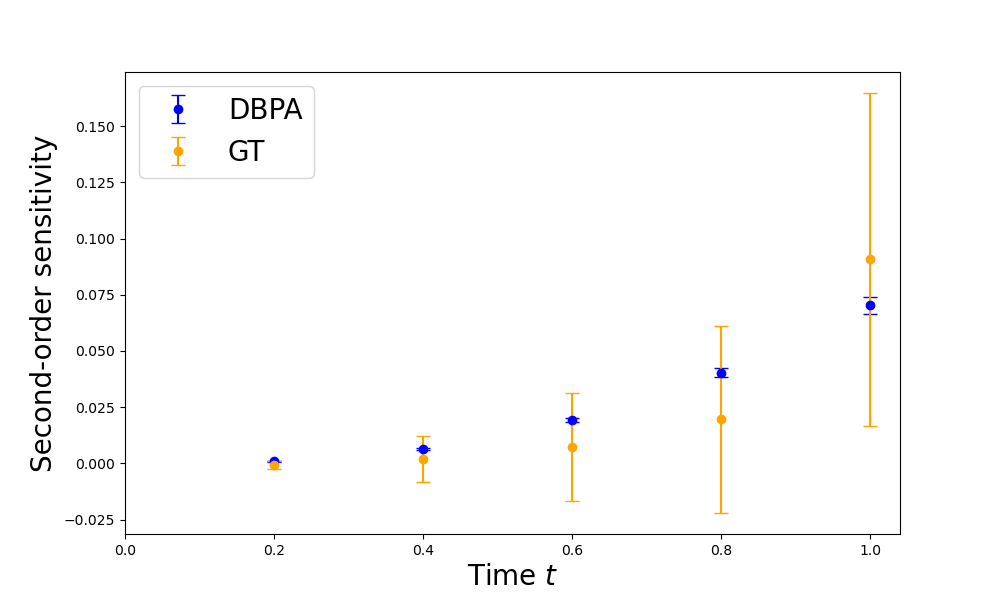}
        \label{fig:sub3}
    }
    \subfloat[\centering$\mathcal{M}$ for $S_{\theta}^{(1,10)}(x,f,t)$]{
        \includegraphics[width=0.46\textwidth]{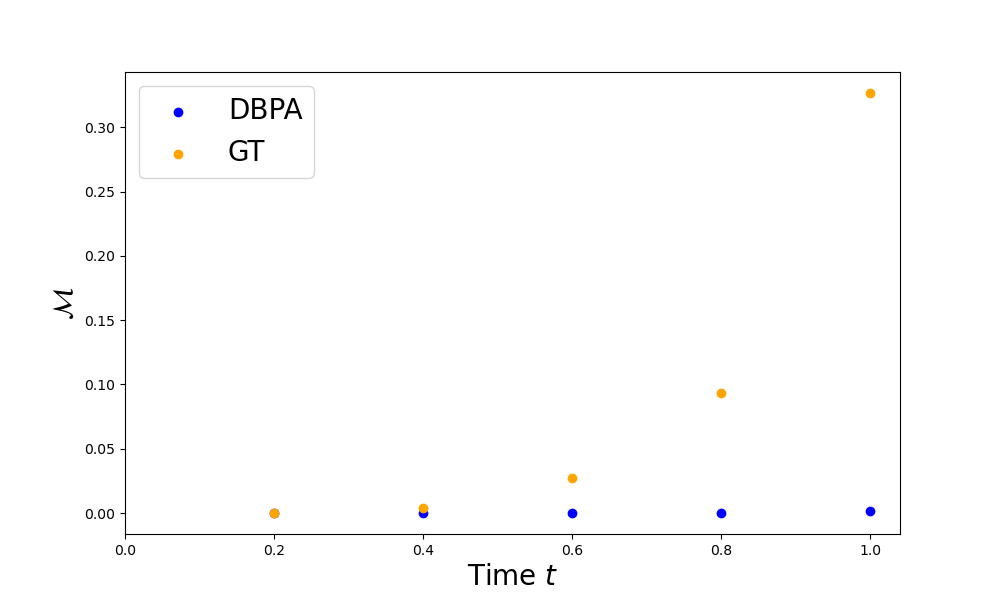}
        \label{fig:sub4}
    }
    \end{subfigure}
    \begin{subfigure}[b]{\textwidth}
    \centering
    \subfloat[\centering$S_{\theta}^{(4,8)}(x,f,t)$]{
        \includegraphics[width=0.46\textwidth]{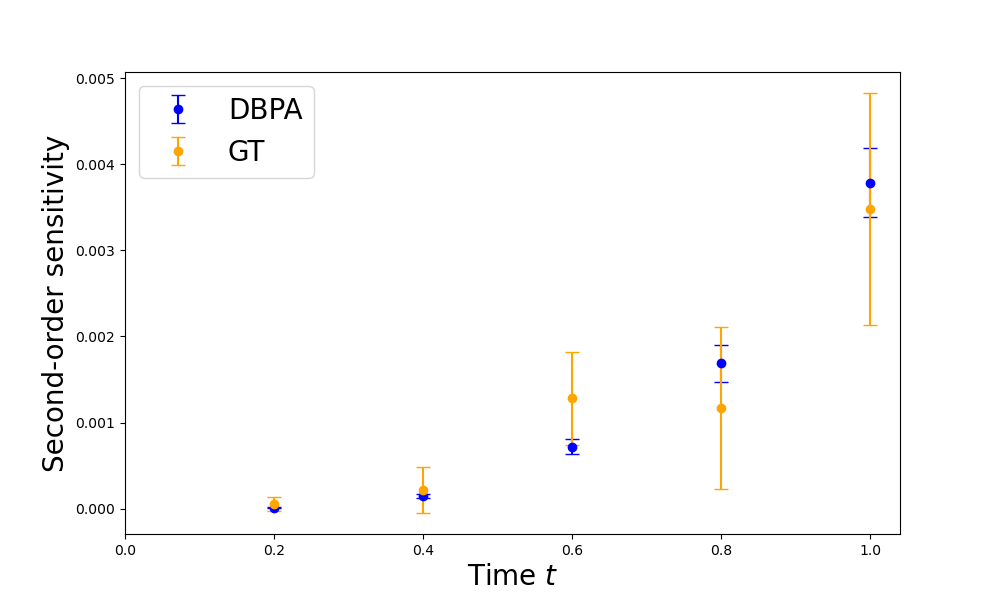}
        \label{fig:sub5}
    }
    \subfloat[\centering$\mathcal{M}$ for $S_{\theta}^{(4,8)}(x,f,t)$]{
        \includegraphics[width=0.46\textwidth]{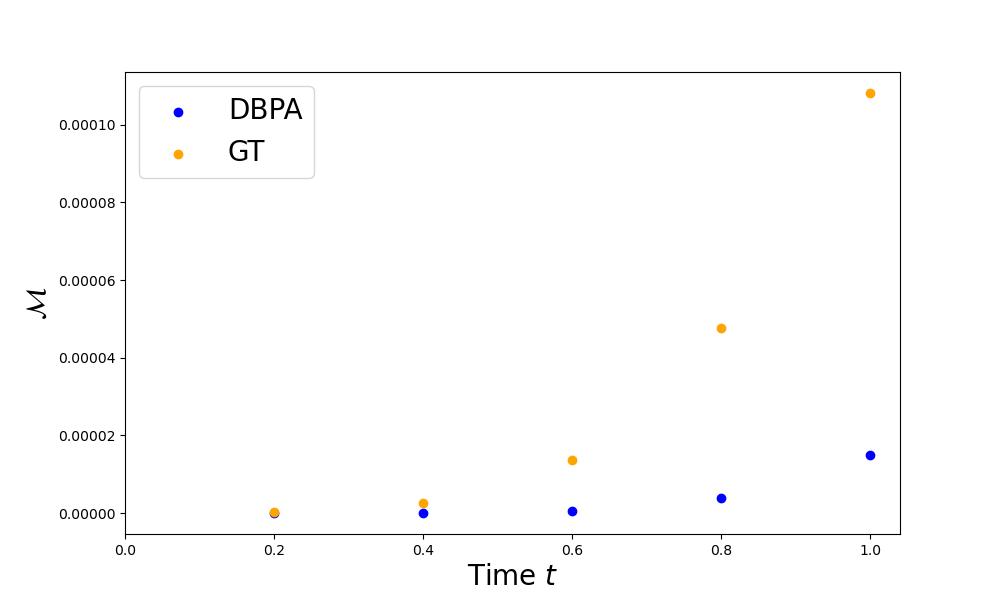}
        \label{fig:sub6}
    }
    \end{subfigure}
    \caption{Toggle switch network. The sensitivity $S_{\theta}^{(i,j)}(x,f,t)$ is computed using $10^4$ DBPA simulations and $5\times10^5$ GT simulations. The parameters of the network are set to $\theta_1 = 1.0$, $\theta_2 = 1.0$, $\theta_3 = 2.5$, $\theta_4 = 0.5$, $\theta_5 = 0.0023$, $\theta_6 = 1.0$, $\theta_7 = 1.0$, $\theta_8 = 1.0$, $\theta_9 = 0.5$, $\theta_{10} = 0.0023$. The initial state is chosen to be $x=(0,0)$ and the output function $f$ to be $f(x) \coloneqq x_1$. In the panels on left-hand side, error bars correspond to two standard deviations. While both methods are unbiased and will lead to the same sensitivity estimate asymptotically, the DBPA can offer large performance improvements over the GT method by having a lower variance-adjusted cost per sample $\mathcal{M}$, as showcased in the panels on the right-hand side.}
    \label{fig:ts}
\end{figure}


\begin{figure}[htbp]
    \centering
    \begin{subfigure}[b]{\textwidth}
    \centering
    \subfloat[\centering$S_{\theta}^{(1,5)}(x,f,t)$ with $f(x)=x_2$]{
        \includegraphics[width=0.46\textwidth]{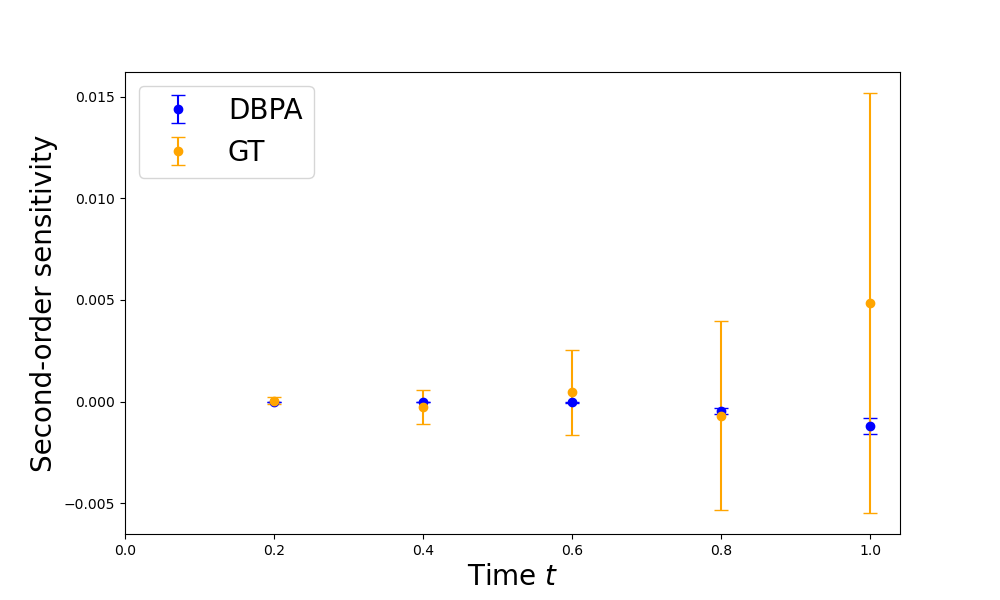}
        \label{fig:sub1}
    }
    \subfloat[\centering$\mathcal{M}$ for $S_{\theta}^{(1,5)}(x,f,t)$ with $f(x)=x_2$]{
        \includegraphics[width=0.46\textwidth]{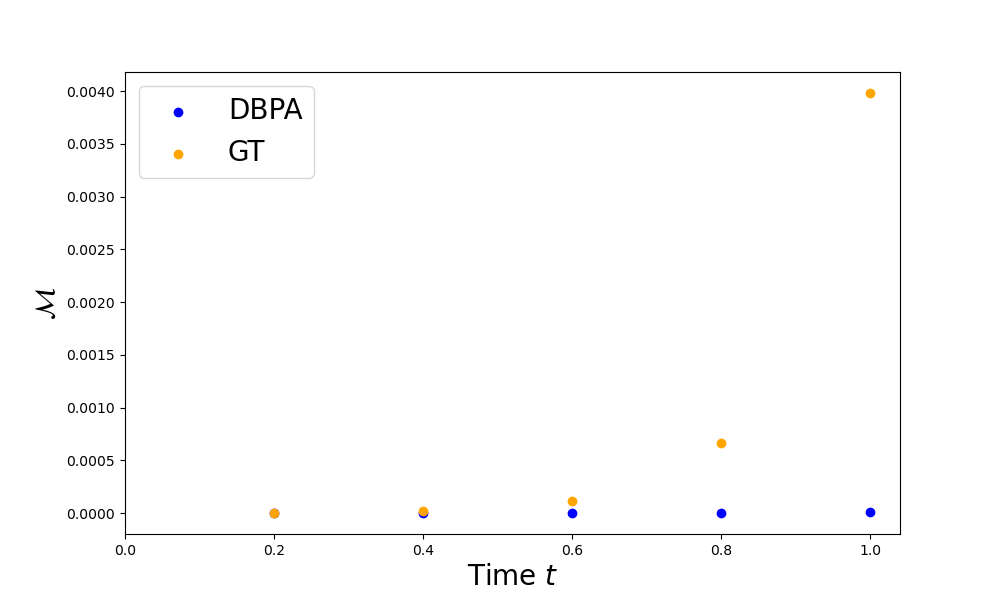}
        \label{fig:sub2}
    }
    \end{subfigure}
    \begin{subfigure}[b]{\textwidth}
    \centering
    \subfloat[\centering$S_{\theta}^{(5,7)}(x,f,t)$ with $f(x)=x_2$]{
        \includegraphics[width=0.46\textwidth]{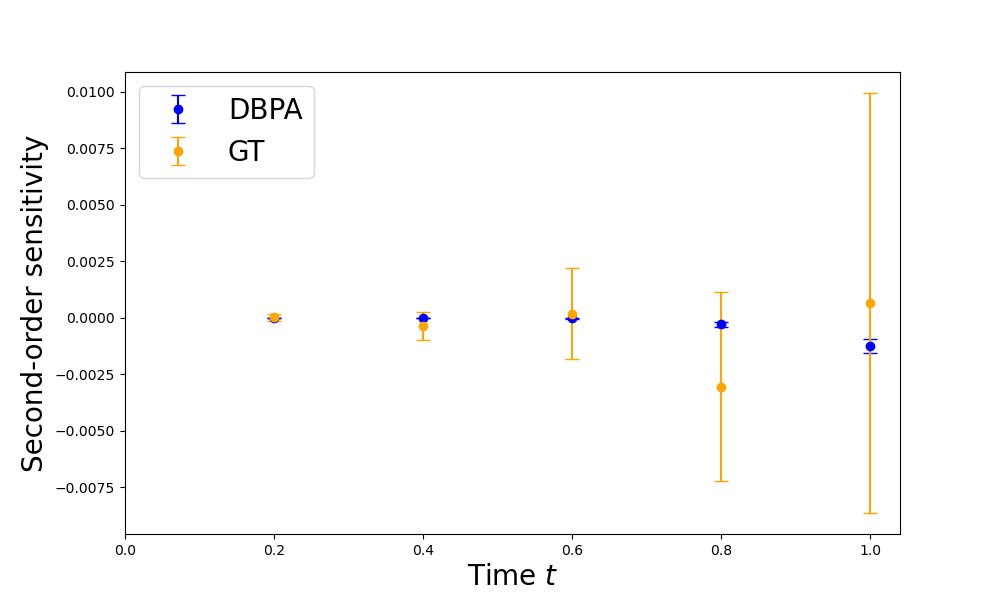}
        \label{fig:sub3}
    }
    \subfloat[\centering$\mathcal{M}$ for $S_{\theta}^{(5,7)}(x,f,t)$ with $f(x)=x_2$]{
        \includegraphics[width=0.46\textwidth]{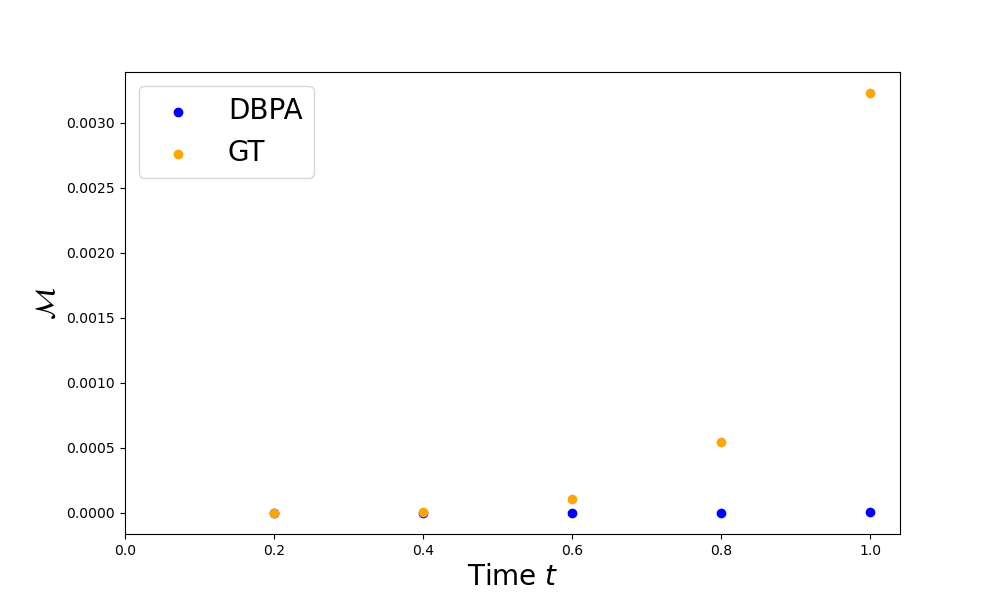}
        \label{fig:sub4}
    }
    \end{subfigure}
    \begin{subfigure}[b]{\textwidth}
    \centering
    \subfloat[\centering$S_{\theta}^{(5,7)}(x,f,t)$ with $f(x)=(x_2)^2$]{
        \includegraphics[width=0.46\textwidth]{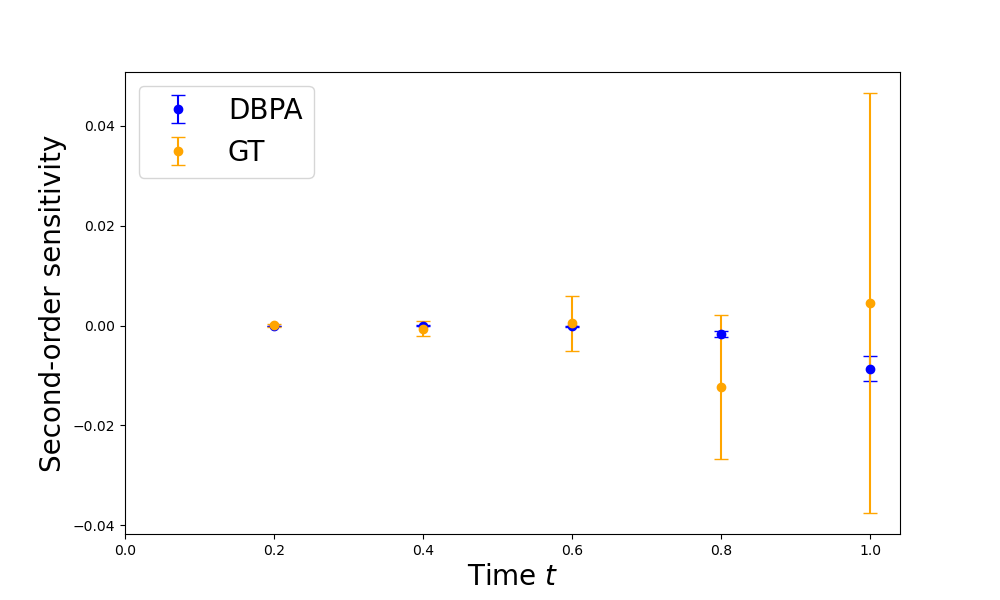}
        \label{fig:sub5}
    }
    \subfloat[\centering$\mathcal{M}$ for $S_{\theta}^{(5,7)}(x,f,t)$ with $f(x)=(x_2)^2$]{
        \includegraphics[width=0.46\textwidth]{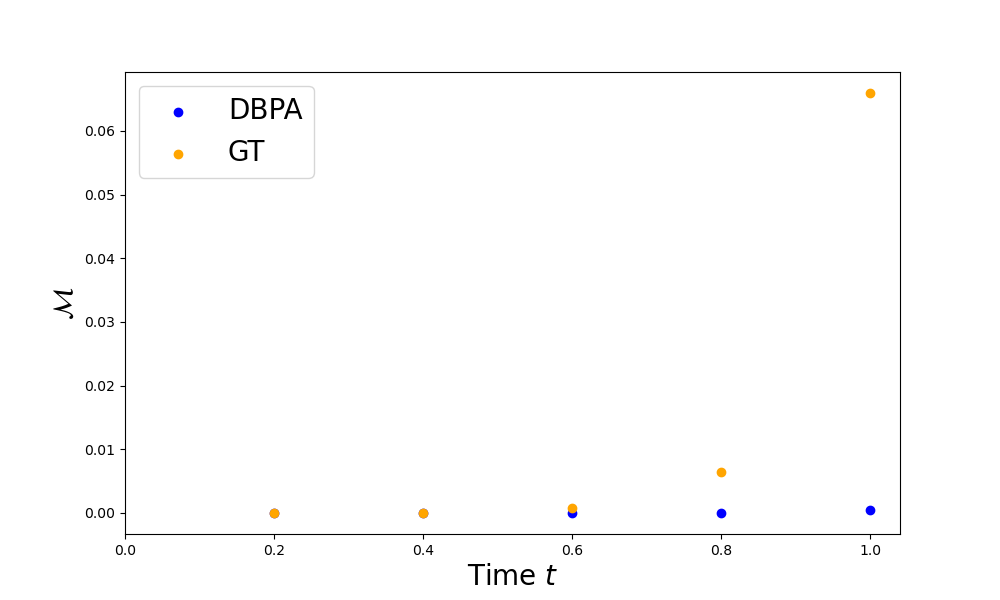}
        \label{fig:sub6}
    }
    \end{subfigure}
    \caption{Antithetic integral controller. The sensitivity $S_{\theta}^{(i,j)}(x,f,t)$ is computed using $10^4$ DBPA simulations and $5\times10^5$ GT simulations. The parameters of the network are set to $\theta_1 = 1.0$, $\theta_2 = 1.0$, $\theta_3 = 2.5$, $\theta_4 = 0.5$, $\theta_5 = 0.0023$, $\theta_6 = 1.0$, $\theta_7 = 1.0$, $\theta_8 = 1.0$, $\theta_9 = 0.5$, $\theta_{10} = 0.0023$. The initial state is chosen to be $x=(0,0,0,0)$. In the panels on left-hand side, error bars correspond to two standard deviations. While both methods are unbiased and will lead to the same sensitivity estimate asymptotically, the DBPA can offer large performance improvements over the GT method by having a lower variance-adjusted cost per sample $\mathcal{M}$, as showcased in the panels on the right-hand side.}
    \label{fig:aic}
\end{figure}
\section{Summary and Outlook}

In applications related to systems and synthetic biology, parameter sensitivities are key to characterise and infer SRNs. In this work, we provided conditions on the existence of second-order sensitivities of an expected output of interest. We also derived a new integral representation for these second-order derivatives. Based on this formula, we introduced the Double BPA which can generate unbiased samples of the Hessian of an average output which takes the form of the expected value of some function of the state at some time. We illustrate on numerical examples that the Double BPA can provide a substantial performance improvement over the only unbiased alternative previously available which is based on the GT method.

Like the GT method, the Double BPA relies on exact simulations of the reaction dynamics. These simulations can be computationally demanding, especially in the presence of time-scale separation. In the spirit of ref.~\cite{gupta2018estimation}, we could extend the approach developed here to rely instead on approximate simulations like those from the tau-leap algorithms, using again eq.~\eqref{eq:thm} as the starting point. Relying on these approximate simulations would however come at the price of introducing a bias. Integrating multi-level strategies as those from ref.~\cite{anderson2012multilevel} would be another way to improve the computational efficiency of our approach. In ref.~\cite{gupta2021deepcme}, the equivalent of eq.~\eqref{eq:thm} for first-order sensitivities was used in the DeepCME framework (see eq.~\eqref{eq:first_order_bpa} for the formula). In this hybrid Deep Learning-Monte Carlo approach, the expected output $\Psi_{\theta}(x,f,t)$ was first approximated using a neural network and this surrogate was used instead of the simulation of auxiliary paths when evaluating the integral over paths. Similarly, we envision leveraging eq.~\eqref{eq:thm} together with a neural approximation of $\Psi_{\theta}(x,f,t)$ and its first-order sensitivities to avoid the simulation of first- and second-order auxiliary paths altogether. We leave the extension of eq.~\eqref{eq:thm} to sensitivities of arbitrary order and to stationary second-order sensitivities as an immediate continuation of the present work.

\section*{Acknowledgments}

This research was funded in whole or in part by the Swiss National Science Foundation (SNSF) grant no. 182653 and 216505.

\appendix
\section{Appendix}
\label{appendix}

In this appendix we prove our main result theorem~\ref{thm:main_theorem}. We will need the following lemma:

\begin{lemma}[From the proof of theorem 3.3 in~\cite{gupta2018estimation}]
\label{lemma:for_second_part}
Pick $T \in \Rplus$. Let $(X_{\theta}(t))$ be the Markov process with generator $\mathbb{A}_{\theta}$ defined in eq.~\eqref{eq:def_generator}. Suppose the propensity function $\lambda$ satisfies assumptions \ref{assumption:propensity_dynkin} and \ref{assumption:propensity_regularity} at $\theta$ for $(i,j)$. Introduce  the function $\lambda_0 : \Natural^N \times \mathbb{R}^p\xrightarrow{} \Rplus$ denoting the sum of propensities: $\lambda_0(x,\theta) \coloneqq \sum_{k=1}^{M} \lambda_{k}(x,\theta)$. Let $\sigma_{l}$ be the $l$-th jump time of the process $(X_{\theta}(t))$ for $l \in \Natural$ with $\sigma_{0}=0$ for convenience. Introduce the inter-jump time $\delta_{l}\coloneqq(\sigma_{l+1} \wedge T) - (\sigma_{l} \wedge T)$ for $l \in \Natural$. For any continuous function $g : [0,\infty) \rightarrow [0,\infty)$, we have:

\begin{equation}
\label{eq:conditioned_delta}
\mathbb{E} \Bigg[\int_{0}^{\delta_{l}}g(s)ds\Bigg| X_{\theta}(\sigma_{l})=y,\sigma_{l}=u\Bigg]=\int_{0}^{T-u}g(s)e^{-\lambda_{0}(y,\theta)s}ds, \quad \forall y \in \mathbb{N}^{N}, \quad \forall u <T.
\end{equation}
\end{lemma}

\begin{proof}
Given $X_{\theta}(\sigma_{l})=y$ and $\sigma_{l}=u$, the cumulative distribution function of $\delta_{l}$ is:

$$
P(\delta_{l}<s|X_{\theta}(\sigma_{l})=y,\sigma_{l}=u) = \left\{
\begin{array}{rr}
0, & \text{if } s<0\\
1 - e^{-\lambda_{0}(y,\theta)s}, & \text{if } 0 \leq s < (T-u)\\
1, & \text{otherwise}.
\end{array}\right. 
$$

Now we can compute:

\begin{equation}
\label{eq:ipp_intermediate}
\mathbb{E} \Bigg[\int_{0}^{\delta_{l}}g(s)ds\Bigg| X_{\theta}(\sigma_{l})=y,\sigma_{l}=u\Bigg]= \int_{0}^{T-u}\Bigg(\int_{0}^{s}g(t)dt\Bigg)\lambda_{0}(y,\theta)e^{-\lambda_{0}(y,\theta)s}ds+\Bigg(\int_{0}^{T-u}g(t)dt\Bigg)e^{-\lambda_{0}(y,\theta)(T-u)}. 
\end{equation}

By integration by parts, observe that:

\begin{align}
\int_{0}^{T-u}\Bigg(\int_{0}^{s}g(t)dt\Bigg)\lambda_{0}(y,\theta)e^{-\lambda_{0}(y,\theta)s}ds
&= -\Bigg[\Bigg(\int_{0}^{s}g(t)dt\Bigg)e^{-\lambda_{0}(y,\theta)s}\Bigg]_{0}^{T-u} + \int_{0}^{T-u}g(s)e^{-\lambda_{0}(y,\theta)s}ds \nonumber\\
&= -\Bigg(\int_{0}^{T-u}g(s)ds\Bigg)e^{-\lambda_{0}(y,\theta)(T-u)} + \int_{0}^{T-u}g(s)e^{-\lambda_{0}(y,\theta)s}ds.\label{eq:ipp}
\end{align}

Replacing eq.~\eqref{eq:ipp} in eq.~\eqref{eq:ipp_intermediate}, we obtain eq.~\eqref{eq:conditioned_delta} as desired. 
\end{proof}

\begin{proof}[Proof of theorem~\ref{thm:main_theorem}]

The proof builds on ideas introduced in~\cite{anderson2015stochastic, gupta2013unbiased, gupta2018estimation}. Pick an $\epsilon \in \Rstarplus$ and introduce $\hat{\theta} \coloneqq (\theta_{1}^{\epsilon}, \theta_{2}^{\epsilon}, \theta_{3}^{\epsilon},\theta_{4}^{\epsilon})$ specified by:
$$
\theta_{1}^{\epsilon} \coloneqq \theta + (e_i + e_j) \epsilon, \quad \theta_{2}^{\epsilon} \coloneqq \theta + e_i \epsilon, \quad \theta_{3}^{\epsilon} \coloneqq \theta + e_j \epsilon, \quad \theta_{4}^{\epsilon} \coloneqq \theta.
$$

For each $k \in \enumreac$ and $q \in [\![1, 4]\!]$, introduce $\lambda_{k,q} : \Natural^N \xrightarrow{} \Rplus$ defined as:
$$
\lambda_{k,q}(x) \coloneqq \lambda_{k}(x, \theta_{q}^{\epsilon}).
$$

For $\hat{x}=(x_{1},x_{2},x_{3},x_{4})$ with $x_{q} \in \Natural^N$, we then introduce:
\begin{align*}
\Lambda_{k,(1,1,1,1)}(\hat{x}) &\coloneqq \lambda_{k,1}(x_{1}) \wedge \lambda_{k,2}(x_{2}) \wedge \lambda_{k,3}(x_{3}) \wedge \lambda_{k,4}(x_{4}), \\
\Lambda_{k,(1,1,0,0)}(\hat{x}) &\coloneqq \lambda_{k,1}(x_{1}) \wedge \lambda_{k,2}(x_{2}) - \Lambda_{k,(1,1,1,1)}(\hat{x}), \\
\Lambda_{k,(0,0,1,1)}(\hat{x}) &\coloneqq \lambda_{k,3}(x_{3}) \wedge \lambda_{k,4}(x_{4}) - \Lambda_{k,(1,1,1,1)}(\hat{x}), \\
\Lambda_{k,(1,0,1,0)}(\hat{x}) &\coloneqq (\lambda_{k,1}(x_{1}) - \lambda_{k,1}(x_{1}) \wedge \lambda_{k,2}(x_{2})) \wedge (\lambda_{k,3}(x_{3}) - \lambda_{k,3}(x_{3}) \wedge \lambda_{k,4}(x_{4})), \\
\Lambda_{k,(0,1,0,1)}(\hat{x}) &\coloneqq (\lambda_{k,2}(x_{2}) - \lambda_{k,1}(x_{1}) \wedge \lambda_{k,2}(x_{2})) \wedge (\lambda_{k,4}(x_{4}) - \lambda_{k,3}(x_{3}) \wedge \lambda_{k,4}(x_{4})), \\
\Lambda_{k,(1,0,0,0)}(\hat{x}) &\coloneqq (\lambda_{k,1}(x_{1}) - \lambda_{k,1}(x_{1}) \wedge \lambda_{k,2}(x_{2})) - \Lambda_{k,(1,0,1,0)}(\hat{x}), \\
\Lambda_{k,(0,1,0,0)}(\hat{x}) &\coloneqq (\lambda_{k,2}(x_{2}) - \lambda_{k,1}(x_{1}) \wedge \lambda_{k,2}(x_{2})) - \Lambda_{k,(0,1,0,1)}(\hat{x}), \\
\Lambda_{k,(0,0,1,0)}(\hat{x}) &\coloneqq (\lambda_{k,3}(x_{3}) - \lambda_{k,3}(x_{3}) \wedge \lambda_{k,4}(x_{4})) - \Lambda_{k,(1,0,1,0)}(\hat{x}), \\
\Lambda_{k,(0,0,0,1)}(\hat{x}) &\coloneqq (\lambda_{k,4}(x_{4}) - \lambda_{k,3}(x_{3}) \wedge \lambda_{k,4}(x_{4})) - \Lambda_{k,(0,1,0,1)}(\hat{x}).
\end{align*}

Introduce a collection of independent, unit-rate Poisson processes $\{(Y_{k, b}(t))_{t \in \mathbb{R}_{+}}\}_{k \in [\![1, M]\!], b\in \{0,1\}^4}$. Let $(X_{\hat{\theta}}^{(1)}(t))_{t \in \mathbb{R}_{+}}$, $(X_{\hat{\theta}}^{(2)}(t))_{t \in \mathbb{R}_{+}}$, $(X_{\hat{\theta}}^{(3)}(t))_{t \in \mathbb{R}_{+}}$ and $(X_{\hat{\theta}}^{(4)}(t))_{t \in \mathbb{R}_{+}}$ be given by the following time change representations:
\begin{align*}
    X_{\hat{\theta}}^{(1)}(t) &\coloneqq x_{0} + \sum_{k=1}^{M} \zeta_{k} \Big(R_{k,(1,1,1,1)}(t) + R_{k,(1,1,0,0)}(t) + R_{k,(1,0,1,0)}(t) + R_{k,(1,0,0,0)}(t)\Big), \\
    X_{\hat{\theta}}^{(2)}(t) &\coloneqq x_{0} + \sum_{k=1}^{M} \zeta_{k} \Big(R_{k,(1,1,1,1)}(t) + R_{k,(1,1,0,0)}(t) + R_{k,(0,1,0,1)}(t) + R_{k,(0,1,0,0)}(t)\Big), \\
    X_{\hat{\theta}}^{(3)}(t) &\coloneqq x_{0} + \sum_{k=1}^{M} \zeta_{k} \Big(R_{k,(1,1,1,1)}(t) + R_{k,(0,0,1,1)}(t) + R_{k,(1,0,1,0)}(t) + R_{k,(0,0,1,0)}(t) \Big), \\
    X_{\hat{\theta}}^{(4)}(t) &\coloneqq x_{0} + \sum_{k=1}^{M} \zeta_{k} \Big(R_{k,(1,1,1,1)}(t) + R_{k,(0,0,1,1)}(t) + R_{k,(0,1,0,1)}(t) + R_{k,(0,0,0,1)}(t)\Big),
\end{align*}

where we define a family $\{(R_{k, b}(t))_{t \in \mathbb{R}_{+}}\}_{k \in [\![1, M]\!], b\in \{0,1\}^4}$ of counting processes as:
$$
R_{k,b}(t) \coloneqq Y_{k,b}\Big(\int_{0}^{t} \Lambda_{k,b}(X_{\hat{\theta}}^{(1)}(s),X_{\hat{\theta}}^{(2)}(s), X_{\hat{\theta}}^{(3)}(s),X_{\hat{\theta}}^{(4)}(s))\Big)ds.
$$

Introduce $X_{\hat{\theta}}(t) \coloneqq (X_{\hat{\theta}}^{(1)}(t), X_{\hat{\theta}}^{(2)}(t), X_{\hat{\theta}}^{(3)}(t), X_{\hat{\theta}}^{(4)}(t))$. Observe that the position of the non-zero elements in $b$ specifies the index of the processes in $(X_{\hat{\theta}}(t))$ being coupled by $(R_{k,b}(t))$, while the norm $|b|$ specifies the number of processes being coupled. For each $q \in [\![1, 4]\!]$, note that $(X_{\hat{\theta}}^{(q)}(t))$ is a Markov process with generator $\generator_{\theta_{q}^{\epsilon}}$ and with intensity $\lambda_{k,q}$ for each reaction $k \in \enumreac$. Since $\sum_{b \in \{0,1\}^4} \Lambda_{k,b} \rightarrow \lambda_k$ as $\epsilon \rightarrow 0$, the processes $(X_{\hat{\theta}}^{(1)}(t))$, $(X_{\hat{\theta}}^{(2)}(t))$,
$(X_{\hat{\theta}}^{(3)}(t))$, and $(X_{\hat{\theta}}^{(4)}(t))$ themselves converge almost surely to $(X_{\theta}(t))$ as $\epsilon \rightarrow 0$.\newline

\textbf{Step 1.} Let $\sigma_{l}$ be as previously the $l$-th jump time of the process $(X_{\theta}(t))$ for $l \in \Natural$ with $\sigma_{0}=0$ for convenience. Introduce the function $\lambda_0 : \Natural^N \times \mathbb{R}^p\xrightarrow{} \Rplus$ denoting the sum of propensities: $\lambda_0(x,\theta) \coloneqq \sum_{k=1}^{M} \lambda_{k}(x,\theta)$. Let us first prove that:
\begin{align}
\begin{split}
\label{sensitivity_first_expression}
s_{\theta}^{(i,j)}(f, T) &= \sum_{k=1}^{M} \sumoverjumps \Bigg[ \frac{\partial^2 \lambda_{k}}{\partial \theta_{i} \partial \theta_{j}}(X_{\theta}(\sigma_l), \theta) \Bigg(\int_{0}^{T-\sigma_l}\Delta_{\zeta_{k}}\Psi_{\theta}(X_{\theta}(\sigma_{l}), f, T-\sigma_{l}-s) e^{-\lambda_{0}(X_{\theta}(\sigma_l), \theta)s}ds\Bigg) \\
&+ \frac{\partial \lambda_{k}}{\partial \theta_{i}}(X_{\theta}(\sigma_l), \theta) \Bigg(\int_{0}^{T-\sigma_l} \Delta_{\zeta_{k}}S_{\theta}^{(j)}(X_{\theta}(\sigma_{l}), f, T-\sigma_{l}-s) e^{-\lambda_{0}(X_{\theta}(\sigma_l), \theta)s}ds\Bigg) \\
&+ \frac{\partial \lambda_{k}}{\partial \theta_{j}}(X_{\theta}(\sigma_l), \theta) \Bigg(\int_{0}^{T-\sigma_l} \Delta_{\zeta_{k}}S_{\theta}^{(i)}(X_{\theta}(\sigma_{l}), f, T-\sigma_{l}-s) e^{-\lambda_{0}(X_{\theta}(\sigma_l), \theta)s}ds\Bigg)\Bigg].
\end{split}
\end{align}

\textbf{Step 1.1.} To start with, recall that by definition of second-order derivatives:
\begin{equation}
\label{finite_difference}
S_{\theta}^{(i,j)}(x_{0},f,T) = \lim_{\epsilon \rightarrow 0} \mathbb{E}_{x_0} \Bigg[\frac{f\big(X_{\hat{\theta}}^{(1)}(T)\big)-f\big(X_{\hat{\theta}}^{(2)}(T)\big)-f\big(X_{\hat{\theta}}^{(3)}(T)\big)+f\big(X_{\hat{\theta}}^{(4)}(T)\big)}{\epsilon^{2}}\Bigg].
\end{equation}

For each $k \in \enumreac$ and $b \in \{0,1\}^4$, the first time each counting process $(R_{k,b}(t))$ fires is defined by:
$\tau_{k, b}^{\epsilon} \coloneqq \inf \Big\{t \in \Rplus:R_{k,b}(t) \geq 1 \Big\}$. Taking $d \in [\![1, 3]\!]$, we define in turn: $
\tau_{k, d}^{\epsilon} \coloneqq \min_{b \in \{0,1\}^4, |b|=d} \tau_{k, b}^{\epsilon}$, as well as: $\tau_{k}^{\epsilon} \coloneqq \min_{d \in [\![1, 3]\!]} \tau_{k, d}^{\epsilon}$. The first splitting time between the processes $(X_{\hat{\theta}}^{(1)}(t))$, $(X_{\hat{\theta}}^{(2)}(t))$, $(X_{\hat{\theta}}^{(3)}(t))$ and $(X_{\hat{\theta}}^{(4)}(t))$ is given by the stopping time $\tau^{\epsilon}$ defined by: $\tau^{\epsilon} \coloneqq \min_{k \in \enumreac} \tau_{k}^{\epsilon}$. Let $\sigma_{l}^{\epsilon}$ be the $l$-th jump time of the process $(\sum_{k} R_{k,(1,1,1,1)}(t))$ for $l \in \Natural$ with $\sigma_{0}^{\epsilon}=0$. Let us introduce:
\begin{equation}
\label{eq:def_roh_jump}
\rho_{l,k}^{\epsilon}(\theta) \coloneqq \mathbb{E}\bigg[\mathbb{1}_{\{\sigma_{l}^{\epsilon}<\tau^{\epsilon}<\sigma_{l+1}^{\epsilon}, \tau^{\epsilon}=\tau_{k}^{\epsilon}\}} \bigg(f(X_{\hat{\theta}}^{(1)}(T))-f(X_{\hat{\theta}}^{(2)}(T))-f(X_{\hat{\theta}}^{(3)}(T))+f(X_{\hat{\theta}}^{(4)}(T))\bigg)\bigg],
\end{equation}

With these notations at hand, eq.~\eqref{finite_difference} can be rewritten as:
\begin{equation}
\label{finite_difference_bis}
S_{\theta}^{(i,j)}(x_{0},f,T) = \lim_{\epsilon \rightarrow 0} \frac{1}{\epsilon^{2}} \sum_{l=0}^{\infty} \sum_{k=1}^{M} \rho_{l,k}^{\epsilon}(\theta).
\end{equation}

\textbf{Step 1.2.} Let us focus on the individual $\rho_{l,k}^{\epsilon}(\theta)$. Introduce $\{\mathcal{F}_{t}^{\epsilon}\}_{t \in \Rplus}$ the filtration generated by $(X_{\hat{\theta}}(t))$ and $E_{l,k}^{\epsilon}\coloneqq\{\tau^{\epsilon}<\sigma_{l+1}^{\epsilon}, \tau^{\epsilon}=\tau_{k}^{\epsilon}\}$. Using the tower property for conditional expectations, eq.~\eqref{eq:def_roh_jump} can be rewritten as:
\begin{align}
\rho_{l,k}^{\epsilon}(\theta) &= \mathbb{E}\bigg[\mathbb{1}_{\{\sigma_{l}^{\epsilon}<\tau^{\epsilon}\}}\mathbb{1}_{E_{l,k}^{\epsilon}} \mathbb{E}\bigg[f(X_{\hat{\theta}}^{(1)}(T))-f(X_{\hat{\theta}}^{(2)}(T))-f(X_{\hat{\theta}}^{(3)}(T))+f(X_{\hat{\theta}}^{(4)}(T)) \bigg| \mathcal{F}_{T \wedge \tau^{\epsilon}}^{\epsilon}\bigg]\bigg] \nonumber \\
&= \mathbb{E}\bigg[\mathbb{1}_{\{\sigma_{l}^{\epsilon}<\tau^{\epsilon}\}}\mathbb{1}_{E_{l,k}^{\epsilon}}\mathbb{1}_{\tau_{k}^{\epsilon}=\tau_{k,1}^{\epsilon}} \mathbb{E}\bigg[f(X_{\hat{\theta}}^{(1)}(T))-f(X_{\hat{\theta}}^{(2)}(T))-f(X_{\hat{\theta}}^{(3)}(T))+f(X_{\hat{\theta}}^{(4)}(T))\bigg| \mathcal{F}_{T \wedge \tau^{\epsilon}}^{\epsilon}\bigg] \bigg]\nonumber\\
&+ \mathbb{E}\bigg[\mathbb{1}_{\{\sigma_{l}^{\epsilon}<\tau^{\epsilon}\}}\mathbb{1}_{E_{l,k}^{\epsilon}} \mathbb{1}_{\{\tau_{k}^{\epsilon}=\tau_{k,(1,1,0,0)}^{\epsilon}\}\cup\{\tau_{k}^{\epsilon}=\tau_{k,(0,0,1,1)}^{\epsilon}\}} \mathbb{E}\bigg[f(X_{\hat{\theta}}^{(1)}(T))-f(X_{\hat{\theta}}^{(2)}(T))-f(X_{\hat{\theta}}^{(3)}(T))+f(X_{\hat{\theta}}^{(4)}(T))\bigg| \mathcal{F}_{T \wedge \tau^{\epsilon}}^{\epsilon}\bigg]\bigg] \nonumber \\
&+ \mathbb{E}\bigg[\mathbb{1}_{\{\sigma_{l}^{\epsilon}<\tau^{\epsilon}\}}\mathbb{1}_{E_{l,k}^{\epsilon}} \mathbb{1}_{\{\tau_{k}^{\epsilon}=\tau_{k,(1,0,1,0)}^{\epsilon}\}\cup\{\tau_{k}^{\epsilon}=\tau_{k,(0,1,0,1)}^{\epsilon}\}}  \mathbb{E}\bigg[f(X_{\hat{\theta}}^{(1)}(T))-f(X_{\hat{\theta}}^{(2)}(T))-f(X_{\hat{\theta}}^{(3)}(T))+f(X_{\hat{\theta}}^{(4)}(T))\bigg| \mathcal{F}_{T \wedge \tau^{\epsilon}}^{\epsilon}\bigg]\bigg] \label{roh_jump_discunted}\\
&= \rho_{l,k,1}^{\epsilon}(\theta) + \rho_{l,k,2}^{\epsilon}(\theta), \nonumber
\end{align}

where we have defined $\rho_{l,k,1}^{\epsilon}(\theta)$ as the first expectation and $\rho_{l,k,2}^{\epsilon}(\theta)$ as the sum of the second and third expectations in eq.~\eqref{roh_jump_discunted}.\newline

\textbf{Step 1.3.} Let us start with $\rho_{l,k,1}^{\epsilon}(\theta)$. Using the strong Markov property and observing that $X_{\hat{\theta}}^{(1)}(\sigma_{l}^{\epsilon}) = X_{\hat{\theta}}^{(2)}(\sigma_{l}^{\epsilon}) = X_{\hat{\theta}}^{(3)}(\sigma_{l}^{\epsilon}) = X_{\hat{\theta}}^{(4)}(\sigma_{l}^{\epsilon})$ on $\{\sigma_{l}^{\epsilon}<\tau^{\epsilon}\}$, we know that:
\begin{multline*}
\mathbb{1}_{E_{l,k}^{\epsilon}} \mathbb{1}_{\{\tau_{k}^{\epsilon}=\tau_{k,1}^{\epsilon}\}} \mathbb{E}\bigg[f(X_{\hat{\theta}}^{(1)}(T))-f(X_{\hat{\theta}}^{(2)}(T))-f(X_{\hat{\theta}}^{(3)}(T))+f(X_{\hat{\theta}}^{(4)}(T)) \bigg| \mathcal{F}_{T \wedge \tau^{\epsilon}}^{\epsilon}\bigg]\\
= \mathbb{E}\bigg[\mathbb{1}_{E_{l,k}^{\epsilon}}\mathbb{1}_{\{\tau_{k}^{\epsilon}=\tau_{k,1}^{\epsilon}\}} \mathbb{E}\bigg[f(X_{\hat{\theta}}^{(1)}(T))-f(X_{\hat{\theta}}^{(2)}(T))-f(X_{\hat{\theta}}^{(3)}(T))+f(X_{\hat{\theta}}^{(4)}(T)) \bigg| \mathcal{F}_{T \wedge \tau^{\epsilon}}^{\epsilon}\bigg]\bigg| X_{\hat{\theta}}^{(4)}(\sigma_{l}^{\epsilon})\bigg]
\text{ almost surely on } \{\sigma_{l}^{\epsilon}<\tau^{\epsilon}\}.
\end{multline*}

Let $\gamma_{l}^{\epsilon} \coloneqq (\tau^{\epsilon}-\sigma_{l}^{\epsilon}) \wedge (\sigma_{l+1}^{\epsilon} - \sigma_{l}^{\epsilon})$. On the event $\{\sigma_{l}^{\epsilon}<\tau^{\epsilon}\}$, given $X_{\hat{\theta}}^{(4)}(\sigma_{l}^{\epsilon})=x$, $\gamma_{l}^{\epsilon}$ is exponentially distributed as the minimum of the exponentially distributed random variables $(\tau^{\epsilon}-\sigma_{l}^{\epsilon})$ and $(\sigma_{l+1}^{\epsilon} - \sigma_{l}^{\epsilon})$. Its rate is given by: 
$$
r(x,\hat{\theta},\epsilon) = \Lambda_{k,(1,1,1,1)}(x,x,x,x) + \Lambda_{k,(1,1,0,0)}(x,x,x,x) + \Lambda_{k,(0,0,1,1)}(x,x,x,x) + \Lambda_{k,(1,0,1,0)}(x) + \Lambda_{k,(0,1,0,1)}(x,x,x,x) + n_{k}(x, \hat{\theta}),
$$

where $n_{k}(x, \hat{\theta}) \coloneqq \lambda_{k,1}(x)+\lambda_{k,2}(x)-2(\lambda_{k,1}(x)\wedge\lambda_{k,2}(x))-2\Lambda_{k,(1,0,1,0)}(x,x,x,x)+\lambda_{k,3}(x)+\lambda_{k,4}(x)-2(\lambda_{k,3}(x)\wedge\lambda_{k,4}(x))-2\Lambda_{k,(0,1,0,1)}(x,x,x,x)$. Observe that $E_{l,k}^{\epsilon} = \{\tau^{\epsilon}=\sigma_{l}^{\epsilon}+\gamma_{l}^{\epsilon}, \tau^{\epsilon}=\tau_{k}^{\epsilon}\}$. On the event $\{\sigma_{l}^{\epsilon}<\tau^{\epsilon}\}$, given $X_{\hat{\theta}}^{(4)}(\sigma_{l}^{\epsilon})=x$, the probability $p_{l,k}^{\epsilon}(x,\hat{\theta})$ of the event $E_{l,k}^{\epsilon}\cap\{\tau_{k}^{\epsilon}=\tau_{k,1}^{\epsilon}\}$ is given by:
$$
p_{l,k}^{\epsilon}(x,\hat{\theta}) = \frac{n_{k}(x,\hat{\theta})}{r(x,\hat{\theta},\epsilon)}.
$$

Using assumption~\ref{assumption:propensity_regularity}, we can write the Taylor expansions of $\lambda_{k}(x,\cdot)$ and $\partial_{j}\lambda_{k}(x,\cdot)$ around $\theta_{2}^{\epsilon}$ and $\theta_{4}^{\epsilon}$:
$$
\lambda_{k}(x, \theta_{1}^{\epsilon})= \lambda_{k}(x, \theta_{2}^{\epsilon}) + \frac{\partial \lambda_{k}}{\partial \theta_{j}}(x, \theta_{2}^{\epsilon})+o(\epsilon), \quad \lambda_{k}(x, \theta_{3}^{\epsilon}) = \lambda_{k}(x, \theta_{4}^{\epsilon}) + \frac{\partial \lambda_{k}}{\partial \theta_{j}}(x, \theta_{4}^{\epsilon})+o(\epsilon), \quad \frac{\partial \lambda_{k}}{\partial\theta_{j}}(x, \theta_{2}^{\epsilon}) =  \frac{\partial \lambda_{k}}{\partial\theta_{j}}(x, \theta_{4}^{\epsilon}) + \frac{\partial^{2} \lambda_{k}}{\partial \theta_{i}\partial\theta_{j}}(x, \theta_{4}^{\epsilon})+o(\epsilon).
$$

This means we can pick an $\epsilon$ such that:
$$
\lambda_{k,1}(x) > \lambda_{k,2}(x) \iff \partial_{j}\lambda_{k,2}(x)>0, \quad \lambda_{k,3}(x) > \lambda_{k,4}(x) \iff \partial_{j}\lambda_{k,4}(x)>0, \quad \partial_{j} \lambda_{k,2}(x) > \partial_{j} \lambda_{k,4}(x) \iff \partial_{i,j}^{2} \lambda_{k,4}(x)>0.
$$

Assumption~\ref{assumption:propensity_regularity} also guarantees we can pick an $\epsilon$ such that $\partial_{j}\lambda_{k,2}(x)$ and $\partial_{j}\lambda_{k,4}(x)$ have same sign, \emph{e.g.}: 
$$
\partial_{j}\lambda_{k,2}(x) > 0 \iff \partial_{j}\lambda_{k,4}(x) > 0.
$$

On the event $\{\sigma_{l}^{\epsilon}<\tau^{\epsilon}<T\}\cap E_{l,k}^{\epsilon}\cap\{\tau_{k}^{\epsilon}=\tau_{k,1}^{\epsilon}\}$, given  $X_{\hat{\theta}}^{(4)}(\sigma_{l}^{\epsilon})=x$:
\begin{equation}
\label{eq:state_at_splitting_time}
X_{\hat{\theta}}(\tau^{\epsilon})= \left\{
\begin{array}{rr}
(x+\zeta_{k},x,x,x), & \text{if } \partial_{j} \lambda_{k,4}(x)>0 \text{ and } \partial_{i,j}^{2} \lambda_{k,4}(x)>0\\
(x,x,x+\zeta_{k},x), & \text{if } \partial_{j} \lambda_{k,4}(x)>0 \text{ and } \partial_{i,j}^{2} \lambda_{k,4}(x)<0\\
(x,x+\zeta_{k},x,x), & \text{if } \partial_{j} \lambda_{k,4}(x)<0 \text{ and } \partial_{i,j}^{2} \lambda_{k,4}(x)>0\\
(x,x,x,x+\zeta_{k}), & \text{otherwise}.
\end{array}\right.    
\end{equation}

Indeed, for example when $\partial_{j} \lambda_{k,4}(x)>0$ and $\partial_{i,j}^{2} \lambda_{k,4}(x)>0$:
\begin{align*}
\Lambda_{k,(1,0,0,0)}(x,x,x,x) &= (\lambda_{k,1}(x) - \lambda_{k,1} (x)\wedge \lambda_{k,2}(x)) - (\lambda_{k,1}(x) - \lambda_{k,1}(x) \wedge \lambda_{k,2}(x)) \wedge (\lambda_{k,3}(x) - \lambda_{k,3}(x) \wedge \lambda_{k,4}(x)) > 0 \\
\Lambda_{k,(0,1,0,0)}(x,x,x,x) &= (\lambda_{k,2}(x) - \lambda_{k,1}(x) \wedge \lambda_{k,2}(x)) - (\lambda_{k,2}(x) - \lambda_{k,1}(x) \wedge \lambda_{k,2}(x)) \wedge (\lambda_{k,4}(x) - \lambda_{k,3}(x) \wedge \lambda_{k,4}(x)) = 0\\
\Lambda_{k,(0,0,1,0)}(x,x,x,x) &= (\lambda_{k,3}(x) - \lambda_{k,1}(x) \wedge \lambda_{k,2}(x)) -  (\lambda_{k,1}(x) - \lambda_{k,1}(x) \wedge \lambda_{k,2}(x)) \wedge (\lambda_{k,3}(x) - \lambda_{k,3}(x) \wedge \lambda_{k,4}(x)) = 0\\
\Lambda_{k,(0,0,0,1)}(x,x,x,x) &= (\lambda_{k,4}(x) - \lambda_{k,1}(x) \wedge \lambda_{k,2}(x)) - (\lambda_{k,2}(x) - \lambda_{k,1}(x) \wedge \lambda_{k,2}(x)) \wedge (\lambda_{k,4}(x) - \lambda_{k,3}(x) \wedge \lambda_{k,4}(x)) = 0.    
\end{align*}

Let us consider a specific $x \in \Natural^{N}$ and assume for now that $\partial_{j} \lambda_{k}(x,\theta_{4}^{\epsilon})>0$ and $\partial_{i,j}^{2} \lambda_{k}(x,\theta_{4}^{\epsilon})>0$. Using the distribution of $\gamma_{l}^{\epsilon}$ and eq.~\eqref{eq:state_at_splitting_time}, we have:
\begin{multline*}
\mathbb{1}_{\{\sigma_{l}^{\epsilon}<\tau^{\epsilon}\}}\mathbb{E}\bigg[\mathbb{1}_{E_{l,k}^{\epsilon}} \mathbb{1}_{\{\tau_{k}^{\epsilon}=\tau_{k,1}^{\epsilon}\}}\mathbb{E}\bigg[f(X_{\hat{\theta}}^{(1)}(T))-f(X_{\hat{\theta}}^{(2)}(T))-f(X_{\hat{\theta}}^{(3)}(T))+f(X_{\hat{\theta}}^{(4)}(T))\bigg| \mathcal{F}_{T \wedge \tau^{\epsilon}}^{\epsilon}\bigg]\bigg| X_{\hat{\theta}}^{(4)}(\sigma_{l}^{\epsilon})=x\bigg] \\
= \mathbb{1}_{\{\sigma_{l}^{\epsilon}<\tau^{\epsilon}\}} p_{l,k}^{\epsilon}(x,\hat{\theta})r(x,\hat{\theta},\epsilon)\hat{R}_{\hat{\theta}}^{\epsilon}(x,f,T-(\sigma_{l}^{\epsilon} \wedge T),k),
\end{multline*}

where have introduced:
$$
\hat{R}_{\hat{\theta}}^{\epsilon}(x,f,t,k) \coloneqq \int_{0}^{t}\bigg(\Psi_{\theta_{1}^{\epsilon}}(x+\zeta_{k},f,t-s)-\Psi_{\theta_{2}^{\epsilon}}(x,f,t-s)-\Psi_{\theta_{3}^{\epsilon}}(x,f,t-s)+\Psi_{\theta_{4}^{\epsilon}}(x,f,t-s)\bigg)e^{-r(x,\hat{\theta},\epsilon)s}ds.
$$

Recall that as $\epsilon \rightarrow 0$, for each $q \in [\![1, 4]\!]$, $X_{\hat{\theta}}^{q}(t) \rightarrow X_{\theta}(t)$ almost surely. This implies $\sigma_{l}^{\epsilon} \rightarrow \sigma_{l}$ and $\tau^{\epsilon} \rightarrow \infty$ as $\epsilon \rightarrow 0$ almost surely, which means in turn that:
$$
\lim_{\epsilon \rightarrow 0} \mathbb{1}_{\{\sigma_{l}^{\epsilon}<\tau^{\epsilon}\}} \rightarrow 1 \quad \text{almost surely.}
$$

Now express the numerator of $ p_{l,k}^{\epsilon}(x,\hat{\theta})$ as: 
\begin{equation}
\label{eq:numerator_expansion}
n_{k}(x,\hat{\theta}) = \left\{
\begin{array}{rr}
\partial_{i,j}^{2} \lambda_{k,4}\epsilon^{2}+o(\epsilon^{2}), & \text{if } (\partial_{j} \lambda_{k,4}>0 \text{ and } \partial_{i,j}^{2} \lambda_{k,4}>0) \text{ or } (\partial_{j} \lambda_{k,4}<0 \text{ and } \partial_{i,j}^{2} \lambda_{k,4}<0)\\
-\partial_{i,j}^{2} \lambda_{k,4}\epsilon^{2}+o(\epsilon^{2}), & \text{otherwise.}\\
\end{array}\right.   
\end{equation}

Indeed, for example when $\partial_{j} \lambda_{k,4}>0$ and $\partial_{i,j}^{2} \lambda_{k,4}>0$:
$$
n_{k}(x,\hat{\theta}) 
= \bigg(\partial_{j} \lambda_{k}(x, \theta_{2}^{\epsilon}) - \partial_{j} \lambda_{k}(x, \theta_{4}^{\epsilon})+o(\epsilon)\bigg)\epsilon
= \partial_{ij}^{2} \lambda_{k}(x, \theta_{4}^{\epsilon})\epsilon^{2}+o(\epsilon^{2}).
$$

Observe that $r(X_{\hat{\theta}}^{(4)}(\sigma_{l}^{\epsilon},\theta),\hat{\theta},\epsilon) \rightarrow \lambda_{0}(X_{\hat{\theta}}^{(4)}(\sigma_{l}^{\epsilon},\theta))$ almost surely as $\epsilon \rightarrow 0$. Therefore using eq.~\eqref{eq:numerator_expansion} and the continuous mapping theorem, we have:
\begin{multline*}
\lim_{\epsilon \rightarrow 0} \frac{1}{\epsilon^{2}}p_{l,k}^{\epsilon}(X_{\hat{\theta}}^{(4)}(\sigma_{l}^{\epsilon}),\hat{\theta})r(x,\hat{\theta},\epsilon)\hat{R}_{\hat{\theta}}^{\epsilon}(X_{\hat{\theta}}^{(4)}(\sigma_{l}^{\epsilon}),f,T-(\sigma_{l}^{\epsilon} \wedge T),k) \\
= \partial_{ij}^{2} \lambda_{k}(X_{\theta}(\sigma_{l}), \theta)\int_{0}^{T-\sigma_l} \Delta_{\zeta_{k}}\Psi_{\theta}(X_{\theta}(\sigma_{l}), f, T-\sigma_{l}-s) e^{-\lambda_{0}(X_{\theta}(\sigma_l), \theta)s}ds  \quad \text{almost surely.}
\end{multline*}

The result is the same when we use eq.~\eqref{eq:state_at_splitting_time} and~\eqref{eq:numerator_expansion} starting from the other assumptions on $\partial_{j} \lambda_{k,4}$ and $\partial_{i,j}^{2} \lambda_{k,4}$ at $x$. Using the Portmanteau theorem for expectations~\autocite[theorem 2.1]{billingsley2013convergence}, this leads to:
$$
\lim_{\epsilon \rightarrow 0} \frac{1}{\epsilon^{2}} \rho_{l,k,1}^{\epsilon}(\theta) = \mathbb{E}\Bigg[ \frac{\partial^2 \lambda_{k}}{\partial \theta_{i} \partial \theta_{j}}(X_{\theta}(\sigma_l), \theta) \int_{0}^{T-\sigma_l} \Delta_{\zeta_{k}}\Psi_{\theta}(X_{\theta}(\sigma_{l}), f, T-\sigma_{l}-s)e^{-\lambda_{0}(X_{\theta}(\sigma_l), \theta)s}ds\Bigg].
$$

\textbf{Step 1.4.} Let us now turn to $\rho_{l,k,2}^{\epsilon}(\theta)$. On the event $\{\sigma_{l}^{\epsilon}<\tau^{\epsilon}\}$, given $X_{\hat{\theta}}^{(4)}(\sigma_{l}^{\epsilon})=x$, the probability $p_{l,k}^{\epsilon}(x,\hat{\theta})$ of the event $E_{l,k}^{\epsilon}\cap \{\{\tau_{k}^{\epsilon}=\tau_{k,(1,1,0,0)}^{\epsilon}\}\cup\{\tau_{k}^{\epsilon}=\tau_{k,(0,0,1,1)}^{\epsilon}\}\}$ is: $
p_{l,k}^{\epsilon}(x,\hat{\theta}) = n_{k}(x,\hat{\theta})/r(x,\hat{\theta},\epsilon)
$, where $n_{k}(x,\hat{\theta}) \coloneqq \lambda_{k,1}(x)\wedge\lambda_{k,2}(x)+\lambda_{k,3}(x)\wedge\lambda_{k,4}(x)-2\Lambda_{k,(1,1,1,1)}(x,x,x,x)$. Observe that: $
n_{k}(x,\hat{\theta}) = |\partial_{i}\lambda_{k,4}(x)|\epsilon + o(\epsilon)$. On the event $\{\sigma_{l}^{\epsilon}<\tau^{\epsilon}<T \}\cap E_{l,k}^{\epsilon}\cap \{\{\tau_{k}^{\epsilon}=\tau_{k,(1,1,0,0)}^{\epsilon}\}\cup\{\tau_{k}^{\epsilon}=\tau_{k,(0,0,1,1)}^{\epsilon}\}\}$, given $X_{\hat{\theta}}^{(4)}(\sigma_{l}^{\epsilon})=x$, we have:
$$
X_{\hat{\theta}}(\tau^{\epsilon}) = \left\{
\begin{array}{rr}
(x+\zeta_{k}, x+\zeta_{k}, x, x), & \text{if } \partial_{j} \lambda_{k,4}(x)>0\\
(x, x, x+\zeta_{k}, x+\zeta_{k}), & \text{otherwise.}\\
\end{array}\right.
$$

Similarly as before, assume for now that $\partial_{j} \lambda_{k, 4}(x)>0$ and introduce:
$$
\hat{R}_{\hat{\theta}}^{\epsilon}(x,f,t,k) = \int_{0}^{t}\bigg(\Psi_{\theta_{1}^{\epsilon}}(x+\zeta_{k},f,t-s)-\Psi_{\theta_{2}^{\epsilon}}(x+\zeta_{k},f,t-s)-\Psi_{\theta_{3}^{\epsilon}}(x,f,t-s)+\Psi_{\theta_{4}^{\epsilon}}(x,f,t-s)\bigg)e^{-r(x,\hat{\theta},\epsilon)s}ds.
$$

We have:
\begin{multline*}
\lim_{\epsilon \rightarrow 0} \frac{1}{\epsilon^{2}}p_{l,k}^{\epsilon}(X_{\hat{\theta}}^{(4)}(\sigma_{l}^{\epsilon}),\hat{\theta})r(x,\hat{\theta},\epsilon)\hat{R}_{\hat{\theta}}^{\epsilon}(X_{\hat{\theta}}^{(4)}(\sigma_{l}^{\epsilon}),f,T-(\sigma_{l}^{\epsilon} \wedge T),k) \\
= \lim_{\epsilon \rightarrow 0} \frac{\partial_{i}\lambda_{k,4}(X_{\hat{\theta}}^{(4)}(\sigma_{l}^{\epsilon}))\epsilon + o(\epsilon)}{\epsilon} \int_{0}^{T-(\sigma_{l}^{\epsilon} \wedge T)}\bigg(\frac{\partial_{j}\Psi_{\theta_{2}^{\epsilon}}(X_{\hat{\theta}}^{(4)}(\sigma_{l}^{\epsilon})+\zeta_{k},f,t-s)+ o(\epsilon)}{\epsilon}-\frac{\partial_{j}\Psi_{\theta_{4}^{\epsilon}}(X_{\hat{\theta}}^{(4)}(\sigma_{l}^{\epsilon}),f,t-s)+ o(\epsilon)}{\epsilon}\bigg) e^{-r(X_{\hat{\theta}}^{(4)}(\sigma_{l}^{\epsilon}),\hat{\theta},\epsilon)s}ds.
\end{multline*}

Proceeding similarly for the second expectation in $\rho_{l,k,2}^{\epsilon}(\theta)$ and following similar steps as for $\rho_{l,k,1}^{\epsilon}(\theta)$, we obtain:
\begin{align*}
\begin{split}
\lim_{\epsilon \rightarrow 0} \frac{1}{\epsilon^{2}}\rho_{l,k,2}^{\epsilon}(\theta) &= \mathbb{E}\Bigg[  \frac{\partial \lambda_{k}}{\partial \theta_{i}}(X_{\theta}(\sigma_l), \theta) \Bigg(\int_{0}^{T-\sigma_l} \Delta_{\zeta_{k}}S_{\theta}^{(j)}(X_{\theta}(\sigma_{l}), f, T-\sigma_{l}-s) e^{-\lambda_{0}(X_{\theta}(\sigma_l), \theta)s}ds\Bigg) \\
&+ \frac{\partial \lambda_{k}}{\partial \theta_{j}}(X_{\theta}(\sigma_l), \theta) \Bigg(\int_{0}^{T-\sigma_l} \Delta_{\zeta_{k}}S_{\theta}^{(i)}(X_{\theta}(\sigma_{l}), f, T-\sigma_{l}-s) e^{-\lambda_{0}(X_{\theta}(\sigma_l), \theta)s}ds\Bigg)\Bigg].
\end{split}
\end{align*}

\textbf{Step 2.} Now let us prove that we can transform eq.~\eqref{sensitivity_first_expression} to:
\begin{align}
\begin{split}
\label{eq:intermediate_for_theorem}
s_{\theta}^{(i,j)}
(f, T) &= \sum_{k=1}^{M} \Bigg[\int_{0}^{T} \frac{\partial^2 \lambda_{k}}{\partial \theta_{i} \partial \theta_{j}}(X_{\theta}(s), \theta) \Delta_{\zeta_k} \Psi_{\theta}(X_{\theta}(s), f, T-s)ds \\
&+ \int_{0}^{T} \frac{\partial \lambda_{k}}{\partial \theta_{i}}(X_{\theta}(s), \theta) \Delta_{\zeta_k} S_{\theta}^{(j)}(X_{\theta}(s), f, T-s) ds \\
&+ \int_{0}^{T} \frac{\partial \lambda_{k}}{\partial \theta_{j}}(X_{\theta}(s), \theta) \Delta_{\zeta_k} S_{\theta}^{(i)}(X_{\theta}(s), f, T-s) ds\Bigg].
\end{split}
\end{align}

\textbf{Step 2.1.} Let us start with the first term in the summation in eq.~\eqref{sensitivity_first_expression}. Using eq.~\eqref{eq:conditioned_delta} from lemma~\ref{lemma:for_second_part} with $g(s)\coloneqq\Delta_{\zeta_{k}}\Psi_{\theta}(y, f, T-u-s)$, observe that:
\begin{align*}
\int_{0}^{T-u}\Delta_{\zeta_{k}}\Psi_{\theta}(y, f, T-u-s)e^{-\lambda_{0}(y, \theta)s}ds &= \mathbb{E} \Bigg[\int_{0}^{\delta_{l}}\Delta_{\zeta_{k}}\Psi_{\theta}(y, f, T-u-s)ds\Bigg| X_{\theta}(\sigma_{l})=y,\sigma_{l}=u\Bigg]\\
&= \mathbb{E} \Bigg[\int_{u}^{\delta_{l}+u}\Delta_{\zeta_{k}}\Psi_{\theta}(y, f, T-s)ds\Bigg| X_{\theta}(\sigma_{l})=y,\sigma_{l}=u\Bigg]\\
&= \mathbb{E} \Bigg[\int_{u}^{\sigma_{l+1} \wedge T}\Delta_{\zeta_{k}}\Psi_{\theta}(y, f, T-s)ds\Bigg| X_{\theta}(\sigma_{l})=y,\sigma_{l}=u\Bigg],
\end{align*}

which means:
\begin{equation}
\label{eq:integral_term}
\int_{0}^{T-\sigma_{l}}\Delta_{\zeta_{k}}\Psi_{\theta}(X_{\theta}(\sigma_{l}), f, T-\sigma_{l}-s)e^{-\lambda_{0}(X_{\theta}(\sigma_{l}), \theta)s}ds = \mathbb{E} \Bigg[\int_{\sigma_{l}}^{\sigma_{l+1} \wedge T}\Delta_{\zeta_{k}}\Psi_{\theta}(X_{\theta}(\sigma_{l}), f, T-s)ds\Bigg| \mathcal{F}_{\sigma_{l}}, \sigma_{l}<T\Bigg].   \end{equation}

Substituting in eq.~\eqref{sensitivity_first_expression} the expression given in eq.~\eqref{eq:integral_term}, we get:
\begin{align*}
&\mathbb{E}\Bigg[\sumoverjumps \frac{\partial^2 \lambda_{k}}{\partial \theta_{i} \partial \theta_{j}}(X_{\theta}(\sigma_l), \theta) \Bigg(\int_{0}^{T-\sigma_l}\Delta_{\zeta_{k}}\Psi_{\theta}(X_{\theta}(\sigma_{l}), f, T-\sigma_{l}-s) e^{-\lambda_{0}(X_{\theta}(\sigma_l), \theta)s}ds\Bigg)\Bigg] \nonumber\\
&= \mathbb{E}\Bigg[\sumoverjumps \frac{\partial^2 \lambda_{k}}{\partial \theta_{i} \partial \theta_{j}}(X_{\theta}(\sigma_l), \theta) \mathbb{E} \Bigg[\int_{\sigma_{l}}^{\sigma_{l+1} \wedge T}\Delta_{\zeta_{k}}\Psi_{\theta}(X_{\theta}(\sigma_{l}), f, T-s)ds\Bigg| \mathcal{F}_{\sigma_{l}}, \sigma_{l}<T\Bigg] \Bigg]\\
&= \mathbb{E}\Bigg[\sumoverjumps \int_{\sigma_{l}}^{\sigma_{l+1} \wedge T}\frac{\partial^2 \lambda_{k}}{\partial \theta_{i} \partial \theta_{j}}(X_{\theta}(\sigma_l), \theta) \Delta_{\zeta_{k}}\Psi_{\theta}(X_{\theta}(\sigma_{l}), f, T-s)ds\Bigg]\\
&= \mathbb{E}\Bigg[\int_{0}^{T}\frac{\partial^2 \lambda_{k}}{\partial \theta_{i} \partial \theta_{j}}(X_{\theta}(s), \theta) \Delta_{\zeta_{k}}\Psi_{\theta}(X_{\theta}(s), f, T-s)ds\Bigg].
\end{align*}

\textbf{Step 2.2.} Let us now turn to the second and third terms in the summation in eq.~\eqref{sensitivity_first_expression}. Following similar steps and introducing $g(s)\coloneqq\Delta_{\zeta_{k}}S_{\theta}^{(j)}(y, f, T-u-s)$ or $g(s)\coloneqq\Delta_{\zeta_{k}}S_{\theta}^{(i)}(y, f, T-u-s)$, we get:
\begin{align*}
\mathbb{E}\Bigg[\sumoverjumps \frac{\partial \lambda_{k}}{\partial \theta_{i}}(X_{\theta}(\sigma_l), \theta) \Bigg(\int_{0}^{T-\sigma_l} \Delta_{\zeta_{k}}S_{\theta}^{(j)}(X_{\theta}(\sigma_{l}), f, T-\sigma_{l}-s) e^{-\lambda_{0}(X_{\theta}(\sigma_l), \theta)s}ds\Bigg) \Bigg] &= \mathbb{E}\Bigg[\int_{0}^{T} \frac{\partial \lambda_{k}}{\partial \theta_{i}}(X_{\theta}(s), \theta) \Delta_{\zeta_k} S_{\theta}^{(j)}(X_{\theta}(s), f, T-s) ds\Bigg], \\
\mathbb{E}\Bigg[ \sumoverjumps \frac{\partial \lambda_{k}}{\partial \theta_{j}}(X_{\theta}(\sigma_l), \theta) \Bigg(\int_{0}^{T-\sigma_l} \Delta_{\zeta_{k}}S_{\theta}^{(i)}(X_{\theta}(\sigma_{l}), f, T-\sigma_{l}-s) e^{-\lambda_{0}(X_{\theta}(\sigma_l), \theta)s}ds\Bigg)\Bigg] &= \mathbb{E}\Bigg[\int_{0}^{T} \frac{\partial \lambda_{k}}{\partial \theta_{j}}(X_{\theta}(s), \theta) \Delta_{\zeta_k} S_{\theta}^{(i)}(X_{\theta}(s), f, T-s) ds\Bigg].
\end{align*}

\end{proof}

\printbibliography

\end{document}